\pdfoutput=1
\documentclass[letterpaper]{article}
\usepackage{natbib}
\usepackage{amsmath, amsthm, amssymb, wasysym, verbatim, bbm, color, graphics, geometry, cancel}
\usepackage{setspace}
\usepackage{ mathrsfs }
\usepackage[ruled, vlined]{algorithm2e}
\usepackage{hyperref}
\usepackage{lastpage}
\usepackage{enumerate, subcaption}
\usepackage{fancyhdr}
\usepackage{xcolor}
\usepackage{graphicx, multirow}
\usepackage{listings}
\usepackage{ dsfont }
\usepackage{tikz}
\usepackage{booktabs}
\usepackage{longtable}
\usepackage{array}
\usepackage{wrapfig}
\usepackage{float}
\usepackage{colortbl}
\usepackage{pdflscape}
\usepackage{tabu}
\usepackage{threeparttable}
\usepackage{threeparttablex}
\usepackage[normalem]{ulem}
\usepackage{makecell}
\usepackage[shortlabels]{enumitem}
\usepackage[T1]{fontenc}

\DeclareMathOperator*{\pr}{\mathbb{P}}
\DeclareMathOperator*{\E}{\mathbb{E}}

\DeclareMathOperator*{\prob}{\overset{p}{\to}}
\DeclareMathOperator*{\dist}{\overset{\mathcal{D}}{\to}}

\DeclareMathOperator*{\iid}{\overset{\text{i.i.d.}}{\sim}}

\newcommand{\R}{\mathbb{R}}

\newcommand{\N}{\mathbb{N}}

\newtheorem{thm}{Theorem}[section]
\theoremstyle{definition}

\newtheorem{prop}[thm]{Proposition}
\newtheorem{lem}[thm]{Lemma}
\newtheorem{cor}[thm]{Corollary}

\theoremstyle{remark}

\SetKwInput{KwInput}{Input}                
\SetKwInput{KwOutput}{Output}              

\title{Tight Distribution-Free Confidence Intervals for Local Quantile Regression}
\author{Jayoon Jang\thanks{Department of Statistics, Stanford University} \and Emmanuel Cand\`es\thanks{Department of Statistics and Mathematics, Stanford University}}
\date{\today}
\begin{document}
	
	\maketitle
	
	\begin{abstract}
		It is well known that it is impossible to construct useful confidence intervals (CIs) about the mean or median of a response $Y$ conditional on features $X = x$ without making strong assumptions about the joint distribution of $X$ and $Y$. 
		This paper introduces a new framework for reasoning about problems of this kind by casting the conditional problem at different levels of resolution, ranging from coarse to fine localization. In each of these problems, we consider
		local quantiles defined as the marginal quantiles of $Y$ when $(X,Y)$
		is resampled in such a way that samples $X$ near $x$ are up-weighted
		while the conditional distribution $Y \mid X$ does not change. We then introduce
		the Weighted Quantile method, which
		asymptotically produces the uniformly most accurate confidence intervals for these local quantiles no matter the (unknown) underlying
		distribution. Another method, namely, the Quantile Rejection method, achieves finite sample validity under no assumption whatsoever.  We conduct
		extensive numerical studies demonstrating that both of these methods are
		valid. In particular, we show that the Weighted Quantile procedure achieves nominal
		coverage as soon as the effective sample size is in the range of 10 to 20.
	\end{abstract}
	
	\section{Introduction} 
	\subsection{Motivation} \label{subsec:motivation}
	Consider the following real dataset from \cite{efron1991compliance}: 164 male doctors were treated with cholestryamine, a drug known to decrease cholesterol level. In addition to the cholesterol level, the compliance of each subject, defined as the proportion of the drug actually taken from the intended dose, is also available. Figure \ref{fig:compliance} shows a scatter plot of decrease in cholesterol versus compliance. 
	
	\begin{figure}[h]
		\centering
		\includegraphics[width=0.8\linewidth]{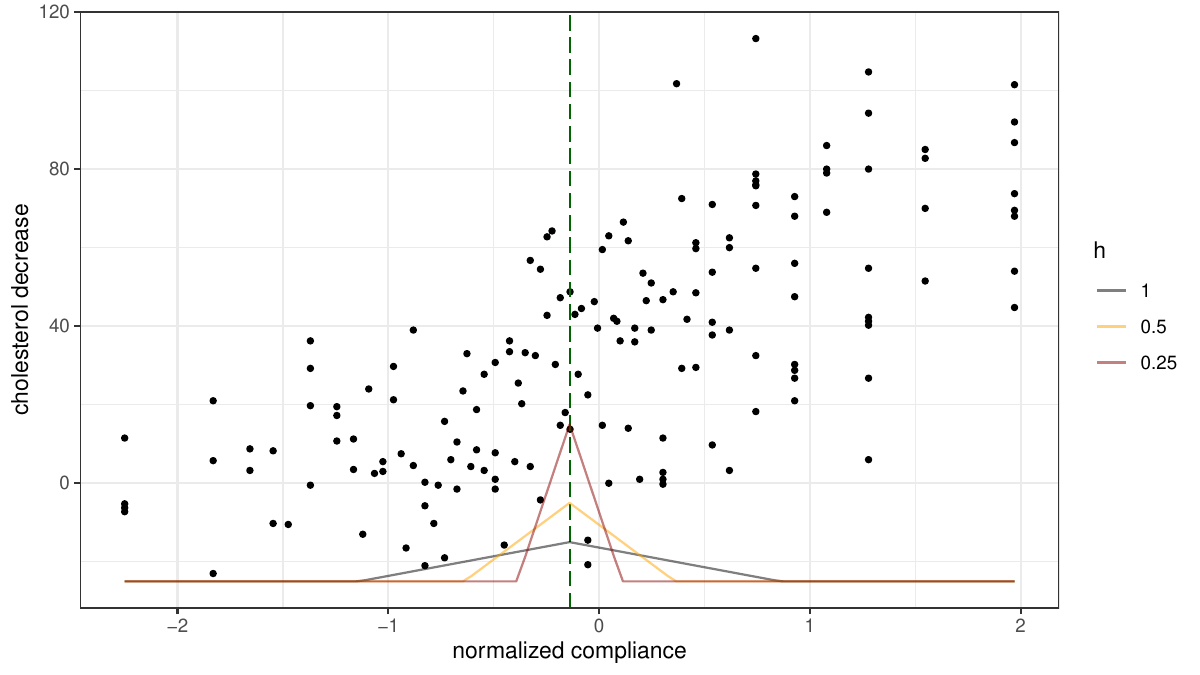}
		\caption{Scatter plot of cholestryamine data and weighting kernels with different values of the bandwidth $h$.}
		\label{fig:compliance}
	\end{figure}
	
	To understand the relationship between a set of predictors $X$ and an outcome $Y$, statisticians often report inferences about the conditional mean $\E[Y|X=x]$ or the conditional median $\text{Med}[Y|X=x]$. Inference problems of these kinds are traditionally modeled by using a surface plus noise model $Y = f(X) + \text{noise}$. To provide inference based on the fitted model, we usually require assumptions such as smoothness of $f$ and some structure on the noise such as sub-Gaussianity. These assumptions can be problematic since they are \textit{empirically unverifiable} as pointed out in \cite{Donoho1988}.
	Moreover, it is unclear what we even mean by the conditional distribution of $Y$ at $X$ exactly being $x$ since it is hard to understand if there actually exists a population of people having a level of compliance $x$ equal to 50\% whatever the level of precision.

	Keeping this in mind, we suggest a shift in emphasis from these classical objects of inference. Our main idea is instead to consider the distribution of $Y$ when $X$ is \textit{near a point of interest}: more specifically, we aim to determine the quantile of $Y$ when $X$ is near $x_0$.
	
	Imagine a researcher comes up with a kernel $K$ and a bandwidth $h$ that models the weights and degree of nearness to the point $x_0$. This yields a reweighted covariate distribution $Q_X^{x_0}$ defined as \begin{align*}
		Q_X^{x_0} \propto K\left( \frac{x_0-\cdot}{h}\right)P_X,
	\end{align*} where $P_X$ is the original covariate distribution. If the kernel is non-increasing with respect to the distance from 0, $Q_X^{x_0}$ gives more weight to samples with covariates close to $x_0$ and less weight to samples that are farther away from $x_0$. We take our new object of inference to be \begin{align*}
		\theta = \text{Med}(Y; (X,Y) \sim Q_X^{x_0} \times P_{Y|X}),
	\end{align*} i.e., the median of the marginal distribution of $Y$ where $X$ is localized to the space near $x_0$. 
	
	Going back to the compliance example, say we are using a triangular kernel of bandwidth $h = 0.5$ for reweighting. Then, our object of inference at compliance $50\%$ is the median of the decrease in cholesterol level for people who comply between $49.5\%$ and $50.5\%$ with more weight towards $50\%$ compliance; this value comes from a more interpretable population than one in which some people have a compliance exactly equal to $50\%$. Additionally, by adjusting $h$, we have the freedom to change the resolution at which we localize the covariate space. If we take the kernel bandwidth to be small, we are smoothing the outcome over a small neighborhood. Figure \ref{fig:compliance} plots the triangular kernel centered at $x_0 = 50\%$ (corresponding to -0.14 for normalized compliance) as the bandwidth $h$ varies.
	
	By rethinking the object of inference, in a way that considers the resolution of the observation of the underlying covariate distribution, we are able to provide reliable inference without making any modeling assumptions.
	This is in stark contrast to classical nonparametric methods that impose unverifiable assumptions on the underlying distribution to offer confidence intervals for the conditional mean or quantile. Additionally, these methods often produce confidence intervals that depend on unknown constants. We provide a discussion on nonparametric methods in Section \ref{subsec: past works}.

	\subsection{Problem statement}
	Consider i.i.d.~data $Z_1, \ldots, Z_n \sim P$, where each $Z_i = (X_i, Y_i)$ is a random vector in $\R^d \times \R$, $Y_i$ is a response variable of interest and $X_i$ is a $d$-dimensional vector of covariates. Suppose $P$ has an unknown density $p(x, y) = p(y|x)p(x)$. \footnote{$p$ can be the Radon-Nikodym derivative with respect to any arbitrary measure, but we will use the term ``density" for simplicity.}
	We are interested in the local distribution of the response variable $Y$ at a point $x_0 \in \R^d$. Localization is done by reweighting the distribution of the covariates with some uniformly bounded kernel $K$ that is non-negative and obeys $\int_{\R^d} K(u) du = 1$. Define the covariate-shifted distribution (depending on $K, x_0, h$) as 
	
	\begin{align*}
		\dfrac{dQ}{dP}(x,y) \propto K\left( \frac{x_0-x}{h}\right).
	\end{align*} 
	
	We are interested in making claims about the quantiles of $Y$ with respect to the marginal distribution $Q_Y$ of $Y$, that is, the distribution of $Y$ when $(X,Y) \sim  Q$. 
	
	In the interest of lightening the notation, the distributions defined above do not explicitly show the dependence on $p$, $K$ $x_0$ and $h$.\footnote{The covariate-shifted joint probability distribution function will be denoted as $q(x, y)$, and the marginal probability distribution function and the cumulative distribution function will be denoted as $q_Y(y)$ and $Q_Y(y)$, respectively.}
	
	Our new object of inference is 
	\begin{align}
		\theta_{p} =  Q_{Y}^{-1}(p)= \inf\{y: Q_Y(y) \ge p\},\label{eq:object} 
	\end{align}
	which is the $p$-th quantile of the distribution $Q_{Y}$. This parameter has an intuitive  interpretation; as we shall see, it can also be reliably inferred without making any unverifiable assumptions on the underlying distribution.

	\subsection{Summary of key results}
	
	This paper builds distribution-free confidence intervals for $\theta_p$ that have length vanishing to 0 as the sample size goes to infinity. 
	The first method we propose is the Weighted Quantile method (Algorithm \ref{alg:wq}) introduced in Section \ref{sec:wq}. This method achieves asymptotically valid coverage and provides an optimal confidence interval in the sense that it is asymptotically uniformly most accurrate unbiased.
	The second is the Quantile Rejection method (Algorithm \ref{alg:rjq}). The method uses rejection sampling to achieve finite sample coverage. 
	
	In Section \ref{sec:Simulation},
	we provide various numerical experiments, which empirically confirm that both of our proposed methods have valid coverage. 
	In Section \ref{sec:realdata}, we apply our proposed methods to the compliance example from Section \ref{subsec:motivation} and a California housing dataset. 
	Finally, we show in Section \ref{Sec:indistinguishable} that it is impossible to design shorter intervals.
	
	The two methods proposed in the paper are implemented in \url{https://github.com/Jayoon/resolution_paper}. The R code for reproducing the experiments and analysis is also available in the same repository.

	\subsection{Past/related works} \label{subsec: past works}
	
	\paragraph{Distribution-free inference}
	Much interest has been devoted to constructing distribution-free \textit{prediction intervals} for $Y_{n+1}$ given exchangeable training data $\{(X_1, Y_1), \ldots, (X_n, Y_n)\}$. Often, this is done using techniques known as conformal inference \citep{vovk2005algorithmic}. Conformal prediction intervals achieve distribution-free marginal coverage $\mathbb{P}(Y_{n+1} \in \hat{C}_n(X_{n+1})) \ge 1-\alpha$. Also, it is known to be impossible to construct a bounded prediction interval that satisfies a stronger notion of conditional coverage $\mathbb{P}(Y_{n+1} \in \hat{C}_n(X_{n+1})|X_{n+1}=x) \ge 1-\alpha$ as shown in \cite{vovk2012conditional, LeiWasserman2014}. Prediction intervals differ from confidence intervals in that they have to capture the inherent randomness of $Y$ and so cannot have width converging to 0 even as the sample size goes to infinity. For non-exchangeable data, \cite{tibshirani2019conformal} provides a method for computing a valid prediction interval when the test set is a covariate-shifted distribution of the training set. While our paper also considers a covariate-shifted distribution as a target, the difference lies in that \cite{tibshirani2019conformal} constructs a prediction interval for $Y_{n+1}$, while our paper constructs a confidence interval for a fixed parameter.

	Distribution-free confidence intervals that are finite-sample valid for parameters such as the mean $\mathbb{E}[Y]$, quantile $\text{Quantile}_p[Y]$, conditional mean $\mathbb{E}[Y|X]$, and conditional quantiles $\text{Quantile}_p[Y|X]$ have been extensively studied. A classical result of \cite{bahadur1956nonexistence} shows that it is impossible to get a bounded confidence interval for the mean without restricting the distribution class. The basic idea is that if the distribution class is too large, there always exist two almost indistinguishable distributions that have an arbitrarily large mean difference. Fortunately, constructing finite-sample confidence interval for quantiles is possible since the probability of a quantile being between any observations can be bounded in terms of binomial probabilities \citep{noether1972distribution}. However, it is impossible to construct a non-trivial distribution-free confidence interval for the {\em conditional} quantiles.
	\cite{Barber2020} showed that in a binary regression setting where $Y \in \{0,1\}$ with nonatomic joint distribution $P$ (meaning a distribution without point masses), the expected length of any confidence interval for the conditional label probability $\pr_P[Y=1|X]$ has a non-vanishing lower bound.  Similarly, \cite{MedarametlaCandes2021} show that there does not exist any distribution-free confidence interval for the conditional median that has vanishing width as the sample size goes to infinity. Note that both results have assumed the distribution is nonatomic. 
	When this assumption is removed, \cite{lee2021distribution} characterize a regime where non-trivial distribution-free inference for $\mathbb{E}[Y|X]$ is possible in the case where the response variable $Y$ is bounded. They introduce the concept of \textit{effective support size} of $P_X$ and show that if it is smaller than the square of the sample size, one can construct a confidence interval of vanishing length.

	\paragraph{Nonparametric statistics} 
	In nonparametric inference, the goal is to make as few assumptions on the model as possible. This is usually done by assuming only that the regression function $\mathbb{E}[Y|X] = f(X)$ belongs to a restricted function class and the covariate distribution is continuous. The function class usually characterizes the smoothness of $f$ by positing the existence of a certain number of bounded derivatives. Examples of widely used function classes are H\"older classes and Besov classes \citep{Tybakov2008}. Such assumptions are crucial to make the estimators of $f(X)$ have desirable properties such as consistency or optimality. Moreover, in order to obtain uniform confidence intervals over a function class for a linear functional of the regression function such as $f(x_0)$, assuming knowledge about the smoothness constant cannot be avoided \citep{Low1997, armstrong2018optimal}. 
	The hardness of constructing nonparametric confidence intervals comes from the difficulty of measuring the bias of the regression function estimate. For instance, smoothing methods using kernels or local polynomials require a choice of a bandwidth $h$, which corresponds to the amount of smoothing. Theoretically, based on the assumed smoothness, the bandwidth is chosen to set the rate of convergence of bias and variance to be the same. In practice, we can obtain a data-driven bandwidth by minimizing some criterion such as a generalized cross-validation error or use a plug-in bandwidth, which is obtained by replacing unknown components of the optimal bandwidth with estimates.
	For more literature on nonparametric inference, see for example, \cite{wasserman2006all, gine2021mathematical} and references therein.

	Instead of choosing a single optimal smoothing parameter, \cite{chaudhuri2000scale} consider a family of smooth curve estimates by varying the smoothing parameter or the bandwidth $h$. This has some similarity to our work in the sense that they turn their attention to $\mathbb{E}[\hat{f}_h(x)]$, which is a smoothed version of the regression function $f(x)$. However, their inference focuses only on identifying statistically significant local maximizers and minimizers of $\mathbb{E}[\hat{f}_h(x)]$.

	\section{Weighted Quantile method} \label{sec:wq}
	We now outline our first of two methods.
	First, we introduce the method in Section \ref{subsec:wq}. Then we show that the proposed confidence interval is efficient in Section \ref{subsec:optimality}
	We provide proofs for all the results from this section in Appendix \ref{app:proofs}.
	
	\subsection{Method} \label{subsec:wq}
	
	Suppose we have i.i.d.~samples $Y_1, \ldots Y_n$ from some distribution $P$. Then, a natural estimate of the distribution function of $P$ is its empirical cumulative distribution function $\hat{F}_n(y) = \sum_{i=1}^n n^{-1}I(Y_i \le y)$, which weighs each sample equally. In our setting, we are interested in estimating a shifted distribution $Q$ given samples from $P$.
	Noting that the likelihood ratio of the joint shifted distribution $Q$ and the original distribution $P$ is proportional to $K\left( \frac{x_0-\cdot}{h}\right)$, one possible way to estimate the distribution function of $Q_Y$ is to reweight each sample from $P$ proportional to the likelihood ratio. 
	We denote the reweighted distribution as \begin{align} \label{eq:tildeqn}
		\tilde{Q}_n(y) = \dfrac{1}{n}\sum_{i=1}^n \dfrac{L_i}{\sum_j L_j}I(Y_i \le y),
	\end{align} where $L_i = K\left( \frac{x_0-X_i}{h}\right)$. Without loss of generality, assume that $\sum_{j=1}^n L_j$ is strictly positive. (If $\sum_{j=1}^n L_j = 0$, we only have samples outside the region of interest and, therefore, cannot do meaningful inference.)

	\begin{lem} \label{lem:tildeqn}
		Let $\tilde{Q}_n$ be defined as (\ref{eq:tildeqn}). Then the following holds:
		\begin{enumerate}[(a)]
			\item $\tilde{Q}_n(y)$ is right continuous and monotonically increasing in $y$.
			\item Suppose that $Q_Y(y)$ is differentiable at $y = \theta_p$. Then, for every sequence $\{a_n\}$ such that $a_n = O(n^{-1/2})$, \begin{align*}
				\sqrt{n}\left[(Q(\theta_p + a_n)-Q(\theta_p)) - (\tilde{Q}_n(\theta_p + a_n)-\tilde{Q}_n(\theta_p)) \right] \prob 0
			\end{align*} as $n \to \infty$.
			\item We have \begin{align*}
				\sqrt{n}(\tilde{Q}_n(y)-Q(y)) \dist \mathcal{N}(0, \sigma_p^2(y)), \quad \sigma_p^2(y) = \dfrac{\mathbb{E}_P[L^2(I(Y \le y)-Q(y))^2]}{\mathbb{E}_P[L]^2}.
			\end{align*} 
		\end{enumerate}
	\end{lem}
	Above, $\mathbb{E}_P[\cdot]$ denotes expectation over samples from $P$. 

		By taking appropriate quantiles of the reweighted distribution function $\tilde{Q}_n$ as the endpoints, we can construct confidence intervals that have valid asymptotic coverage.
		
		Let  \begin{align} \label{eq:wqCI}
			\widehat{C}_{n}^{\text{ wq}}(x_0) = [\mathsf{Q}(\tilde{Q}_n ; \  \hat{p}_1), \mathsf{Q}(\tilde{Q}_n ; \  \hat{p}_2)],
		\end{align}
		where $\mathsf{Q}(P; \alpha)$ denotes the $\alpha$-th quantile of the distribution $P$.
		
		\begin{prop} \label{prop:covgGen} For any $\hat{p}_1$ and $\hat{p}_2$ that satisfy 
			\begin{align} \label{eq:quantileCondition}
				\hat{p}_1 = p + \frac{z_{\alpha_1}}{\sqrt{n}}\sigma_p(\theta_p) + o_p(n^{-1/2}) \text{ and } \hat{p}_2 = p + \frac{z_{1-\alpha_2}}{\sqrt{n}}\sigma_p(\theta_p) + o_p(n^{-1/2}),
			\end{align} the asymptotic coverage of $\theta_{p}$ from (\ref{eq:wqCI}) is equal to $1-\alpha_1-\alpha_2$ for $0 \le \alpha_1, \alpha_2 < 1$ and $0 < \alpha_1 + \alpha_2 < 1$. 
		\end{prop}
		
		In Algorithm \ref{alg:wq}, we summarize the Weighted Quantile method with specific $\hat{p}_1$ and $\hat{p}_2$ that satisfies equation (\ref{eq:quantileCondition}). 
		We note that for the case of $p = 0.5$, the asymptotic variance of $\tilde{Q}_n(y)$ at $y = \theta_{0.5}$ simplifies to $\sigma_{0.5}^2(\theta_{0.5}) = \dfrac{\E_P[L^2]}{4\E_P[L]^2}$ and we could use $\hat{\sigma}_{0.5}^2(\theta_{0.5}) = \dfrac{n^{-1}\sum_{i=1}^n L_i^2}{4(n^{-1}\sum_{i=1}^n L_i)^2}$ in Step 4 of Algorithm \ref{alg:wq}.
		
		\begin{algorithm} 
			\KwInput{Data $\mathcal{D} = \{Z_1, \ldots, Z_n\}$, point of interest $x_0$, kernel $K$, bandwidth $h$, significance level $\alpha \in (0,1)$, lower significance level $\alpha_1 \in [0, \alpha]$,  quantile $p \in (0,1)$.}
			\textbf{Process: }
			\begin{enumerate}
				\item Compute $L_i = K\left(\frac{x_0-X_i}{h}\right)$ for $i = 1, \ldots, n$.
				\item Set $\tilde{Q}_n(y) = \sum_{i=1}^n \dfrac{L_i}{\sum_{j=1}^n L_j} I(Y_i \le y)$.
				\item Compute $\tilde{\theta}_p = \tilde{Q}_n^{-1}(p)$.
				\item Compute $\hat{\sigma}_p^2(\theta_p) = \dfrac{n^{-1}\sum_{i=1}^n L_i^2(I(Y_i \le \tilde{\theta}_p)-p)^2}{(n^{-1}\sum_{i=1}^n L_i)^2}$
				\item Compute $	\hat{p}_1 = p + \frac{z_{\alpha_1}}{\sqrt{n}}\hat{\sigma}_p(\theta_p)$  and  $\hat{p}_2 = p + \frac{z_{1-\alpha+\alpha_1}}{\sqrt{n}}\hat{\sigma}_p(\theta_p)$.
				\item Set $\widehat{C}_{n,K,h,p}^{\text{ wq}}(x_0) = [\mathsf{Q}(\tilde{Q}_n ; \  \hat{p}_1), \mathsf{Q}(\tilde{Q}_n ; \  \hat{p}_2)]$.
			\end{enumerate}
			\KwOutput{$\widehat{C}_{n,K,h,p}^{\text{ wq}}(x_0; \mathcal{D})$: confidence interval for ${\theta}_{p}$ in (\ref{eq:object}).}
			\caption{Weighted Quantile Algorithm}
			\label{alg:wq}
		\end{algorithm}  
		
		Our next Theorem shows that this yields exact asymptotic coverage. 
		\begin{thm} \label{thm:wqCorrCovg}
			For all distributions $P$ on $\R^d \times \R$ such that $Q_Y$ is differentiable at $y={\theta}_{p}$, the output of the Weighted Quantile Algorithm has asymptotic $1-\alpha$ coverage, \begin{align*}
				\lim_{n \to \infty} \mathbb{P}\left({\theta}_{p} \in \widehat{C}_{n,K,h,p}^{\mathrm{\ wq}}(x_0; \mathcal{D}) \right) = 1-\alpha.
			\end{align*}
		\end{thm}

	\paragraph{On choosing the bandwidth}
	We recommend using the method only when the effective sample size $n_{\text{eff}} = \dfrac{(\sum_{i=1}^{n} L_i)^2}{\sum_{i=1}^{n} L_i^2}$ is sufficiently large. If the effective sample size is too low, there are not sufficiently many samples near the point of interest to guarantee coverage. Numerical studies in Section \ref{sec:Simulation} show that the method usually attains the desired coverage when $n_{\text{eff}} \ge 10$ for $p \in \{0.2, 0.5, 0.7\}$. For extreme quantiles, we would need more samples as is typically the case with any inference methods for extreme quantiles.
	
	\subsection{Optimality} \label{subsec:optimality}
	
	The Weighted Quantile method guarantees asymptotic coverage of $\theta_p$ for any distribution $Q_Y$ that has a derivative at $y = \theta_p$. We now study the efficiency of the confidence interval by showing that it is asymptotically uniformly most accurate unbiased.
	
	First, we begin by reformulating the problem as an M-estimation task. Let $\rho_p$ be the pinball loss, defined as \begin{align*}
		\rho_p(x) = \begin{cases}
			px, & x \ge 0, \\
			-(1-p)x, & x < 0.
		\end{cases}
	\end{align*}
	Then, the $p$-th quantile, $\xi_p$, of a distribution $F$ can be written as \begin{align*}
		\xi_p = \arg\min_{\theta} \mathbb{E}_{Y \sim F}[\rho_p(Y-\theta)].
	\end{align*}
	
	Therefore, the object of inference can be understood as a solution to the minimization problem \begin{align*}
		\theta_p &= \arg\min_{\theta}\mathbb{E}_{Y\sim Q_Y}[\rho_p(Y-\theta)] \\ 
		&= \arg\min_{\theta}\mathbb{E}_{(X,Y)\sim P} \left[\rho_p(Y-\theta) K\left(\dfrac{x_0-X}{h}\right)\right],
	\end{align*}
	which is a locally weighted quantile. 
	In fact, $\tilde{\theta}_p = \tilde{Q}_n^{-1}(p) = \arg\min_{\theta} \sum_{i=1}^n \rho_p(Y_i-\theta) K\left(\dfrac{x_0-X_i}{h}\right)$ is an M-estimator. 
	
	A notion of semiparametric optimality is adequate in this setting since we have an infinite-dimensional model and are interested in estimating ${\theta}_p \in \R$.
	Semiparametric efficiency bounds stem from an ingenious idea of Stein \citep{stein1956efficient} and have been expanded by \cite{koshevnik1976non, bickel1993efficient} and many others. The idea is to consider one-dimensional parametric submodels containing $P$, compute the local asymptotic minimax bound, and take the supremum over all possible submodels. More formally, let $\mathscr{P}$ be a collection of smooth one-dimensional submodels $\{P_t\}$ where $P_{0} = P$ and $\theta(P_t)$ is differentiable at $t=0$. Then by taking the supremum over all possible submodels' local asymptotic minimax bound, we obtain the
	semiparametric efficiency bound for $\theta_p$ as follows:
	\begin{align}
		\sup_{\{P_t\} \in \mathscr{P}} \lim_{C \to \infty} \liminf_{n \to \infty} \sup_{|t| \le C/\sqrt{n}} \mathbb{E}_{P_t} [(\sqrt{n}(\hat{\theta}_n - \theta_p(P_t)))^2] \ge  \sup_{\{P_t\} \in \mathscr{P}}\dfrac{\frac{d}{dt}\theta_p(P_t)^2|_{t = 0}}{I(0)}. \label{eq:semilb}
	\end{align} The denominator on the right-hand side is the Fisher information in the parametric submodel at $t = 0$. Any estimator that attains the lower bound is called \textit{semiparametrically efficient}. We call $\phi$ the Efficient Influence Function (EIF) for estimating $\theta(P)$ if an estimator $\hat{\theta}_n$ is asymptotically linear with influence function $\phi$, and attains the lower bound in (\ref{eq:semilb}). For more details on semiparametric models and asymptotic efficiency, see Chapter 25 of \cite{van2000asymptotic} and references therein. 
	
	In order to compute the semiparametric efficiency bound for $\theta_p$, we introduce some notation. Let $m(\theta, Z) = \rho_p(Y-\theta) K\left(\dfrac{x_0-X}{h}\right)$ for $Z = (X, Y)$, and $M(\theta) = \mathbb{E}_P[m(\theta, Z)]$. Let $\mathcal{P}$ be a collection of distributions such that for each $P \in \mathcal{P}$ there exists a unique minimizer of $M(\theta)$ over $\R$, $\text{Var}_P(\nabla m(\Lambda(P), Z))$ is finite and $\nabla^2 \mathbb{E}_P[m(\Lambda(P), Z)]$ is finite and positive. Then, Proposition \ref{prop:semiLowerBound} shows that we can compute the semiparametric efficiency bound for $\theta_p$ in the form of a worst-case variance $\sigma_*^2$. 
	
	\begin{prop} \label{prop:semiLowerBound}
		The semiparametric efficiency bound for $\theta_p(\cdot)$ relative to the paths $\mathscr{P}$ in the model $\mathcal{P}$ at the distribution $P$ is $$\sigma_*^2 = \dfrac{\mathbb{E}_P[L^2(I(Y \le \theta_p)-p)^2]}{\mathbb{E}_P[L]^2Q_Y'(\theta_p)^2}, \qquad L = K\left(\dfrac{x_0-X}{h}\right).$$ In other words,
		\begin{align*}
			\sup_{\{P_t\} \in \mathscr{P}} \lim_{C \to \infty} \liminf_{n \to \infty} \sup_{|t-t_0| \le C/\sqrt{n}} \mathbb{E}_{P_t} [(\sqrt{n}(\hat{\theta}_n - \theta_p(P_t)))^2] \ge \sigma_*^2
		\end{align*}  and the Efficient Influence Function of $\theta_p$ is $\phi_P^*(Z) = \dfrac{L(I(Y\le \theta_p)-p)}{\mathbb{E}_P[L]Q_Y'(\theta_p)}$.
	\end{prop}
	
	
	We now show that $\tilde{\theta}_p = \tilde{Q}_n^{-1}(p)$ is an efficient estimator for $\theta_p$ by providing an asymptotic expansion of the quantile obtained from the reweighted cumulative distribution $\tilde{Q}_n$ in (\ref{eq:tildeqn}). Since the Weighted Quantile method defines the lower and upper bounds of the confidence interval in terms of the quantiles of $\tilde{Q}_n$, the expansion allows us to asymptotically analyze the interval. 
	\begin{prop} \label{prop:wqexp}
		Suppose $Q_Y(y)$ is differentiable at $y = \theta_p$ with a positive derivative. Then, for any $p_n = p + O(n^{-1/2})$, $\tilde{\theta}_{p_n, n} = \tilde{Q}_n^{-1}(p_n)$ satisfies \begin{align*}
			\tilde{\theta}_{p_n,n} = \theta_p + \dfrac{p_n-p}{Q'(\theta_p)} - \dfrac{\tilde{Q}_n(\theta_p)-p}{Q'(\theta_p)} + \tilde{R}_n,
		\end{align*} with $\sqrt{n}\tilde{R}_n \prob 0$.
	\end{prop}
	
	Taking $p_n = p$ in Proposition \ref{prop:wqexp}, we obtain the following corollary:
	\begin{cor} \label{cor:thetapeff}
		$\tilde{\theta}_p$ is a semiparametrically efficient estimator of $\theta_p$.
	\end{cor}
	
	Testing procedures based on efficient estimators are asymptotically uniformly most powerful. For two-sided testing, the test is asymptotically uniformly most powerful unbiased, and inverting the test leads to an asymptotically uniformly most accurate unbiased (UMAU) confidence interval \citep{choi1996asymptotically}. $(1-\alpha)$-UMAU confidence interval $C$ for $\theta$ satisfies the following:
	\begin{enumerate} [(i)]
		\item $\pr(\theta \in C) \ge 1-\alpha$.
		\item $\pr(\theta' \in C) \le 1-\alpha$ for all $\theta' \ne \theta$.
		\item For any other confidence interval $C_1$ satisfying (i) and (ii), $\pr(\theta' \in C(X)) \le \pr(\theta' \in C_1(X))$ for all $\theta' \ne \theta$.
	\end{enumerate}  It turns out that the confidence interval from the Weighted Quantile method when $\alpha_1 = \alpha_2 = \alpha/2$ is asymptotically uniformly most accurate unbiased.
	
	\begin{thm} \label{thm:wqumau}
		Let $\widehat{C}_{n}^{\mathrm{\ wq}}$ be the output of Algorithm \ref{alg:wq} with $\alpha_1 = \alpha/2$. Then the confidence interval $\widehat{C}_{n}^{\mathrm{\ wq}}$ is asymptotically uniformly most accurate unbiased. 
	\end{thm}
	
	
	\section{Quantile Rejection method} \label{sec:rejection}
	Whereas the previous section provided an asymptotically valid method, in this section, we give a method valid in finite samples.
	This rests on distribution-free quantile inference and rejection sampling.
	\subsection{Distribution-free quantile inference} \label{subsec:dfqi}
	We begin by observing that distribution-free, finite-sample-valid, confidence intervals for any quantile of a distribution can be computed via order statistics \citep{noether1972distribution}. To see this, suppose $X_1, \ldots, X_n$ are i.i.d.~samples from any $F$ (not necessarily continuous) and set $\theta_p$ to be the $p$-th quantile of $F$. Let $N^{\text{lo}} = \#\{1 \le i \le n: X_i \le \theta_p \}$ and $N^{\text{hi}} = \#\{1 \le i \le n: X_i \ge \theta_p \}$. Then, $N^{\text{lo}} \sim B(n, F(\theta_p))$ and $N^{\text{hi}} \sim B(n, 1- F(\theta_p-))$ where $F(x-)$ denotes the left limit of $F$ at $x$. Note that $N^{\text{lo}}$ stochastically dominates $B(n,p)$ and $N^{\text{hi}}$ stochastically dominates $B(n, 1-p)$ since $F(\theta_p) \ge p$ and $1-F(\theta_p-) \ge 1-p$.
	We deduce that  
	\begin{align*}
		\pr(X_{(i)} \le \theta_p) = \pr(N_{\text{lo}} \ge I_{i, \max}) \ge \pr(B(n,p) \ge I_{i, \max})
	\end{align*} and 
	\begin{align*}
		\pr(X_{(i)} \ge \theta_p) = \pr(N_{\text{hi}} \ge n-I_{i, \max}+1) &\ge \pr(B(n,1-p) \ge n-I_{i, \max}+1) \\
		&= \pr(B(n,p) \le I_{i, \max}-1),
	\end{align*} where $I_{i, \max}$ (resp. $I_{i, \min}$) is the maximum (resp. minimum) index of the order statistics with the same value as $X_{(i)}$. 
	When all the $X_i$'s are distinct, $I_{i,\max}$ and $I_{i, \min}$ will simply be equal to $i$.
	For ease of notation, define $X_{(0)} = -\infty$ and $X_{(n+1)} = \infty$ and set
	\begin{equation} \label{eq:quantileBounds}
		\begin{aligned}
			\hat{\ell} & = \sup\{i : 0 \le i \le n \text{ and } \pr(B(n,p) < I_{i,\max}) \le \alpha_1 \}, \\
			\hat{u} & = \sup\{j : 1 \le j \le n+1 \text{ and } \pr(B(n,p) > I_{j,\min}) \le \alpha_2 \}.
		\end{aligned}
	\end{equation} Then, we have that the coverage of $\theta_p$ by $[X_{(\hat{\ell})}, X_{(\hat{u})}]$ is \begin{equation} \label{eq:dfqi-covg}
		\begin{split}
			\pr(\theta_p \in [X_{(\hat{\ell})}, X_{(\hat{u})}]) &= \pr(N^{\text{lo}} \ge I_{i,\max} \text{ and } N^{\text{hi}} \ge n-I_{j,\min} + 1) \\
			&= 1-\pr(N^{\text{lo}} < I_{i,\max} \text{ or } N^{\text{hi}} < n-I_{j,\min} + 1) \\
			&\ge 1-\pr(N^{\text{lo}} < I_{i,\max})-\pr(N^{\text{hi}} < n-I_{j,\min} + 1) \\
			&\ge 1-\alpha_1-\alpha_2.
		\end{split}
	\end{equation} 
	
	Note that the two-sided interval asymptotically vanishes as long as $F(\theta_p)-F(\theta_p-\delta) > 0$ and $F(\theta_p + \delta)-F(\theta_p) > 0$  for all $\delta \in (0, \delta_1)$ for some sufficiently small $\delta_1$, or if $F$ jumps at $\theta_p$.
	
	We summarize the above procedure in Algorithm \ref{alg:dfqi}. The validity of the algorithm follows from the above argument. 
	
	\begin{algorithm} 
		\KwInput{Data $\mathcal{D} = \{X_1, \ldots, X_n\}$ i.i.d.~samples from a distribution $F$, lower significance level $\alpha_1 \in [0, 1)$, upper significance level $\alpha_2 \in [0, 1)$, quantile $p \in (0,1)$.}
		\textbf{Process: }
		\begin{enumerate}
			\item Compute the order statistics of $\mathcal{D}$, $X_{(1)} \le \ldots \le X_{(n)}$ and set $X_{(0)} = -\infty$ and $X_{(n+1)} = \infty$.
			\item Compute $I_{i, \max}$ (resp.~$I_{i, \min}$) to be the maximum (resp.~minimum) index of the order statistics with the same value as  $X_{(i)}$ for $i = 1, \ldots, n$.
			\item Set $\hat{\ell}$ and $\hat{u}$ as in \eqref{eq:quantileBounds}.
		\end{enumerate}
		\KwOutput{$[X_{(\hat{\ell})}, X_{(\hat{u})}]$: confidence interval for $p$-th quantile of $F$.}
		\caption{Distribution-free confidence interval for the $p$-th quantile $\theta_p$}
		\label{alg:dfqi}
	\end{algorithm}

	\subsection{Rejection sampling strategy}
	Rejection sampling \citep{von195113} is a method to obtain samples from a target distribution $Q$ that is hard to sample from directly by using samples from a proposal distribution $P$ that can be sampled from more easily.  Suppose the likelihood ratio of $Q$ and $P$ is uniformly bounded by some constant $M$, and that $M^{-1}\dfrac{dQ}{dP}$ is computable. To perform rejection sampling, we begin by generating i.i.d.~samples from $P$. Then, for each sample $Y$ from $P$, we calculate $M^{-1}Q(Y)/p(Y)$. If this ratio is greater than a random sample $U$ from a uniform distribution between 0 and 1, we accept $Y$ as a sample from $Q$.
	
	In particular, we can apply rejection sampling in our setting to obtain i.i.d.~samples from the shifted distribution $Q$ given samples from $P$. 
	Since the likelihood ratio is uniformly bounded by \begin{align*}
		\dfrac{q(x,y)}{p(x,y)} = \dfrac{K\left( \frac{x_0-x}{h}\right)}{\int p(u)K \left( \frac{x_0-u}{h}\right)du} \le  \dfrac{K_{\max}}{\int p(u)K\left( \frac{x_0-u}{h}\right)du} \overset{\text{let}}{=} M
	\end{align*} and $\dfrac{q(x,y)}{Mp(x,y)} = \dfrac{K\left(\frac{x_0-X_i}{h}\right)}{K_{\max}}$ is known, we can formulate Algorithm \ref{alg:rjq} as follows:
	\begin{enumerate}
		\item Apply rejection sampling to the samples $\mathcal{D}\sim P$ and obtain $\mathcal{D}'$, which are i.i.d.~samples from the target distribution $Q$.
		\item Apply distribution-free quantile inference on $\mathcal{D}'$. 
	\end{enumerate}
	
	\begin{algorithm} 
		\DontPrintSemicolon
		\KwInput{Data $\mathcal{D} = \{Z_1, \ldots, Z_n\}$, point of interest $x_0$, kernel $K$ uniformly bounded by $K_{\max}$, bandwidth $h$, significance level $\alpha \in (0,1)$, lower significance level $\alpha_1 \in [0, \alpha]$,  quantile $p \in (0,1)$.}
		\textbf{Process}:
		\begin{enumerate}
			\item Compute $w_i = \dfrac{K\left(\frac{x_0-X_i}{h}\right)}{K_{\max}}$.
			\item Sample $\mathbf{U}=\{U_i\}_{i=1}^n \iid \text{U}(0,1)$ and include samples with indices obeying $U_i \le w_i$ in $LS(x_0)= \{Z_{i_1}, \ldots, Z_{i_N}\}$.
			\item Take $\widehat{C}_{n, h, K, p}^{\text{ rj}}(x_0)$ to be the $(1-\alpha)$-confidence interval
			with lower tail probability $\alpha_1$ and upper tail probability $\alpha-\alpha_1$ using samples  $\{Y_{i_1}, \ldots, Y_{i_N}\}$ and Algorithm \ref{alg:dfqi}.
		\end{enumerate}
		
		\KwOutput{$\widehat{C}_{n, h, K}^{\text{ rj}}(x_0;\mathcal{D}, \mathbf{U})$: $(1-\alpha)$-confidence interval for ${\theta}_{p}$.}
		\caption{Quantile Rejection Algorithm}
		\label{alg:rjq}
	\end{algorithm}

	\begin{thm}
		For all distributions $P$ on $\R^d \times \R$, and for all sample sizes $n \ge 1$, the output of Algorithm \ref{alg:rjq} covers $\theta_p$ in (\ref{eq:object}) with probability at least $1-\alpha$. That is, $\widehat{C}_{n, h, K, p}^{\mathrm{\ rj}}(x_0;\mathcal{D}, \mathbf{U})$ satisfies \begin{align} \label{eq:rjcovg}
			\mathbb{P}_{\mathcal{D}, \mathbf{U}}\left({\theta}_{p} \in \widehat{C}_{n, h, K, p}^{\mathrm{\ rj}}(x_0;\mathcal{D}, \mathbf{U}) \right) \ge 1-\alpha.
		\end{align} Moreover, suppose that for some $\delta > 0$,  $Q_Y(\theta_p-\delta) < Q_Y(\theta_p) < Q_Y(\theta_p+\delta)$ or $Q_Y$ has a jump at $\theta_p$, then the width of the confidence interval converges to 0 as $n \to \infty$ if $0 < \alpha_1 < \alpha$. 
	\end{thm}
	\begin{proof}
		The coverage statement (\ref{eq:rjcovg}) follows directly from the fact that the samples obtained from rejection sampling are i.i.d.~draws from $Q_Y$ and the distribution-free quantile inference procedure has $1-\alpha$ coverage guarantee as shown in (\ref{eq:dfqi-covg}). 
		
		The relationship $0 < \alpha_1 < \alpha$ implies that the interval is two-sided. 
		The number of samples obtained from rejection sampling, that is $|LS(x_0)| = N$, follows a binomial distribution $B(n, 1/M)$, and hence $N$ goes to $\infty$ almost surely. Since the $CI_p(\cdot)$ procedure provides a confidence interval with width converging to 0 as $n \to \infty$ under the condition we assumed for $Q_Y$, we have that the width of the algorithm converges to 0 as $n \to \infty$.
	\end{proof}
	
	\paragraph{Remarks}
	The distribution-free quantile inference procedure will usually provide coverage slightly greater than $1-\alpha$. We can make the coverage exact by adding another layer of randomness in the procedure as in \cite{zielinski2005best}. 
	
	We also note that derandomization of Algorithm \ref{alg:rjq} is feasible using $p$-value aggregation techniques from \cite{Vovk2020}. However, we do not recommend this in practice as this will cause the intervals to be wider.
	
	\section{Numerical experiments} \label{sec:Simulation}
We examine the empirical coverage and width of the confidence intervals constructed by the Weighted Quantile and Quantile Rejection methods. We shall see that in some settings, the confidence intervals must necessarily be wide in order to be valid. Since we are posing a new inference problem, for comparing purposes, we were not able to find existing methods that achieve the same goal.

\subsection{Setup} \label{subsec: setup} 
We generate $n$ i.i.d.~samples $\{(X_1, Y_1), \ldots, (X_n, Y_n)\}$ with $Y_i = f_{\text{signal}}(X_i) + \epsilon_i$ and $X_i \sim \text{Unif}[0, 1]$. The regression functions $f$ are selected from a set commonly used in the nonparametric regression literature \citep{donoho1994ideal}. Figure \ref{fig:signals} illustrates these functions, and their formulas are provided in Appendix \ref{a.simul.fun}.  

The noise is drawn from one of the distributions specified below:
\begin{enumerate}[{Setting} 1., align=left]
	\item $\epsilon \sim \mathcal{N}(0, (0.3)^2)$,
	\item $\epsilon \mid X \sim \mathcal{N}(0, (0.3(X^2+1))^2)$,
	\item $\epsilon \mid X \sim \mathcal{N}(0, (0.3(X^2-X+5/4))^2)$.
\end{enumerate}
We show results for the homoscedastic case (Setting 1) in this section and present results for other settings in Appendix \ref{subsec:hetero}.

The vertical dotted red lines in Figure \ref{fig:signals} are local points $x_0$ where we will compute the confidence intervals. The points were chosen to include extreme points or points at which the regression function is rapidly changing.

\begin{figure}[ht]
	\begin{subfigure}{.5\textwidth}
		\centering
		\includegraphics[width=.8\linewidth]{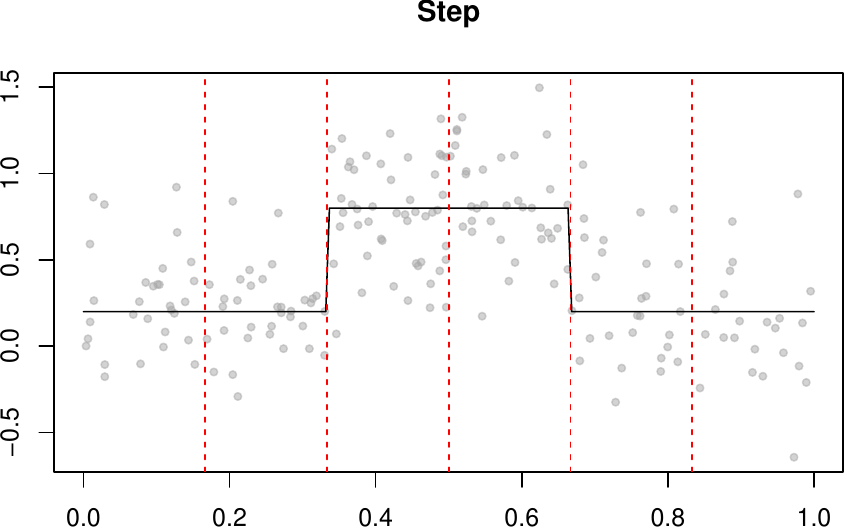}
	\end{subfigure}%
	\begin{subfigure}{.5\textwidth}
		\centering
		\includegraphics[width=.8\linewidth]{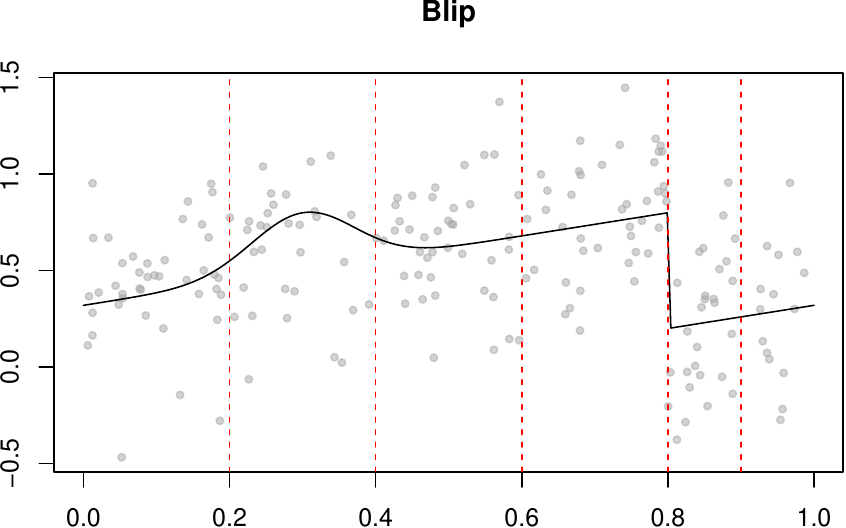}
	\end{subfigure}
	
	\begin{subfigure}{.5\textwidth}
		\centering
		\includegraphics[width=.8\linewidth]{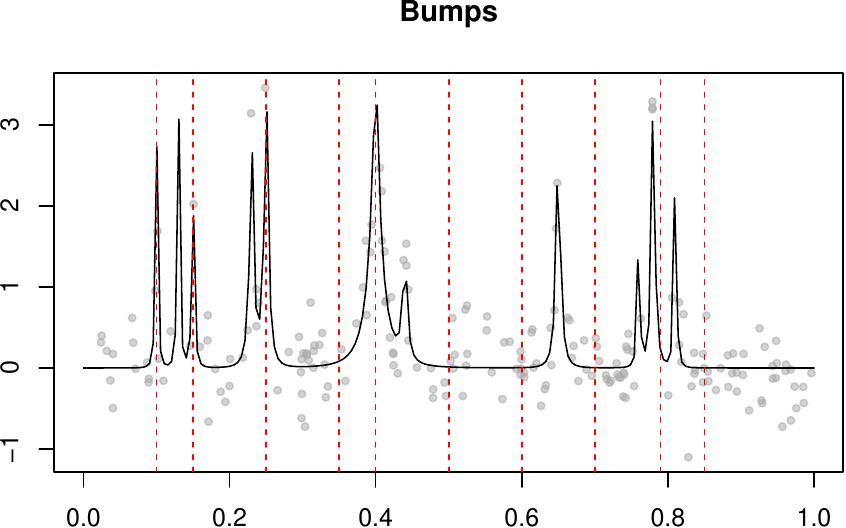}
	\end{subfigure}%
	\begin{subfigure}{.5\textwidth}
		\centering
		\includegraphics[width=.8\linewidth]{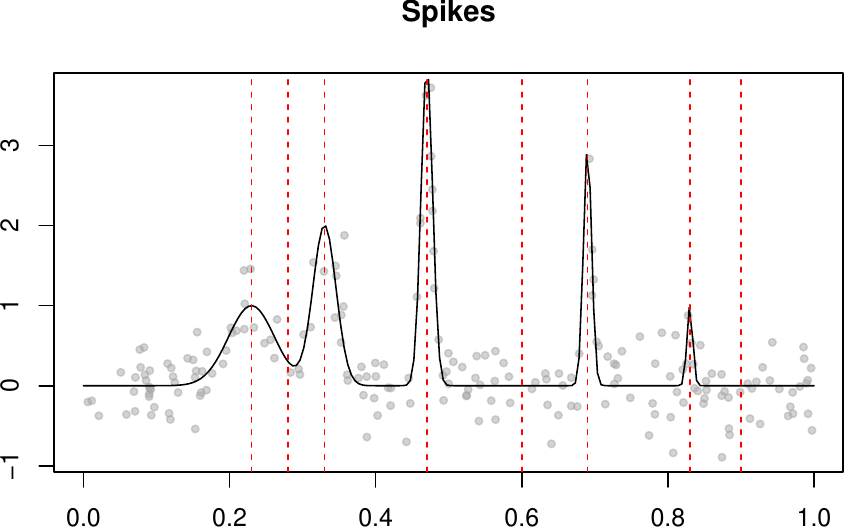}
	\end{subfigure}
	
	\begin{subfigure}{.5\textwidth}
		\centering
		\includegraphics[width=.8\linewidth]{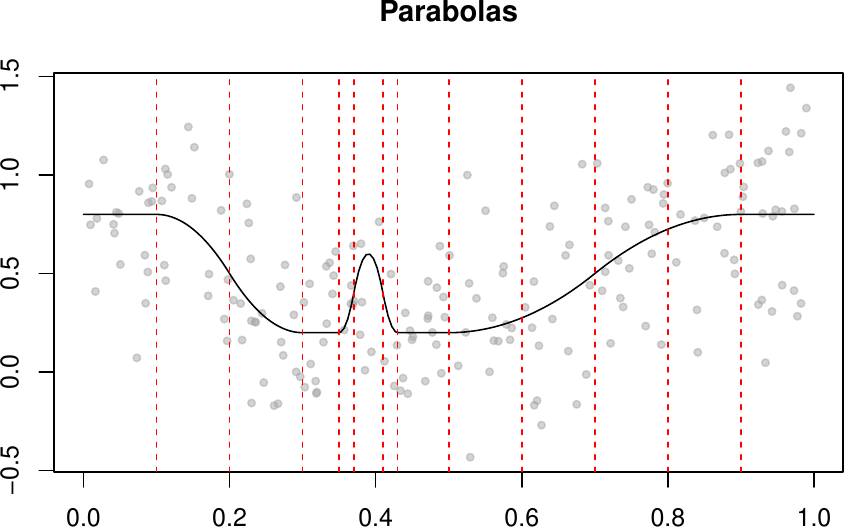}
	\end{subfigure}%
	\begin{subfigure}{.5\textwidth}
		\centering
		\includegraphics[width=.8\linewidth]{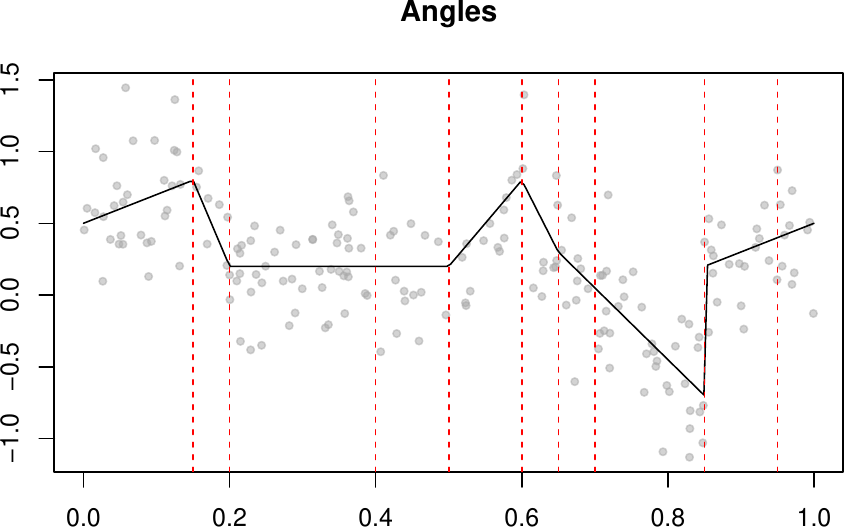}
	\end{subfigure}
	\caption{Regression functions and sampled data in setting 1.}
	\label{fig:signals}
\end{figure}

Recall that our target depends on the kernel, bandwidth, and the quantile, which can be flexibly chosen by the user.
We consider the following different settings for the target $\theta$:
\begin{enumerate}
	\item Kernel $K$: triangular kernel $K(x) = (1-|x|)_{+}$. We show results for the biweight kernel $K(x) = \frac{15}{16}(1-x^2)_+^2$ in Appendix \ref{subsec:biweight kernel}
	.
	\item Bandwidth $h$: 0.1, 0.08, 0.06, 0.04. This was chosen to make the expected sample size at least  10. 
	\item Quantile $p$: $p = 0.5$. We show results for $p = 0.2, 0.7$ in Appendix \ref{subsec:diff quantiles}
	.
	\item Local points $x_0$: for exact values, see Figures \ref{fig:step-s1-q0.5-tri} - \ref{fig:angles-s1-q0.5-tri}.
\end{enumerate}

We show how the object of inference varies with $h$ when the signal is the spikes function in Figure \ref{fig:spikes_varyh} (smaller values of $h$ mean increased resolution). 

\begin{figure}[ht]
	\centering
	\includegraphics[width=0.75\linewidth]{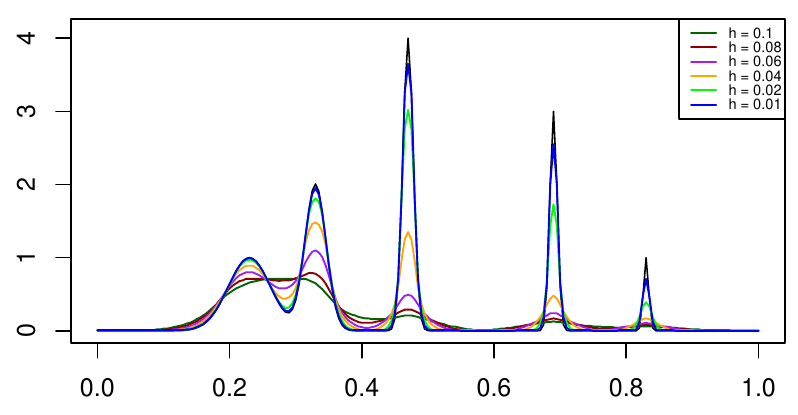}
	\caption{Object of inference $\theta_{1/2}$ for the Spikes signal when $h$ varies in setting 1.}
	\label{fig:spikes_varyh}
\end{figure}

We apply the Weighted Quantile and Quantile Rejection methods with target coverage of $1-\alpha = 0.9$ and $\alpha_1 = 0.05$. For each configuration, we take the sample size to be $n = 200$ and generate a total of $n_{\text{sim}} = 1000$ datasets. 

\subsection{Coverage and width}
Figures \ref{fig:bump-s1-q0.5-tri} - \ref{fig:spikes-s1-q0.5-tri} show the average width and empirical coverage. Additional results for different signals are in Appendix \ref{app:set1results}. We can see that both the Weighted Quantile method and Quantile Rejection method achieve 90\% coverage, regardless of the underlying distribution. The Quantile Rejection method tends to overcover and has a wider average width than the Weighted Quantile method. This is because it operates on a reduced sample size. In addition, the Quantile Rejection method outputs unbounded confidence intervals when the local sample size is extremely small. For example, the confidence interval for the median when $\alpha = 0.1$ will be the real line if there are only 4 local samples. When computing the width, we take the average of the bounded confidence intervals' width and indicate the percentage of infinite length CIs below each plot. While we have only an asymptotic guarantee for the Weighted Quantile method, the simulations show that it achieves the nominal coverage even at moderate sample sizes and points where rapid variations occur.

\begin{figure}[H]
	\centering
	\includegraphics[width=\textwidth]{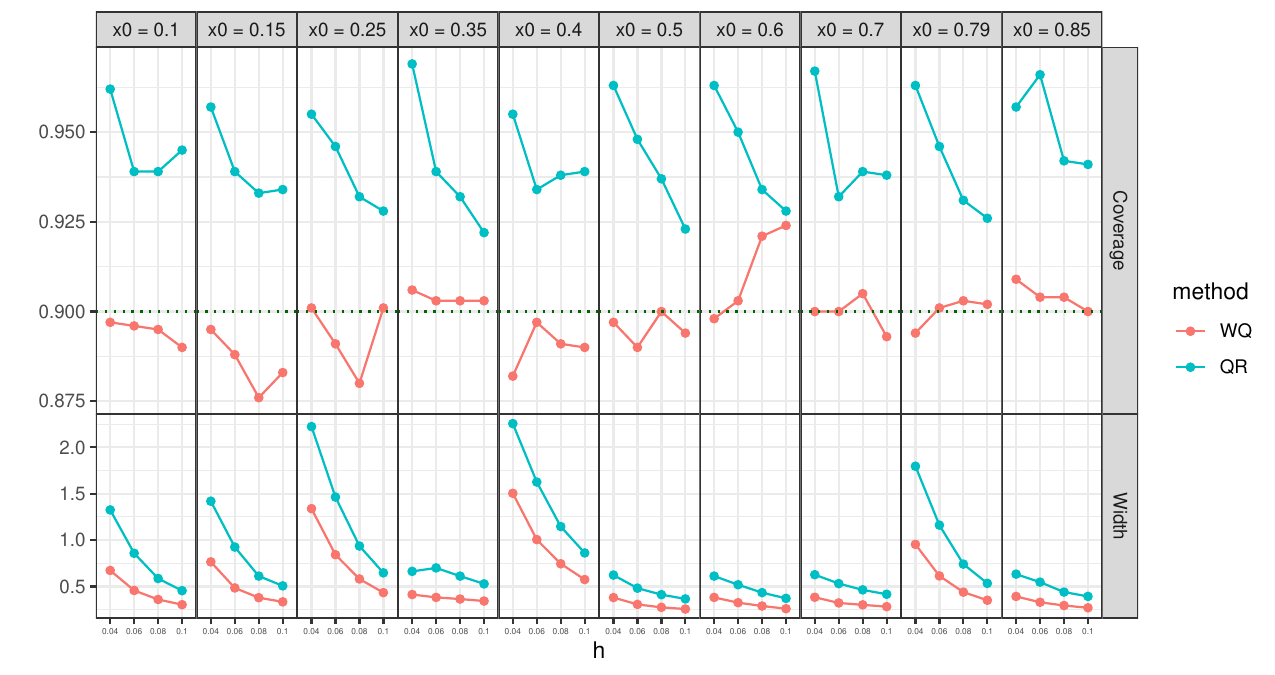}
	\caption{Coverage and width for the Bump signal, setting 1.}
	\label{fig:bump-s1-q0.5-tri}
\end{figure}

\begin{figure}[H]
	\centering
	\includegraphics[width=\textwidth]{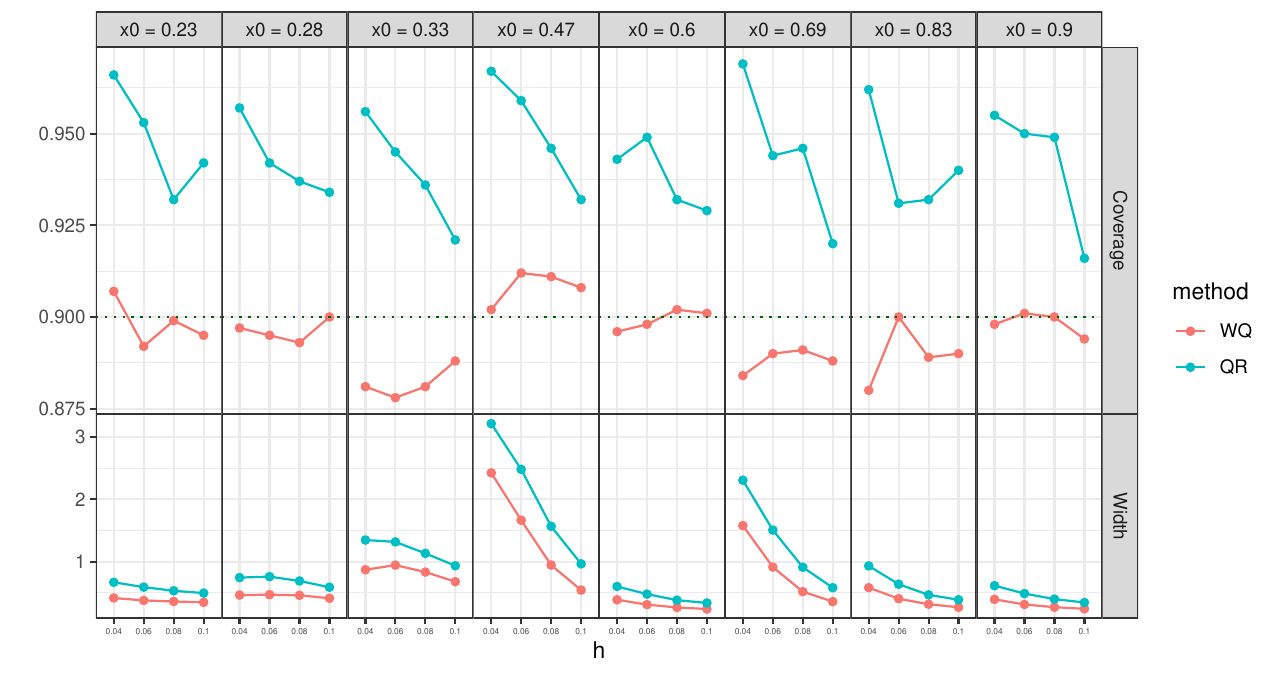}
	\caption{Coverage and width for the Spikes signal, setting 1.}
	\label{fig:spikes-s1-q0.5-tri}
\end{figure}

\section{Real data examples} \label{sec:realdata}
\subsection{Compliance data}
In this section, we go back to the compliance example \citep{efron1991compliance} introduced in Section \ref{subsec:motivation} and apply the Weighted Quantile and Quantile Rejection methods. 
We provide confidence intervals for the object of inference $\theta_{p}$ on a grid of $x_0$'s ranging from $-2$ to $2$ and quantile values $p \in \{0.3, 0.5, 0.75\}$. The bandwidth $h \in \{0.25, 0.5, 0.1\}$ is applied to the triangular kernel for $\alpha = 0.1$. We plot the confidence intervals obtained by the two proposed methods. 

\begin{figure}[h]
	\centering
	\includegraphics[width=.9\textwidth]{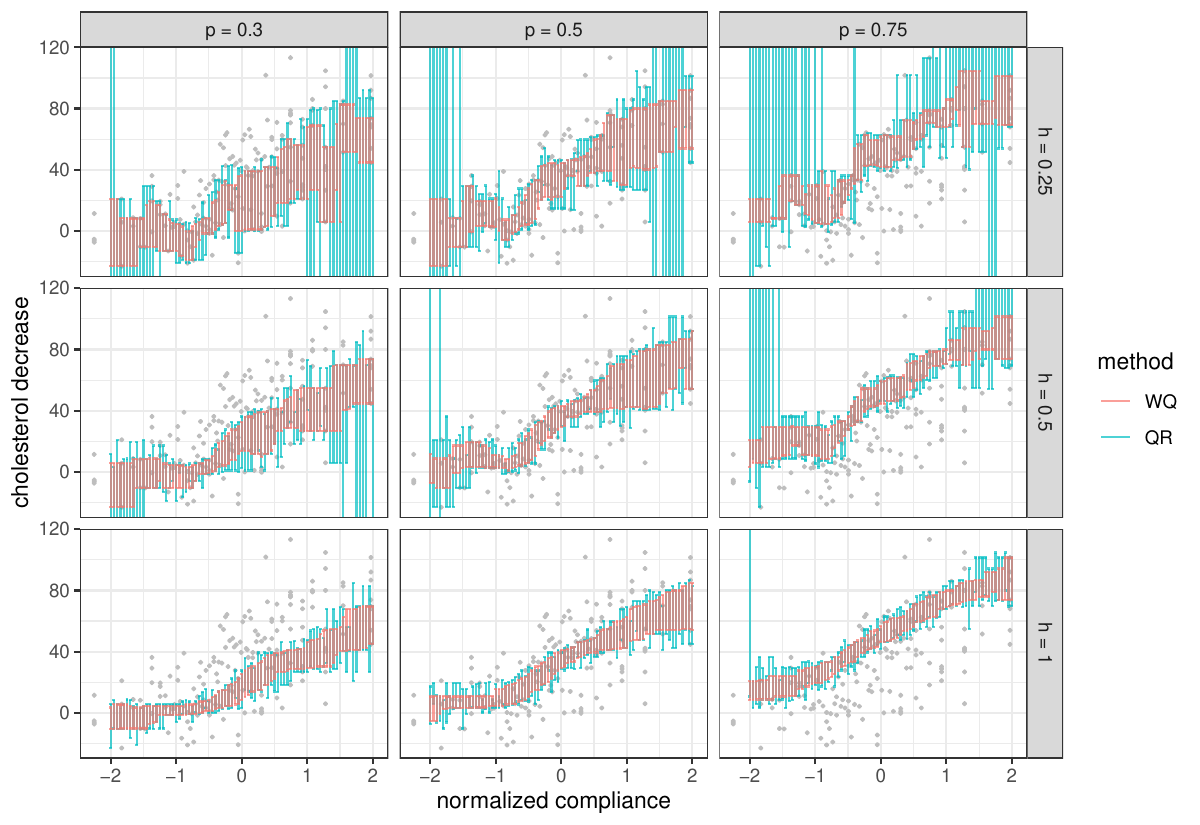}
	\caption{Confidence intervals using the Weighted Quantile (WQ) and Quantile Rejection (QR) methods for $\theta_p$. A scatter plot of the data is also shown.}
	\label{fig:compliance-cis}
\end{figure}

We cannot verify the coverage of the confidence intervals since the true underlying distribution is unknown. However, we can say with confidence 90\% that the median decrease in cholesterol level for people who comply between 45\% and 55\% (this corresponds to $h = 0.211$ in normalized compliance) with more weights on people that comply near 50\%, is in the range of 18.0 to 44.5 when we apply the Weighted Quantile method (18.0 to 47.25 for the Quantile Rejection method).

Moreover, we can observe that the confidence intervals get narrower as the bandwidth $h$ increases. This is natural since the effective sample size increases. We also see that the confidence intervals obtained from the Weighted Quantile method are narrower than those obtained from the Quantile Rejection method, as expected. 

\subsection{California housing data}

\begin{figure}[h!]
	\centering
	\includegraphics[width=0.6\linewidth]{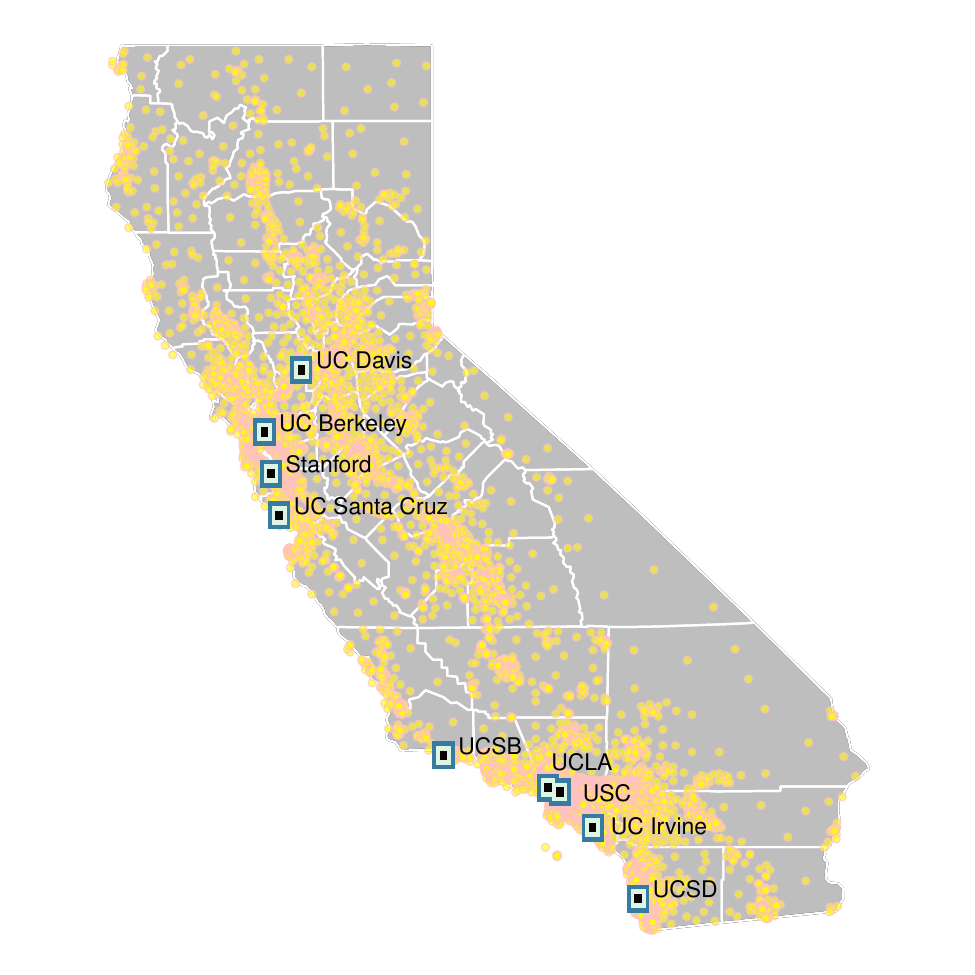}
	\caption{Geospatial representation of the housing dataset in California, with nine locations of interest highlighted. Yellow points on the map correspond to individual locations in the dataset. Squares mark the nine locations of interest, with bandwidths of $h=0.05$, $0.1$, and $0.15$ indicated in different colors.}
	\label{fig:ca_plot}
\end{figure}

We employ our proposed methods to obtain confidence intervals for the median housing prices of California districts, based on a modified version of the California Housing dataset. The dataset is based on the 1990 California census data and was originally obtained from the StatLib repository. The dataset contains information on the location of each census block group in California in terms of longitude and latitude along with housing characteristics such as average number of rooms and median age. We denote the variables as $Y$: median house value, 
$X_1$: longitude,  $X_2$: latitude, $X_3$: median housing age, and $X_4$: average number of rooms.


We begin by considering the local weighted quantile $\theta_{0.5}$ with longitude and latitude chosen from 9 universities in California as displayed in Figure \ref{fig:ca_plot}. We use a triangular kernel and set the bandwidth to be $h \in {0.05, 0.1, 0.15}$ for both longitude and latitude. The squares surrounding each region in Figure \ref{fig:ca_plot} show the local regions considered when we set $h$ to be 0.05, 0.1, and 0.15. The confidence intervals obtained from applying the Weighted Quantile and Quantile Rejection methods are plotted in Figure \ref{fig:housing_x1x2}.  We observe that the width of the confidence interval decreases when the bandwidth $h$ increases as effective sample size increases. Note that UCLA at $h = 0.05$ has confidence interval of [500,001, 500,001]. This is because the housing price in this dataset is right-censored at \$500,001 and a large proportion of the houses in this region has a value greater than $500$k.

\begin{figure}[h!]
	\centering
	\includegraphics[width=.8\linewidth]{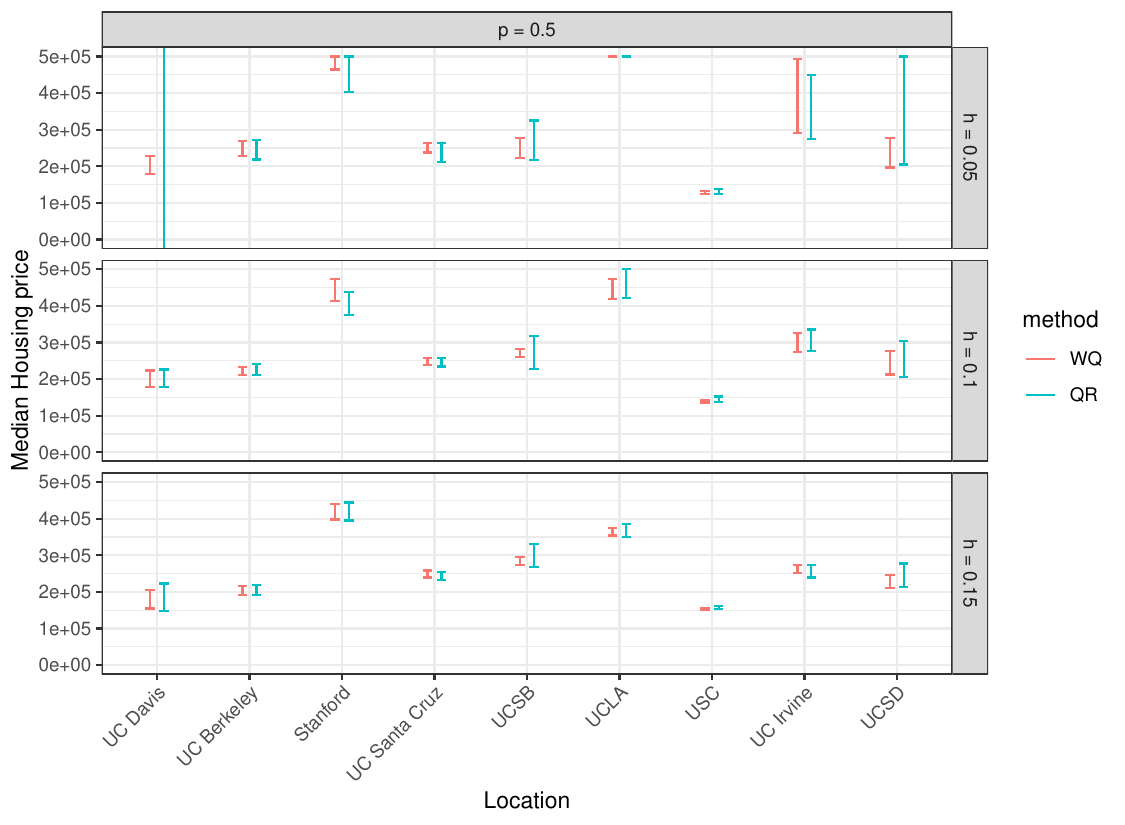}
	\caption{Confidence intervals using the WQ and QR methods for $\theta_{0.5}$ when using covariates longitude ($X_1$) and latitude ($X_2$).}
	\label{fig:housing_x1x2}
\end{figure}
\begin{figure}[h!]
	\centering
	\includegraphics[width=.8\linewidth]{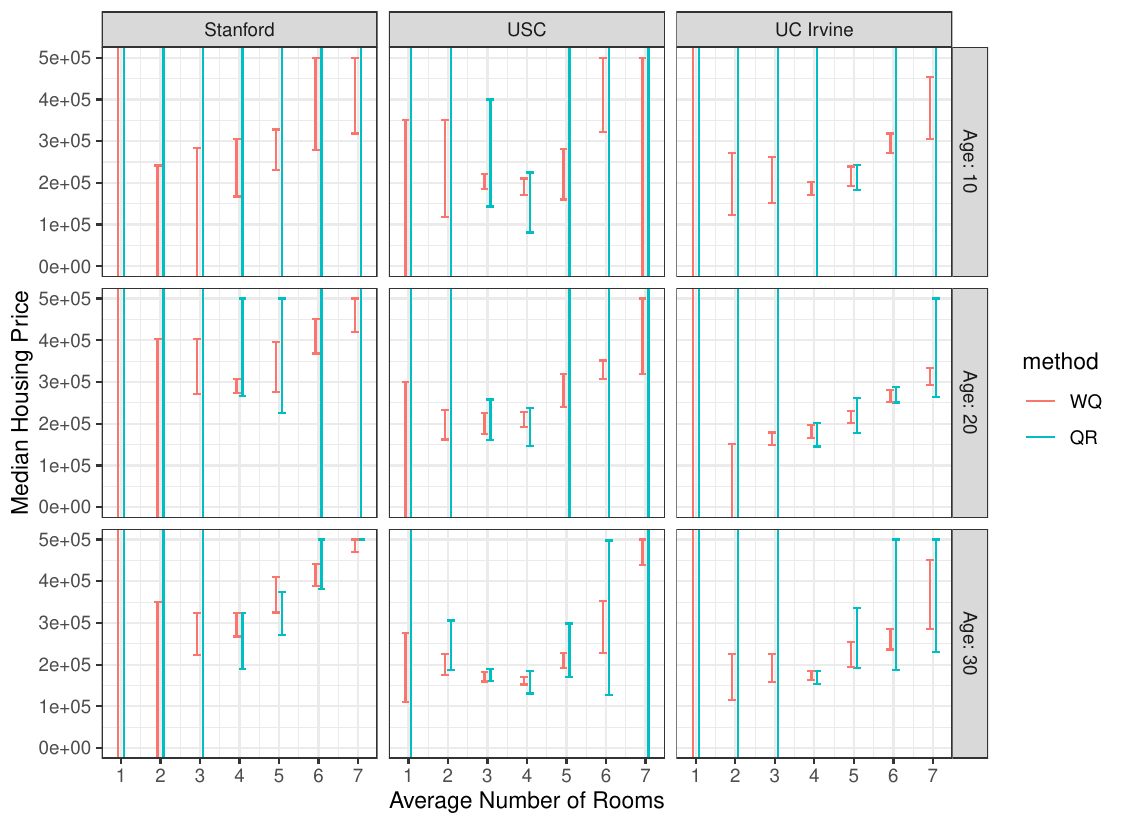}
	\caption{Confidence intervals using the WQ and QR methods for $\theta_{0.5}$ when using covariates longitude ($X_1$), latitude ($X_2$), median housing age ($X_3$) and average number of rooms ($X_4$).}
	\label{fig:housing_x1x2x3x4}
\end{figure}

Next, we apply the method using two additional covariates: median age of houses and average number of rooms. For the median age of houses, we set the point of interest $x_{03}$ to be 10, 20, and 30 with $h = 5$, and for the average number of rooms, we set $x_{04}$ to be 1 to 7 with $h = 1$, both with the triangular kernel. For the longitude and latitude, we set $h = 0.2$ to ensure sufficient samples in the region of interest.
The confidence intervals obtained from applying the proposed methods in three locations (Stanford, USC, UC Irvine) are plotted in Figure \ref{fig:housing_x1x2x3x4}. The confidence intervals are generally wider than those obtained using only two covariates.
We observe a trend of increasing housing prices when the average number of rooms exceeds 4 in all locations. It is not evident from the confidence intervals that the age of housing has a drastic impact on house prices.

\section{Indistinguishable distributions} \label{Sec:indistinguishable}
Finally, we investigate the efficiency of the methods.
We have shown in Section \ref{subsec:optimality} that the Weighted Quantile method is unimprovable at least in a local asymptotic minimax (LAM) sense. We here demonstrate the existence of a distribution that is almost indistinguishable form the true distribution but has a significantly different value of $\theta_p$. This implies that the confidence interval covering the target must be sufficiently wide to factor in the uncertainty.

We look at the setting where the underlying regression function is the Spikes signal with the target at $x_0 = 0.47$ and a triangular kernel with $h = 0.04$. The value of $\theta_{0.5}$ is 1.35, and in the simulation study, we have seen that with $n = 200$ samples, the average width of the confidence interval for $\theta_{0.5}$ is 2.49.

The distribution function of the shifted distribution is 
\begin{align*}
	Q_{Y}(y) = \int_{x_0-h}^{x_0+h}\int_{-\infty}^{y} \varphi_\epsilon(z-f(x)) dz \dfrac{K\left(\frac{x_0-x}{h}\right)}{\int_{x_0-h}^{x_0+h} K\left(\frac{x_0-u}{h}\right) du}  dx  \overset{\text{let}}{=} \int_{x_0-h}^{x_0+h} g(y,x) dx,
\end{align*} 
where $\varphi_\epsilon$ is the density function of the Gaussian distribution with mean 0 and standard deviation $\epsilon = 0.3$, $f$ is the Spikes signal, and $K(x) = (1-|x|)_{+}$. 
Figure \ref{fig:decom} shows the cdf of $Q_Y$.

Our goal is to find a distribution $P'$ that is very close to $P$ in the sense that it is almost indistinguishable with $n$ samples. Note that $Q_Y$  depends only on values of $x$ in the range $[x_0-h, x_0+h]$. Therefore, we take $P'$ to be equal to $P$ when $X \not\in [x_0-h, x_0+h]$. Now, we decompose $Q_Y$ into two parts according to whether $x \in [x_0-h_0, x_0+h_0]$ or not. We view $Q_Y$ as a mixture of two distribution functions $F_1$ and $F_2$:
\begin{equation*}
	Q_Y(y) =  \int_{x_0-h_0}^{x_0+h_0} g(y,x) dx + \int_{|x-x_0|>h_0} g(y,x) dx  \\
	= w({h_0}) F_1(y) + (1-w({h_0}))F_2(y),
\end{equation*}
where \begin{align*}
	w({h_0}) &= \dfrac{\int_{x_0-h_0}^{x_0+h_0} K\left(\frac{x_0-u}{h}\right) du}{\int_{x_0-h}^{x_0+h} K\left(\frac{x_0-u}{h}\right) du} = 1-\left(\dfrac{h-h_0}{h}\right)^2. 
\end{align*}

Figure \ref{fig:decom} shows the mixture components in the case where $h_0 = 0.012$. 

\begin{figure}[h!]
	\centering
	\includegraphics[width=0.7\linewidth]{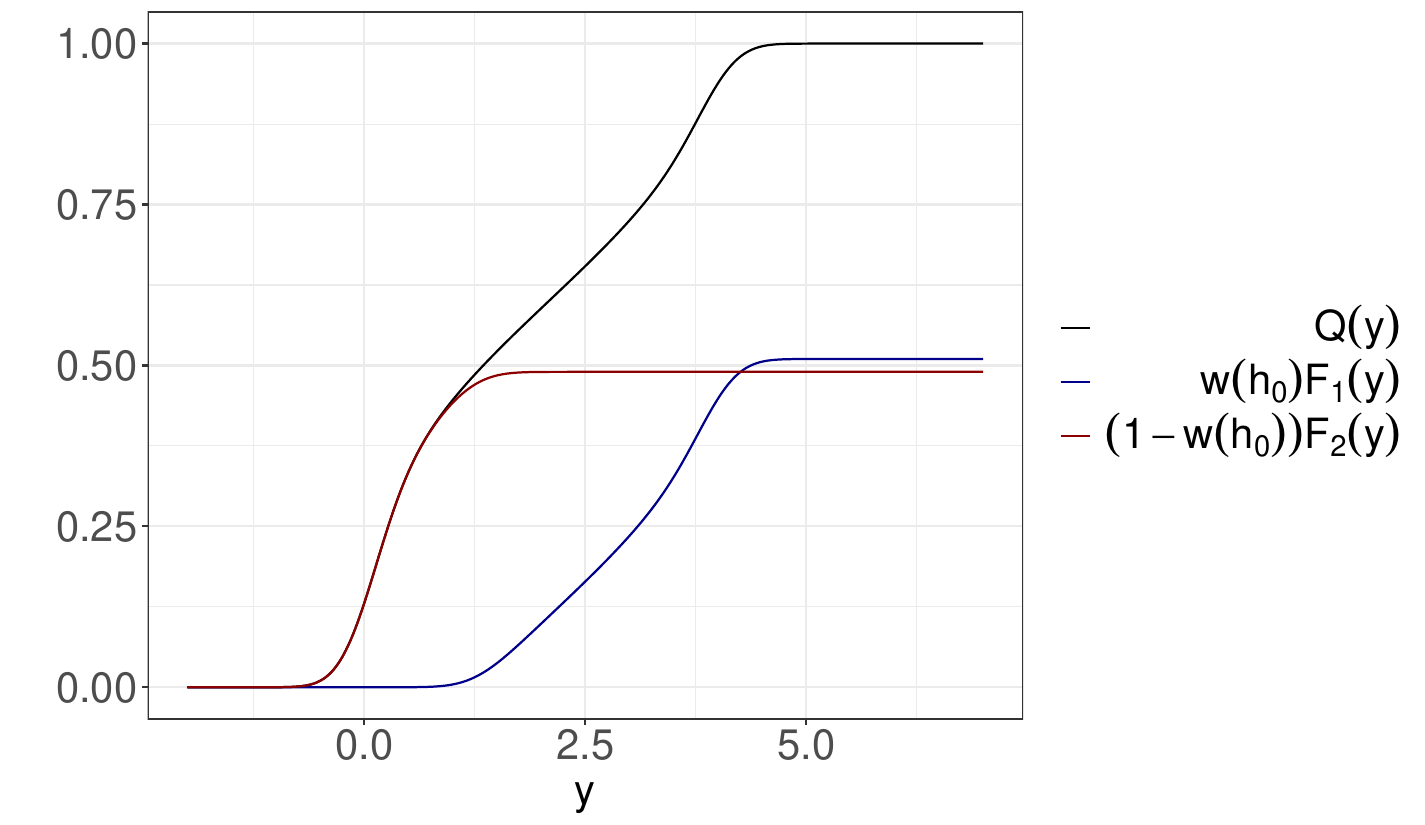}
	\caption{Plot of $Q_Y$, $w({h_0}) F_1(y)$ and $(1-w({h_0}))F_2(y)$ when $h_0 = 0.012$.}
	\label{fig:decom}
\end{figure}

The idea is to modify the distribution of $Y|X$ when $X \in [x_0-h_0, x_0+h_0]$. Thinking in terms of the mixture components, we find a distribution function $G_1(y)$ so that $\tilde{Q}_Y(y) = w({h_0}) G_1(y) + (1-w({h_0}))F_2(y)$ and $\tilde{Q}_Y^{-1}(1/2) \gg Q_Y^{-1}(1/2)$. $G_1(y)$ can be any distribution since we are in a distribution-free setting. Say we want $\tilde{Q}_Y^{-1}(1/2) = \theta(P')$. Define $G_1(y)$ by moving the mass of $F_1(y)$ such that $y < \theta(P')$ to $y = \theta(P')$ so that $G_1(y) = (F_1(y)+ F_1( \theta(P')))I(y > \theta(P'))$.
Figure \ref{fig:indist} plots $\tilde{Q}_Y(y)$ and $Q_Y(y)$ and marks $\theta(P) = {Q}_Y^{-1}(1/2) = 1.35$ and $ \theta(P') = \tilde{Q}_Y^{-1}(1/2) = 2.7$.

\begin{figure}[h!]
	\centering
	\includegraphics[width=0.7\linewidth]{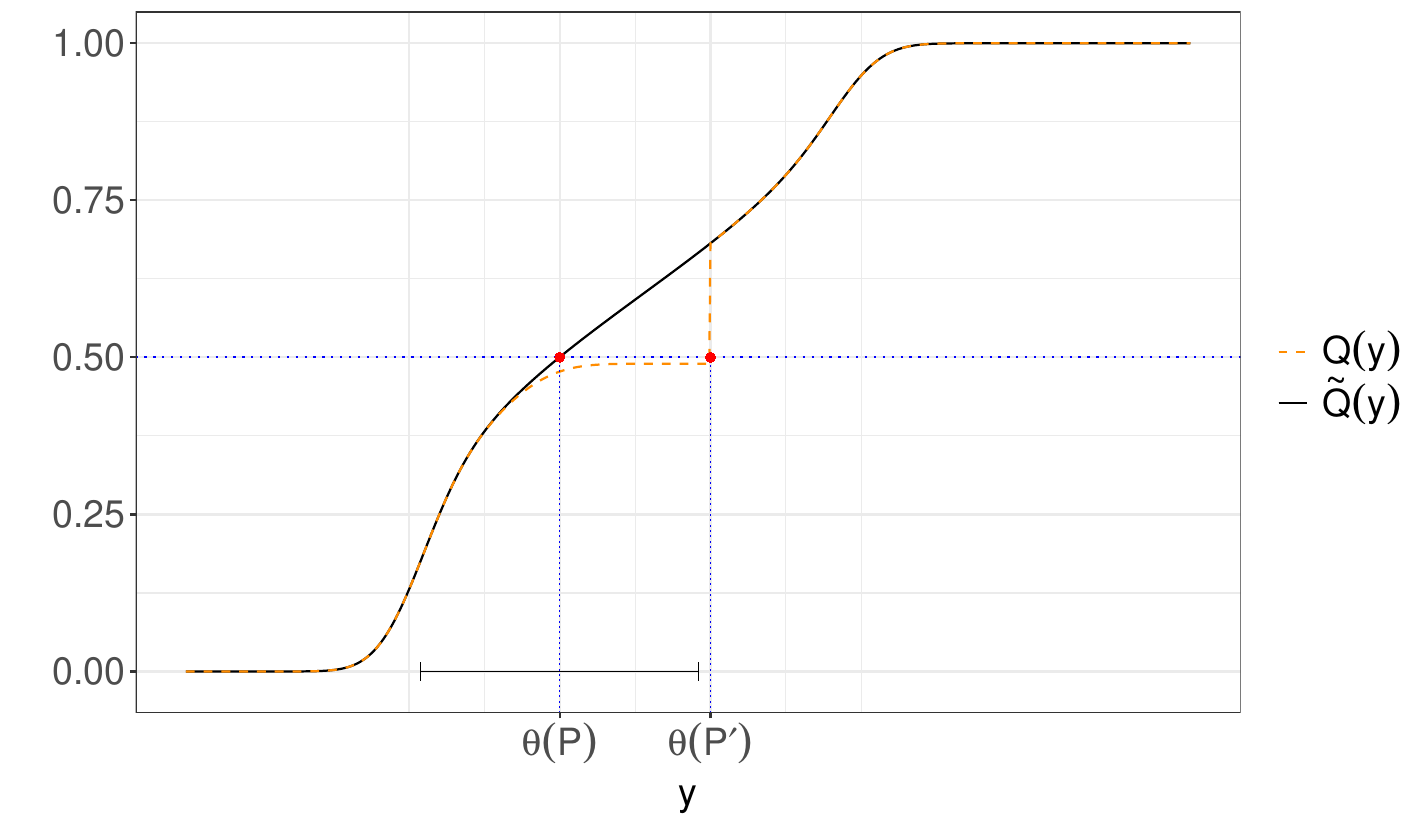}
	\caption{Plot of $Q_Y(y)$ and $\tilde{Q}_Y(y)$ when $h_0 = 0.012$. Error bar centered at $\theta(P)$ with size of average width from Weighted Quantile CIs is also shown. }
	\label{fig:indist}
\end{figure}

Note that $G_1$ can be constructed from $(X,Y) \sim P'$ with $X \sim \text{Unif}[0,1]$ and  
\begin{align*}
	Y|X=x =\begin{cases}
		\mathcal{N}(f(x), (0.3)^2),  & x \not\in [x_0-h_0, x_0+h_0], \\
		\mathcal{N}(f(x), (0.3)^2)I(Y>\theta^*) + \delta_{\theta^*}, & x \in [x_0-h_0, x_0+h_0].
	\end{cases} 
\end{align*} 
In this case, \begin{align*}
	d_{TV}(P, P') = \int_{x_0-h_0}^{x_0+h_0}\int_{-\infty}^{\theta^*} \varphi_{\epsilon}(y-f(x)) dy dx. 
\end{align*}
When $\theta(P') = 2.7$, we have $d_{TV}(P,P') \approx 0.010 = \dfrac{2}{200}$. This implies that $P$ and $P'$ are nearly indistinguishable from 200 random samples. However, our value of $\theta(P')$ differs by more than half the $90\%$ confidence interval width of $\theta(P)$ which is 1.245. Therefore, we can see that our method is not conservative but actually reflects the necessary uncertainty since nothing is assumed about the underlying distribution.

	\section{Discussion}
	
	This paper challenges the traditional inference targets such as the conditional mean and quantile, by pointing out the ambiguity of the notion of conditional distribution of a response variable $Y$ given a covariate $X$ being equal to some specific value. Instead, we propose a new object of inference that allows us to gain a localized understanding of the outcome, with the level of resolution being flexible and determined by the user's needs. Our approach yields reliable and easily computable confidence intervals without hidden constants.

	A finite-sample valid confidence interval for the univariate mean that is asymptotically efficient can be constructed when the distribution is supported on a compact set \citep{romano2000finite}. We leave it as future work to investigate if the Quantile Rejection method which is finite-sample valid is also efficient. If not, it is of interest to see if this gap can be reduced or if it is theoretically impossible to obtain finite-sample validity without loss of efficiency.

	\subsection*{Acknowledgments}
	J.J. would like to thank Sky Cao, Isaac Gibbs, Kevin Guo, Ying Jin and Chiara Sabatti for helpful discussions and feedback.
	J.J. was partly supported by ILJU Academy and Culture Foundation and a Ric Weiland Graduate Fellowship. E.C. was supported by the Office of Naval Research grant N00014-20-1-2157, the National Science Foundation grant DMS-2032014, the Simons Foundation under award 814641, and the ARO grant 2003514594. 
	
	\appendix
	\section{Deferred proofs} \label{app:proofs}
	\begin{lem}[\cite{lehmann2005testing} Lemma 11.3.3] \label{weightedCLT}
		Let $Y_1, Y_2, \ldots$ be i.i.d.~with mean 0 and finite variance $\sigma^2$. Let $w_1, w_2, \ldots$ be a sequence of constants. If $\dfrac{\max_{i=1, \ldots, n} w_i^2}{\sum_{j=1}^n w_j^2} \to 0$, then $\sum_{i=1}^n \dfrac{w_i Y_i}{\sqrt{\sum_{j=1}^n w_j^2}} \dist \mathcal{N}(0, \sigma^2)$.
	\end{lem}
	
	\begin{lem}[Lemma in \cite{david2004order} Section 10.2] \label{lem:order10.2}
		Let $\{V_n\}$, and $\{W_n\}$ be two sequences of random variables such that \begin{enumerate}
			\item $W_n = O_p(1)$, 
			\item for every $y$ and every $\epsilon > 0$, \begin{enumerate}
				\item $\lim_{n \to \infty}\pr(V_n \le y, W_n \ge y+\epsilon) = 0$,
				\item $\lim_{n \to \infty}\pr(V_n \ge y+\epsilon, W_n \le y) = 0$.
			\end{enumerate}
		\end{enumerate}
		Then, $V_n - W_n \prob 0$.
	\end{lem} 
	\begin{proof}
		Proof in p.286 of \cite{david2004order}.
	\end{proof}

	\subsection{Proof of Lemma \ref{lem:tildeqn}} \label{proof:lem:tildeqn}
	\begin{enumerate}[(a)]
		\item Right continuity holds since $I(Y_i \le y)$ are right continuous functions and $L_i$'s are non-negative with their sum $\sum_j L_j$ being strictly positive. Monotonicity holds by the construction of $\tilde{Q}_n$.
		\item
		Let $\tilde{Q}_{n, LR}(y) = \sum_{i=1}^n \dfrac{L_i}{n\mathbb{E}_P[L]}I(Y_i \le y)$ be the estimator of $Q_Y$ that is not normalized by $\sum_{i=1}^n L_i$, but has weights that are exactly the likelihood ratio of $P$ and $Q$. We first show that \begin{align*}
			\sqrt{n}\left[(Q(\theta_p + a_n)-Q(\theta_p)) - (\tilde{Q}_{n, LR}(\theta_p + a_n)-\tilde{Q}_{n, LR}(\theta_p)) \right] \overset{\text{let}}{=} W_n \prob 0
		\end{align*} as $n \to \infty$.
		Note that $\tilde{Q}_{n, LR}(y)$ is an unbiased estimator of $Q_Y(y)$ since \begin{align*}
			\mathbb{E}_P[Q_{n, LR}(y)] = \frac{\mathbb{E}_P[L_iI(Y_i \le y)]}{\mathbb{E}_P[L_i]} = \int_{-\infty}^{y}\int_{\mathcal{X}}\dfrac{q(x)}{p(x)}p(y,x)dx dy = \int_{-\infty}^{y}q_Y(y) dy = Q_Y(y).
		\end{align*} 
		Thus, $\mathbb{E}_P[W_n]=0$ and it suffices to show that $\mathbb{E}_P[W_n^2] \to 0$ to show that $W_n \prob 0$.
		We have that \begin{align*}
			\mathbb{E}_P[W_n^2] = n\left(\mathbb{E}_P[\tilde{Q}_{n, LR}(\theta_p + a_n) - \tilde{Q}_{n, LR}(\theta_p)]^2-(Q(\theta_p + a_n)-Q(\theta_p))^2\right).
		\end{align*} 
		So, we are left to show that for any $a_n = o(1)$, \begin{align}
			\mathbb{E}_P[\tilde{Q}_{n, LR}(\theta_p + a_n) - \tilde{Q}_{n, LR}(\theta_p)]^2 = (Q(\theta_p + a_n)-Q(\theta_p))^2 + o(n^{-1}).
		\end{align} Let $\rho_n = \min(\theta_p, \theta_p + a_n)$ and $\rho_n' = \max(\theta_p, \theta_p + a_n)$. Then, \begin{align*}
			\mathbb{E}_P[\tilde{Q}_{n, LR}(\theta_p + a_n) - \tilde{Q}_{n, LR}(\theta_p)]^2 = \mathbb{E}_P\Bigg[\sum_{i=1}^n  \dfrac{L_i}{n\mathbb{E}_P[L]}&I(\rho_n < Y_i \le \rho_n') \Bigg]^2 \\
			= \dfrac{1}{n^2\mathbb{E}_P[L^2]} \Bigg\{ \sum_{i=1}^n &\mathbb{E}_P[L_i^2 I(\rho_n < Y_i \le \rho_n')] + \\
			\sum_{1 \le i\ne j \le n} &\mathbb{E}_P[L_iL_j I(\rho_n < Y_i \le \rho_n')I(\rho_n < Y_j \le \rho_n')] \Bigg\} \\
			= \dfrac{\mathbb{E}_P[L^2 I(\rho_n < Y \le \rho_n')]}{n\mathbb{E}_P[L]^2}& + \dfrac{n-1}{n} \dfrac{\mathbb{E}_P^2[L I(\rho_n < Y \le \rho_n')]}{\mathbb{E}_P[L]^2} \\
			= \dfrac{\mathbb{E}_P^2[L I(\rho_n < Y \le \rho_n')]}{\mathbb{E}_P[L]^2} & + \dfrac{\mathbb{E}_P[L^2 I(\rho_n < Y \le \rho_n')]-\mathbb{E}_P^2[L I(\rho_n < Y \le \rho_n')]}{n\mathbb{E}_P[L]^2}\\
			= (Q_Y(\theta_p+a_n)-Q_Y(\theta_p&))^2 + \dfrac{s_n(a_n)}{n\mathbb{E}_P[L]^2}
		\end{align*} for $s_n(a_n) = \mathbb{E}_P[L^2 I(\rho_n < Y \le \rho_n')]-\mathbb{E}_P^2[L I(\rho_n < Y \le \rho_n')]$. We have that $s_n(a_n)$ converges to 0 since $s_n(a_n) \ge 0$ and $s_n(a_n) \le \mathbb{E}_P[L^2 I(\rho_n < Y \le \rho_n')] = \mathbb{E}_Q[LI(\rho_n < Y \le \rho_n')] \le \sup_u K(u) Q(Y \in (\rho_n, \rho_n'])$ converges to 0 as $n \to \infty$ by the differentiability of $Q_Y$ at $\theta_p$ and both $\rho_n$ and $\rho_n'$ converges to $\theta_p$. 
		
		We have shown (b) for $\tilde{Q}_{n, LR}$ rather than $\tilde{Q}_n$. Since $\tilde{Q}_n(y) = \dfrac{n\mathbb{E}_P[L]}{\sum_{i=1}^n L_i}\tilde{Q}_{n, LR}(y)$, it suffices to show that \begin{align}
			\sqrt{n}\left(1-\dfrac{n\mathbb{E}_P[L]}{\sum_{i=1}^n L_i}  \right)(\tilde{Q}_{n, IS}(\theta_p+a_n)-\tilde{Q}_{n, IS}(\theta_p)) \prob 0. \label{eq:tildeqn-b-1}
		\end{align}
		We have \begin{align*}
			1-\dfrac{n\mathbb{E}_P[L]}{\sum_{i=1}^n L_i} = \dfrac{1}{\sqrt{n}}\sum_i(L_i-\mathbb{E}_P[L_i])\frac{1}{n^{-1}\sum_i L_i}n^{-1/2} = o_p(n^{-1/2}).
		\end{align*} Moreover, \begin{align*}
			\sqrt{n}(Q_Y(\theta_p+a_n)-Q_Y(\theta_p)) = (Q'(\theta_p)+o(1))a_n n^{1/2} = O(1)
		\end{align*}
		and since $W_n \prob 0$, we have \begin{align*}
			\sqrt{n}(\tilde{Q}_{n, IS}(\theta_p+a_n)-\tilde{Q}_{n, IS}(\theta_p)) = O_p(1)
		\end{align*} and obtain (\ref{eq:tildeqn-b-1}).
		\item Applying the multivariate central limit theorem to the sum of $(L_iI(Y_i \le y), L_i)^T$'s, we obtain \begin{align*}
			\sqrt{n} \begin{bmatrix}
				a_n \\ b_n
			\end{bmatrix}= \sqrt{n} \begin{bmatrix}
				n^{-1}\sum_{i=1}^n (L_iI(Y_i \le y)-\mathbb{E}_P[L]Q(y)) \\ n^{-1}\sum_{i=1}^{n} (L_i-\mathbb{E}_P[L]) \end{bmatrix} \dist \mathcal{N} \left(
			0, \begin{bmatrix}
				\text{Var}(LI(Y\le y)) & \text{Cov}(LI(Y \le y), L) \\ 
				\text{Cov}(LI(Y \le y), L) & \text{Var}(L)
			\end{bmatrix}
			\right)
		\end{align*} 
		We can rewrite $\tilde{Q}_n(y)-Q(y)$ as \begin{align*}
			\tilde{Q}_n(y)-Q_Y(y) &= \dfrac{n^{-1}\sum_{i=1}^n(L_iI(Y_i\le y) - \mathbb{E}_P[L]Q_Y(y)) + \mathbb{E}_P[L]Q_Y(y) }{n^{-1}\sum_{i=1}^n (L_i-\mathbb{E}_P[L]) + \mathbb{E}_P[L]}-Q(y) \\
			&= \dfrac{a_n-b_nQ(y)}{b_n + \mathbb{E}_P[L]}.
		\end{align*}
		Then, for function $f(a,b) = \dfrac{a-bQ(y)}{b+\mathbb{E}_P[L]}$ which is differentiable at $(a,b) = (0,0)$, we can apply the delta method and get \begin{align*}
			\sqrt{n}(f(a_n,b_n)-f(0,0)) \dist \mathcal{N}(0, \nabla f(0,0)^T \Sigma \nabla f(0,0))
		\end{align*} for $\Sigma = \begin{bmatrix}
			\text{Var}(LI(Y\le y)) & \text{Cov}(LI(Y \le y), L) \\ 
			\text{Cov}(LI(Y \le y), L) & \text{Var}(L)
		\end{bmatrix}$. Since $\nabla f (a,b) = \begin{bmatrix}
			(b+\mathbb{E}_P[L])^{-1} \\ \frac{-Q(y)\mathbb{E}_P[L]}{(b+\mathbb{E}_P[L])^2}
		\end{bmatrix}$, we have \begin{align*}
			\sqrt{n}(\tilde{Q}_n(y)-Q(y)) \dist \mathcal{N}\left(0, \dfrac{\text{Var}(LI(Y\le y) - Q(y))}{\mathbb{E}_P[L]^2} \right)
		\end{align*} and $\sigma_p^2(y) = \dfrac{\text{Var}(LI(Y\le y) - Q(y))}{\mathbb{E}_P[L]^2} = \dfrac{\mathbb{E}_P[L^2(I(Y \le y)-Q(y))^2]}{\mathbb{E}_P[L]^2}$,
		since $\mathbb{E}_P[LI(Y\le y)]=Q(y)$.
	\end{enumerate}
	
	\subsection{Proof of Proposition \ref{prop:covgGen}}
	Note that using Slutsky's theorem and (c) of Lemma \ref{lem:tildeqn}, we have \begin{align*}
		P(\mathsf{Q}(\tilde{Q}_n ; \  \hat{p}_1) \le \theta_p) &= P(\hat{p}_1 \le \tilde{Q}_n(\theta_p)) \\
		&= P(\sqrt{n}(\hat{p}_1-p) \le \sqrt{n}(\tilde{Q}_n(\theta_p)-p)) \\
		&\to P(z_{\alpha_1} \le Z) = 1-\alpha_1
	\end{align*} as $n \to \infty$ for $Z \sim \mathcal{N}(0,1)$. Similar calculations on the upper bound gives the desired result. 
	
	\subsection{Proof of Theorem \ref{thm:wqCorrCovg}}
	It suffices to show that \begin{align*}
		\hat{\sigma}_p^2(\theta_p) \prob \sigma_p^2(\theta_p)
	\end{align*} as it will imply that the $\hat{p}_1$ and $\hat{p}_2$ in Step 5 of Algorithm \ref{alg:wq} satisfies (\ref{eq:quantileCondition}) by Slutsky's theorem and continuous mapping theorem.
	
	First, the denominator converges in probability to a positive value. By  law of large numbers and continuous mapping theorem, $(n^{-1} \sum_{i=1}^n L_i)^2 \prob (\E_P[L])^2$ as the kernel is uniformly bounded. 
	For the numerator, if suffices to show that \begin{align*}
		n^{-1}\sum_{i=1}^n L_i^2I(Y_i \le \tilde{\theta}(p)) \prob \mathbb{E}_P[L^2I(Y\le \theta_p)].
	\end{align*} Denote $Z_n(y) = n^{-1}\sum_{i=1}^n L_i^2I(Y_i \le y)$ and $z(y) = \mathbb{E}_P[L^2I(Y \le y)]$. Then, we need to show that $Z_n(\tilde{\theta}_{p}) \prob z(\theta_p)$ for $\tilde{\theta}_{p} = \tilde{Q}_n^{-1}(p)$. Say the kernel is uniformly bounded by some constant $K_{\max} > 0$. For some $\epsilon > 0$, there exists $\delta > 0$ such that \begin{align} \label{eq:delta}
		\mathbb{E}_Q[I(\theta_p < Y < \theta_p + \delta)] \le \dfrac{\epsilon}{2K_{\max} \mathbb{E}_P[L]}
	\end{align} since $Q_Y$ is differentiable at $\theta_p$. Note that \begin{align} \label{eq:consistency}
		\mathbb{P}(|Z_n(\tilde{\theta}_{p})- z(\theta_p)| > \epsilon) = &\mathbb{P}(|Z_n(\tilde{\theta}_{p})- z(\theta_p)| > \epsilon, |\tilde{\theta}_{p}-\theta_p| \le \delta) \\ 
		&+ \mathbb{P}(|Z_n(\tilde{\theta}_{p})- z(\theta_p)| > \epsilon, |\tilde{\theta}_{p}-\theta_p| > \delta). \label{eq:consistency2}
	\end{align}
	For the second term in (\ref{eq:consistency2}), \begin{align*}
		\mathbb{P}(|Z_n(\tilde{\theta}_{p})- z(\theta_p)| > \epsilon, |\tilde{\theta}_{p}-\theta_p| > \delta) \le \mathbb{P}(|\tilde{\theta}_{p}-\theta_p| > \delta) \to 0 
	\end{align*} as $n \to \infty$ by the consistency of $\tilde{\theta}_{p}$ from Proposition \ref{prop:wqexp}. Since $Z_n(y)$ is monotone increasing, if $|\tilde{\theta}_{n,p}-\theta_p| \le \delta$, \begin{align*}
		|Z_n(\tilde{\theta}_{p})- z(\theta_p)| \le \min\{|Z_n(\theta_p + \delta) - z(\theta_p)|,  |Z_n(\theta_p - \delta) - z(\theta_p)|\}.
	\end{align*} Therefore, we can bound the first term (\ref{eq:consistency}) by \begin{align} \label{eq:consistency3}
		\mathbb{P}(|Z_n(\tilde{\theta}_{p})- z(\theta_p)| > \epsilon, |\tilde{\theta}_{p}-\theta_p| \le \delta) \le &\mathbb{P}(|Z_n(\theta_p + \delta) - z(\theta_p)| > \epsilon, |\tilde{\theta}_{p}-\theta_p| \le \delta) \\
		&+\mathbb{P}(|Z_n(\theta_p- \delta) - z(\theta_p)| > \epsilon, |\tilde{\theta}_{p}-\theta_p| \le \delta). \label{eq:consistency4}
	\end{align}
	Using the triangular inequality, (\ref{eq:consistency3}) is bounded by \begin{align*}
		\mathbb{P}(|Z_n(\theta_p + \delta) - z(\theta_p)| > \epsilon) \le \mathbb{P}(|Z_n(\theta_p + \delta) - z(\theta_p+\delta)| + |z(\theta_p+\delta) - z(\theta_p)|>\epsilon). 
	\end{align*}
	Now, note that \begin{align*}
		|z(\theta_p +\delta)-z(\theta_p)| &= \mathbb{E}_P[L^2I(\theta_p < y \le \theta_p+\delta)] \\
		&= \mathbb{E}_Q[LI(\theta_p < y \le \theta_p+\delta)] \mathbb{E}_P[L] \ (\because \text{Change of measure}) \\
		&\le K_{\max} \mathbb{E}_Q[I(\theta_p < y \le \theta_p+\delta)] \mathbb{E}_P[L] \ (\because L \le K_{\max}) \\
		&\le \epsilon/2 \ (\because (\ref{eq:delta}))
	\end{align*} and so \begin{align*}
		\mathbb{P}(|Z_n(\theta_p + \delta) - z(\theta_p+\delta)| + |z(\theta_p+\delta) - z(\theta_p)|>\epsilon) < \mathbb{P}(|Z_n(\theta_p + \delta) - z(\theta_p+\delta)| > \epsilon/2)
	\end{align*} which converges to 0 as $n \to \infty$ since $Z_n(\theta_p+\delta) \prob z(\theta_p+\delta)$ by the weak law of large numbers.
	Analogous to the above argument, (\ref{eq:consistency4}) converges to 0. Therefore, we have $\hat{\sigma}^2_{p}(\theta_p) \prob \sigma^2_p(\theta_p)$.

	\subsection{Proof of Proposition \ref{prop:semiLowerBound}} \label{proof:prop:semiLowerBound}
	
	Semiparametric efficient bound for M-estimators can be computed using results from \cite{newey1990semiparametric}.
	The Efficient Influence Function (EIF) for estimating $\Lambda(P) = \arg\min_{\theta} \mathbb{E}_P[m(\theta, Z)]$ is \begin{align*}
		\phi_P^*(z) = -[\nabla_{\theta}^2 M(\Lambda(P))]^{-1} \nabla_{\theta} m(\Lambda(P), z)
	\end{align*} for $M(\theta) = \mathbb{E}_P[m(\theta, Z)]$. 
	So, the semiparametric efficiency bound for estimating $\theta_p$ is $\sigma_*^2 = \dfrac{\text{Var}_P(\nabla_{\theta} m(\theta_p,Z))}{[\nabla_{\theta}^2 M(\theta_p)]^2}$.
	Note that \begin{align*}
		m(\theta, Z) = L(p(Y-\theta)_+ - (1-p)(Y-\theta)_{-}), \\
		\nabla_{\theta}{m}(\theta,Z) = L(-p I(Y > \theta)+(1-p)I(Y \le \theta)) 
	\end{align*} and 
	\begin{align*}
		\mathbb{E}_P[\nabla_{\theta}{m}(\theta,Z)] &= \mathbb{E}_P[L]\mathbb{E}_Q[-p I(Y > \theta)+(1-p)I(Y \le \theta)] \\
		&= \mathbb{E}_P[L](Q_Y(\theta)-p), \\
		\mathbb{E}_P[\nabla_{\theta}m(\theta_p, Z)] &= 0.
	\end{align*}
	Also, \begin{align*}
		\mathbb{E}_P[\nabla_{\theta}{m}(\theta,Z)]^2 &= \mathbb{E}_P[L^2\{(1-p)^2I(Y \le \theta) + p^2 I(Y > \theta) \}] \\
		&=\mathbb{E}_P[L^2\{p^2 + (1-2p)I(Y \le \theta) \}]
	\end{align*} and so
	\begin{align*}
		\text{Var}_P(\nabla_{\theta} m(\theta_p,Z)) = \mathbb{E}_P[\nabla_{\theta}{m}(\theta_p,Z)]^2 = \mathbb{E}_P[L^2\{p^2 + (1-2p)I(Y \le \theta_p) \}]=\mathbb{E}_P[L^2(I(Y \le \theta_p)-p)^2].
	\end{align*} We proved the equality of the numerator. Now, note that
	\begin{align*}
		M(\theta) = \mathbb{E}_P[m(\theta, Z)] &= \mathbb{E}_P[L]\mathbb{E}_Q[\rho_p(Y-\theta)] \\
		&=\mathbb{E}_P[L]\left((1-p)\int_{-\infty}^{\theta}(\theta-y)dQ_Y(y) +p\int_{\theta}^{\infty} (y-\theta)dQ_Y(y) \right)  \\
		&=\mathbb{E}_P[L]\left((1-p)\int_{-\infty}^{\theta}Q_Y(y) dy - p\int_{\theta}^{\infty} (Q_Y(y)-1)dy\right)
	\end{align*} and hence 
	\begin{align*}
		\nabla_{\theta}M(\theta) &= \mathbb{E}_P[L](Q_Y(\theta)-p), \\
		\nabla^2_{\theta}M(\theta) &= \mathbb{E}_P[L]Q_Y'(\theta).
	\end{align*} Therefore, we have $[\nabla_{\theta}^2 M(\theta_p)]^2 = \mathbb{E}^2_P[L]Q_Y'(\theta_p)^2$.

	\subsection{Proof of Proposition \ref{prop:wqexp}} \label{proof:prop:wqexp}
	
	Let $\dot{\theta}_{p_n} = \theta_p + \dfrac{p_n-p}{Q'(\theta_p)}$, $V_n = \sqrt{n} (\tilde{\theta}_{p_n,n}-\dot{\theta}_{p_n})$, and $W_n = \dfrac{\sqrt{n}}{Q'(\theta_p)}(p-\tilde{Q}_n(\theta_p))$. Then, $\sqrt{n}\tilde{R}_n = V_n-W_n$, and it suffices to show that $(V_n, W_n)$ satisfies the conditions of Lemma $\ref{lem:order10.2}$. First, we have that $W_n$ is $O_p(1)$ by Lemma \ref{lem:tildeqn} (c). For $V_n$, note that \begin{align*}
		\{V_n \le y \} &= \left\{\tilde{\theta}_{p_n,n} \le \dot{\theta}_{p_n} + \dfrac{y}{\sqrt{n}} \right\} \\
		&\subseteq \left\{\tilde{Q}_n(\tilde{\theta}_{p_n,n}) \le  \tilde{Q}_n(\dot{\theta}_{p_n} + y/\sqrt{n}) \right\} \\
		&=\left\{ \dfrac{\sqrt{n}}{Q'(\theta_p)}\left(Q(\dot{\theta}_{p_n} + y/\sqrt{n}) - \tilde{Q}_n(\dot{\theta}_{p_n} + y/\sqrt{n}) \right)  \le \dfrac{\sqrt{n}}{Q'(\theta_p)}\left(Q(\dot{\theta}_{p_n} + y/\sqrt{n}) - p_n \right) \right\} \\
		&=\{Z_n \le y_n\},
	\end{align*} where $Z_n = \dfrac{\sqrt{n}}{Q'(\theta_p)}\left(Q(\dot{\theta}_{p_n} + y/\sqrt{n}) - \tilde{Q}_n(\dot{\theta}_{p_n} + y/\sqrt{n}) \right)$ and $y_n = \dfrac{\sqrt{n}}{Q'(\theta_p)}\left(Q(\dot{\theta}_{p_n} + y/\sqrt{n}) - p_n \right)$.
	Using differentiability of $Q(y)$ at $y = \theta_p$, we have that \begin{align*}
		y_n &= \dfrac{\sqrt{n}}{Q'(\theta_p)}\left(Q(\theta_p) + (Q'(\theta_p)+o(1))Q\left(\dfrac{p_n-p}{Q'(\theta_p)} + \dfrac{y}{\sqrt{n}}\right) - p_n \right) \\
		&= y + \dfrac{\sqrt{n}}{Q'(\theta_p)}\left(p-p+\left(\dfrac{p_n-p}{Q'(\theta_p)} \right)o(1) \right) \\
		&= y + o(1),
	\end{align*} where the last equality follows from $p_n-p = O(n^{-1/2})$. 
	For any $\epsilon >0$, there exists some $N \in \N$ such that $|y_n -y| < \epsilon/2$ for all $n \ge N$. Then, for $n \ge N$, \begin{align} \label{eq:lem-order-subset}
		\{V_n \le y, W_n \ge y+\epsilon \} \subseteq \{Z_n \le y_n, W_n \ge y+\epsilon \} \subseteq \{|W_n-Z_n| \ge \epsilon/2\}.
	\end{align}
	Note that \begin{align*}
		Z_n-W_n = \dfrac{\sqrt{n}}{Q'(\theta_p)}\left(Q(\dot{\theta}_{p_n} + y/\sqrt{n}) - \tilde{Q}_n(\dot{\theta}_{p_n} + y/\sqrt{n}) - p + \tilde{Q}_n(\theta_p) \right) \prob 0
	\end{align*} by Lemma \ref{lem:tildeqn} (b) using the fact that $\dot{\theta}_{p_n} + \dfrac{y}{\sqrt{n}} = p + O(n^{-1/2})$. Using $Z_n-W_n \prob 0$ and (\ref{eq:lem-order-subset}), we get that $\pr(V_n \le y, W_n \ge y+\epsilon) \le \pr(|W_n-Z_n| \ge \epsilon/2) \to 0$ as $n \to \infty$, and we have verified the first condition of (b) in Lemma \ref{lem:order10.2}. The second condition of (b) can be proved similarly.
	Now, applying Lemma \ref{lem:order10.2}, $\sqrt{n}\tilde{R}_n \prob 0$.

	\subsection{Proof of Corollary \ref{cor:thetapeff}} \label{proof:cor:thetapeff}
	Taking $p_n = p$ from Proposition \ref{prop:wqexp}, we have \begin{align}
		\tilde{\theta}_p = \theta_p + \dfrac{\tilde{Q}_n(\theta_p)-p}{Q'(\theta_p)} + o_p(n^{-1/2}). \label{eq:exp}
	\end{align} By Lemma 25.23 of \cite{van2000asymptotic}, it suffices to show that \begin{align*}
		\sqrt{n}(\tilde{\theta}_p - \theta_p) &= \dfrac{1}{\sqrt{n}} \sum_{i=1}^n \phi_P^*(Z_i) + o_p(1) \\
		&= \dfrac{1}{\sqrt{n}} \sum_{i=1}^n\dfrac{L_i(I(Y_i \le \theta_p)-p)}{\mathbb{E}_P[L_i]Q_Y'(\theta_p)}+ o_p(1).
	\end{align*} Rearranging (\ref{eq:exp}), we get \begin{align*}
		\sqrt{n}(\tilde{\theta}_p - \theta_p) &= \dfrac{\tilde{Q}_n(\theta_p)-p}{Q'(\theta_p)} + o_p(1) \\
		&=  \dfrac{1}{\sqrt{n}} \sum_{i=1}^n\dfrac{L_i(I(Y_i \le \theta_p)-p)}{(n^{-1}\sum_{j=1}^n L_j )Q_Y'(\theta_p)}+ o_p(1).
	\end{align*} Since we have \begin{align*}
		\dfrac{1}{\mathbb{E}_P[L]} - \dfrac{1}{n^{-1}\sum_{j=1}^n L_j} = o_p(n^{-1/2}),
	\end{align*} we obtain the desired equality.

	\subsection{Proof of Theorem \ref{thm:wqumau}} \label{proof:thm:wqumau}
	From Theorem 2 of \cite{choi1996asymptotically}, we know that $\left[\tilde{\theta}_p \pm \dfrac{z_{\alpha/2}}{\sqrt{n}}\hat{\sigma}_* \right]$ is an asymptotically uniformly most accurate unbiased confidence interval up to equivalence where $\hat{\sigma}_*$ is a consistent estimator of ${\sigma}_*$.
	In the proof of Theorem \ref{thm:wqCorrCovg}, we have shown that the quantiles $\hat{p}_1$ and $\hat{p}_2$ 
	obey equation (\ref{eq:quantileCondition}).
	In particular, the quantiles are $O(n^{-1/2})$ away from $p$. Thus, we can apply Proposition \ref{prop:wqexp} to each quantile when $\alpha_1 = \alpha_2 = \alpha/2$ and obtain \begin{align*}
		\tilde{\theta}_{\hat{p}_1} &= \tilde{\theta}_p +  \dfrac{z_{\alpha/2}}{\sqrt{n}}{\sigma}_* + o_p(n^{-1/2}), \text{ and} \\
		\tilde{\theta}_{\hat{p}_2} &= \tilde{\theta}_p +  \dfrac{z_{1-\alpha/2}}{\sqrt{n}}{\sigma}_* + o_p(n^{-1/2}).
	\end{align*} 
	Hence, we can see that the confidence interval from the weighted quantile method is equivalent to $\left[\tilde{\theta}_p \pm \dfrac{z_{\alpha/2}}{\sqrt{n}}\hat{\sigma}_* \right]$ and is asymptotically uniformly most accurate unbiased.

	{\tiny }\section{Additional numerical studies}
In this section, we provide additional details on the simulation studies and the deferred results for the simulation studies conducted in different settings. 
\subsection{Regression functions} \label{a.simul.fun}
\begin{enumerate}
	\item Step 
	\begin{align*}
	f(x) = 0.2 + 0.6 I_{(1/3, 2/3)}(x).
	\end{align*}
	\item 
	\begin{align*}
	f(x) = &(0.32 + 0.6x + 0.3\exp(-100(x-0.3)^2)) I_{(0, 0.8]}(x) + \\ &(-0.28 + 0.6x + 0.3\exp({-100(x-1.3)^2})) I_{(0.8, 1]}(x).
	\end{align*}
	\item Spikes
	\begin{align*}
	f(x) = &\exp({-500(x-0.23)^2}) + 2\exp({-2000(x-0.33)^2}) + \\ &4\exp({-8000(x-0.47)^2} )+ 3\exp({-16000(x-0.69)^2}) + \\ &\exp({-32000(x-0.83)^2}).
	\end{align*}
	\item Bumps
	\begin{align*}
	f(x) = \sum_{j=1}^{11} \dfrac{h_j}{1+ |(x-t_j)/w_j|^4}
	\end{align*}
	for 
	\begin{align*}
	t &= (0.1, 0.13, 0.15, 0.23, 0.25, 0.4, 0.44, 0.65, 0.76, 0.78, 0.81), \\
	w &= (0.005,0.005, 0.006,0.01,0.01,0.03,0.01,0.01,0.005,0.008,0.005), \\
	h &= (4, 5, 3,4,5,4.2,2.1,4.3,3.1,5.1,4.2).
	\end{align*}
	\item Parabolas
	\begin{align*} f(x) = &0.8-30r(x,0.1)+60r(x,0.2)-30r(x,0.3)+500r(x,0.35)-\\&1000r(x,0.37)+ 
	1000r(x,0.41)-500r(x,0.43)+7.5r(x,0.5)-\\&15r(x,0.7)+7.5r(x,0.9),
	\end{align*}
	where $r(x,c) = (x-c)^2I_{(c, 1]}(x)$.
	\item Angles
	\begin{align*} f(x) = &(2x+0.5) I_{(0, 0.15]}(x) + (-12(x-0.15)+0.8) I_{(0.15, 2]}(x) + \\
	&0.2 I_{(0.2, 0.5]}(x) + (6(x-0.5)+0.2) I_{(0.5, 0.6]}(x) + \\
	&(-10(x-0.6)+0.8) I_{(0.6, 0.65]}(x) + \\
	&(-5(x-0.65)+0.3) I_{(0.65, 0.85]}(x) + \\
	&(2(x-0.85)+0.2)I_{(0.85, 1]}(x).
	\end{align*}
\end{enumerate}

Figure \ref{fig:targets} shows $\theta_{1/2}$ as a function of $x_0$ for a triangular kernel for which $h = 0.04$. 
\begin{figure}
	\vspace{-1.1\baselineskip}
	\begin{subfigure}{.5\textwidth}
		\centering
		\includegraphics[width=\linewidth]{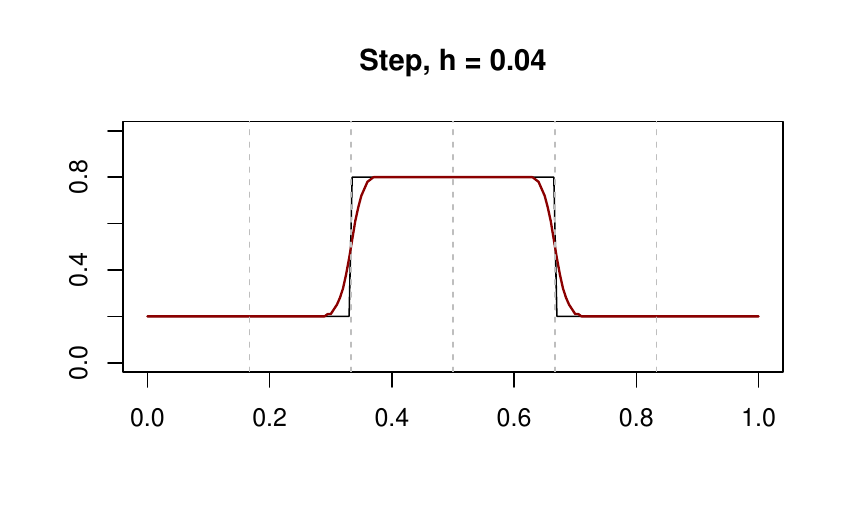}
		\vspace{-1.07\baselineskip}
	\end{subfigure}%
	\begin{subfigure}{.5\textwidth}
		\centering
		\includegraphics[width=\linewidth]{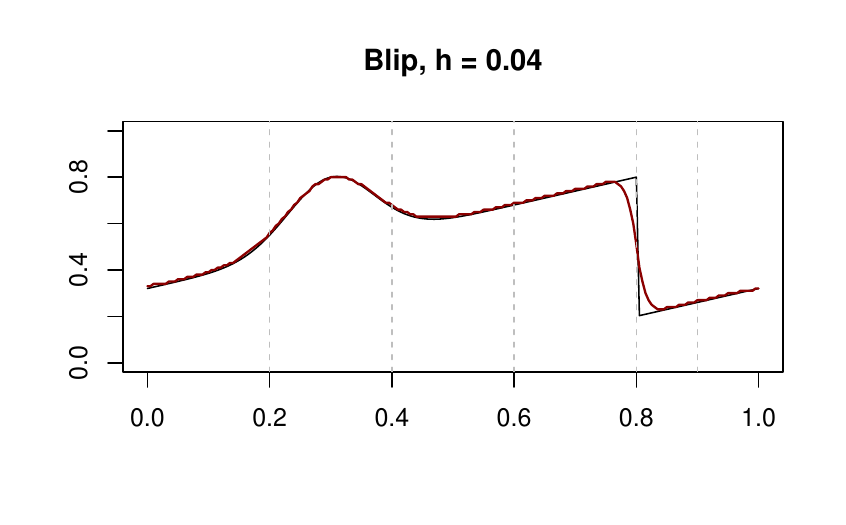}
		\vspace{-1.07\baselineskip}
	\end{subfigure}
	\vspace{-1.07\baselineskip}
	\begin{subfigure}{.5\textwidth}
		\centering
		\includegraphics[width=\linewidth]{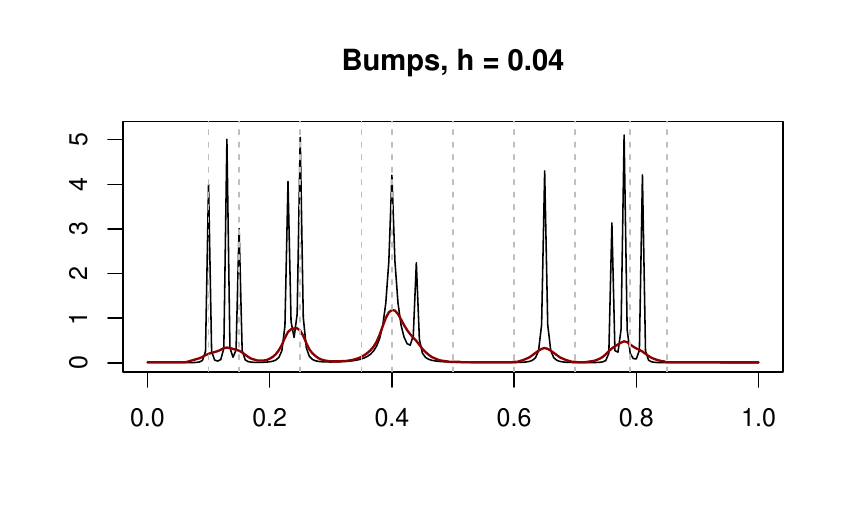}
	\end{subfigure}%
	\begin{subfigure}{.5\textwidth}
		\centering
		\includegraphics[width=\linewidth]{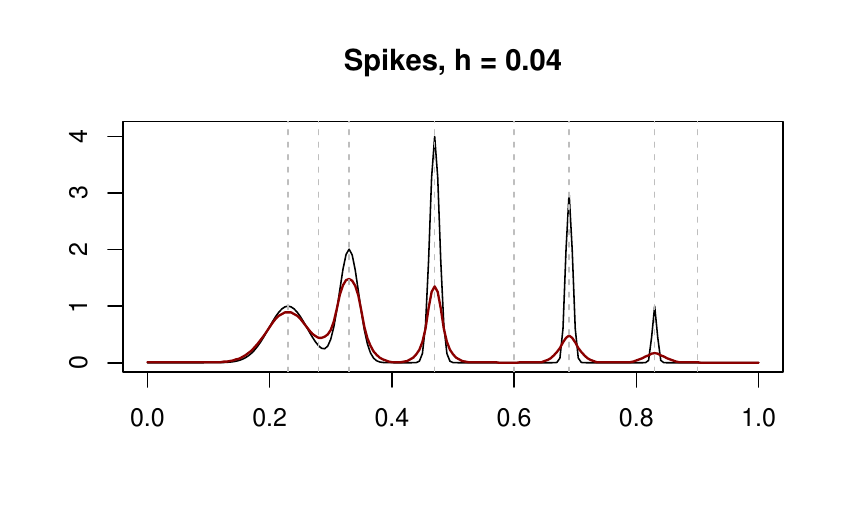}
	\end{subfigure}
	\vspace{-1.07\baselineskip}
	\begin{subfigure}{.5\textwidth}
		\centering
		\includegraphics[width=\linewidth]{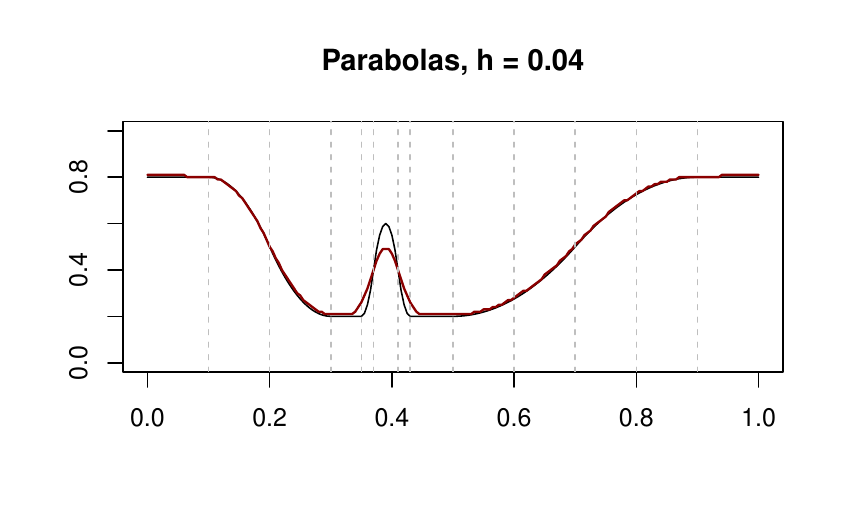}
	\end{subfigure}%
	\begin{subfigure}{.5\textwidth}
		\centering
		\includegraphics[width=\linewidth]{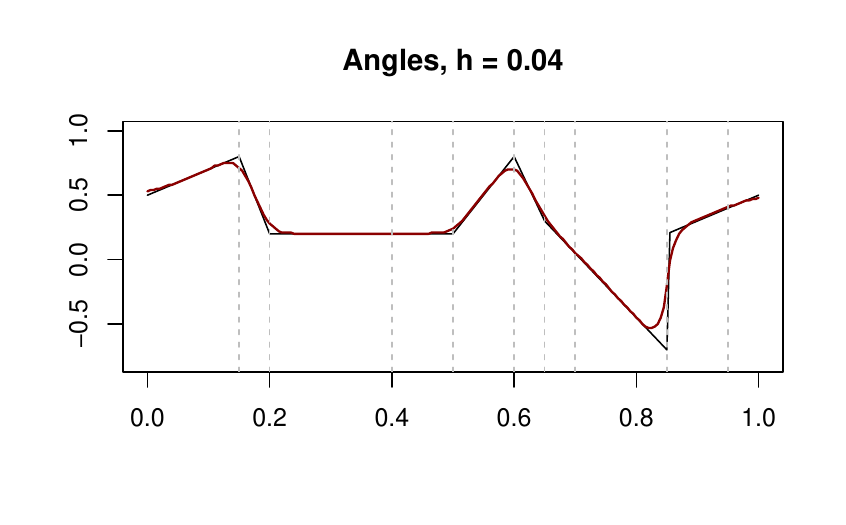}
	\end{subfigure}
	\caption{Regression functions (black) and the object of inference $\theta_{1/2}$ (dark red) when $h = 0.04$ in setting 1.}
	\label{fig:targets}
\end{figure}

\subsection{Additional results from Setting 1} \label{app:set1results}
\begin{figure}[H]
	\centering
	\includegraphics[width=\textwidth]{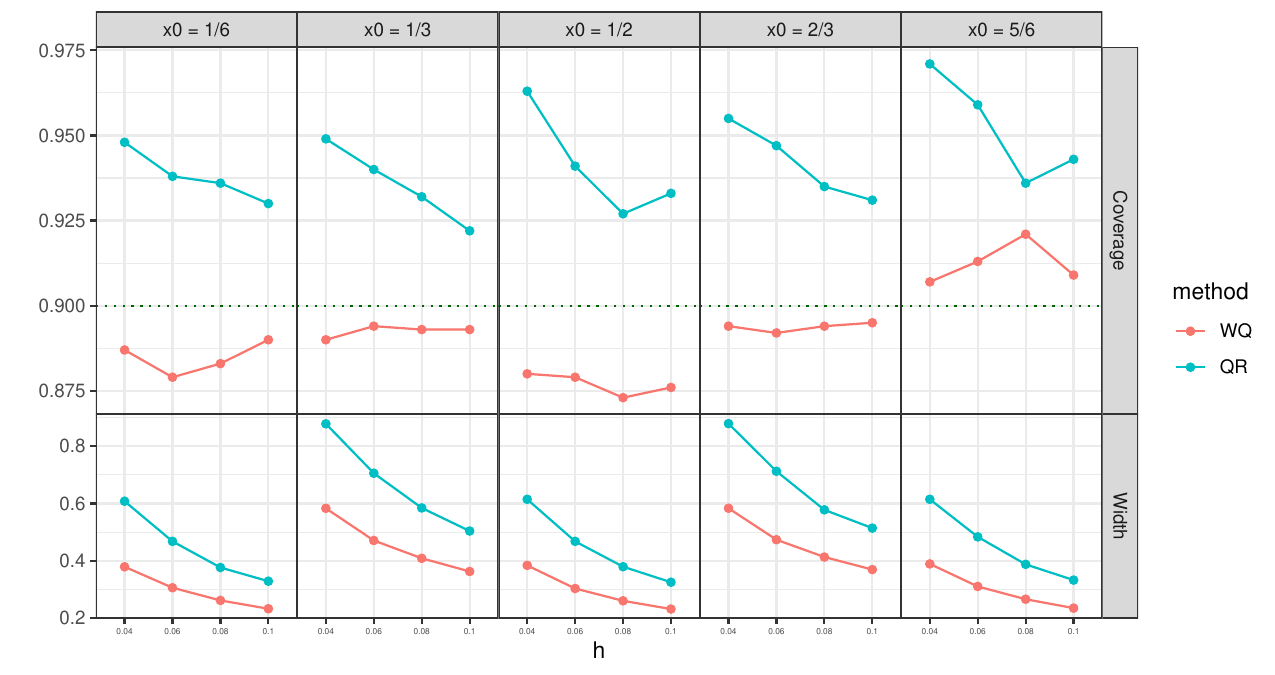}
	\caption{Coverage and width for the Step signal, setting 1.}
	\label{fig:step-s1-q0.5-tri}
\end{figure}

\begin{figure}[H]
	\centering
	\includegraphics[width=\textwidth]{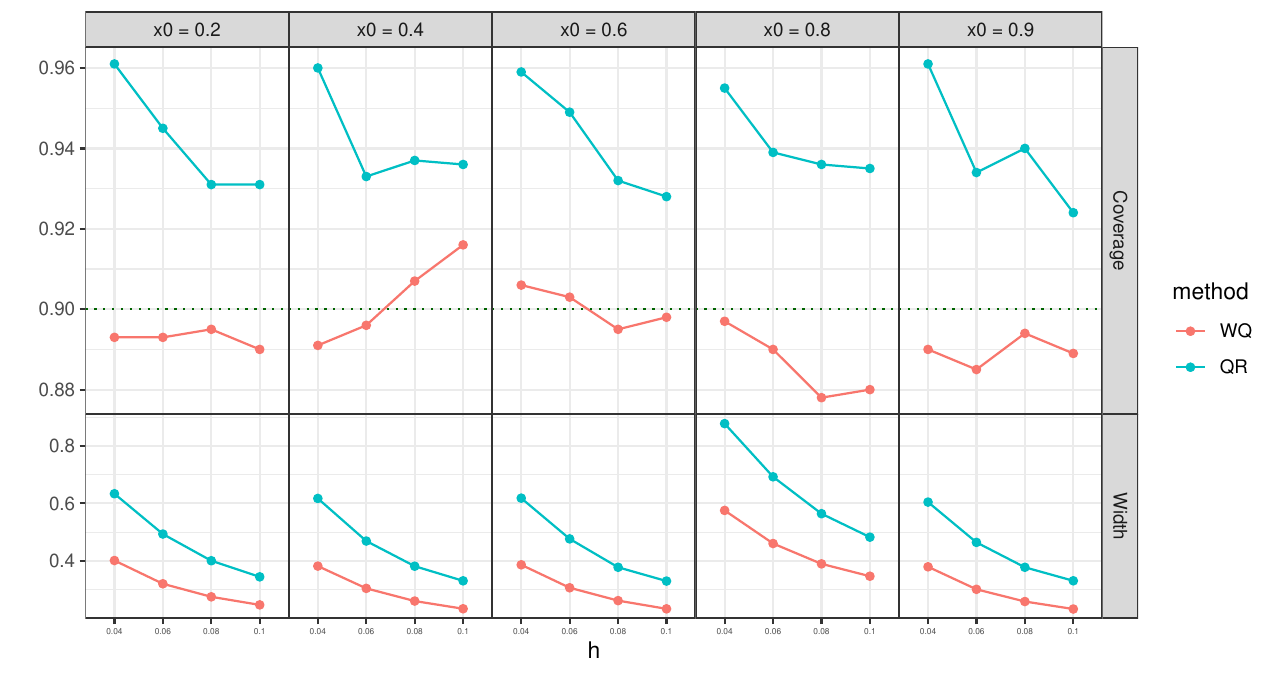}
	\caption{Coverage and width for the Blip signal, setting 1.}
	\label{fig:blip-s1-q0.5-tri}
\end{figure}

\begin{figure}[H]
	\centering
	\includegraphics[width=\textwidth]{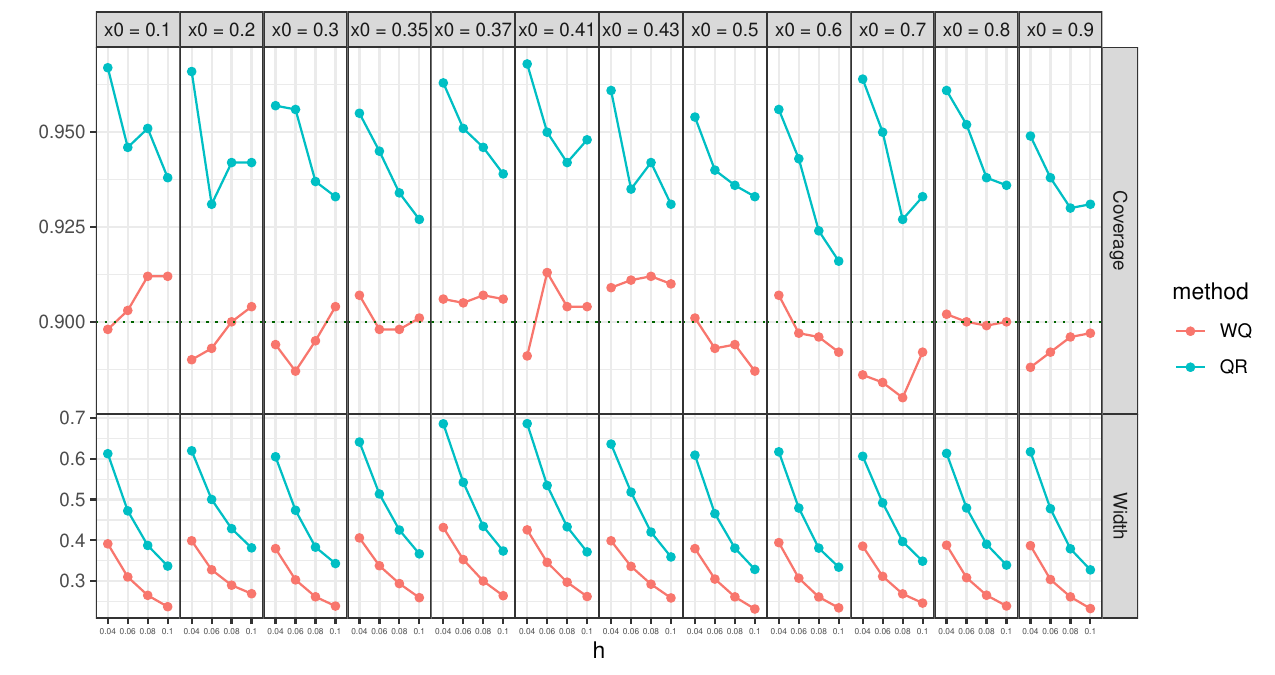}
	\caption{Coverage and width for the Parabolas signal, setting 1.}
	\label{fig:parabolas-s1-q0.5-tri}
\end{figure}

\begin{figure}[H]
	\centering
	\includegraphics[width=\textwidth]{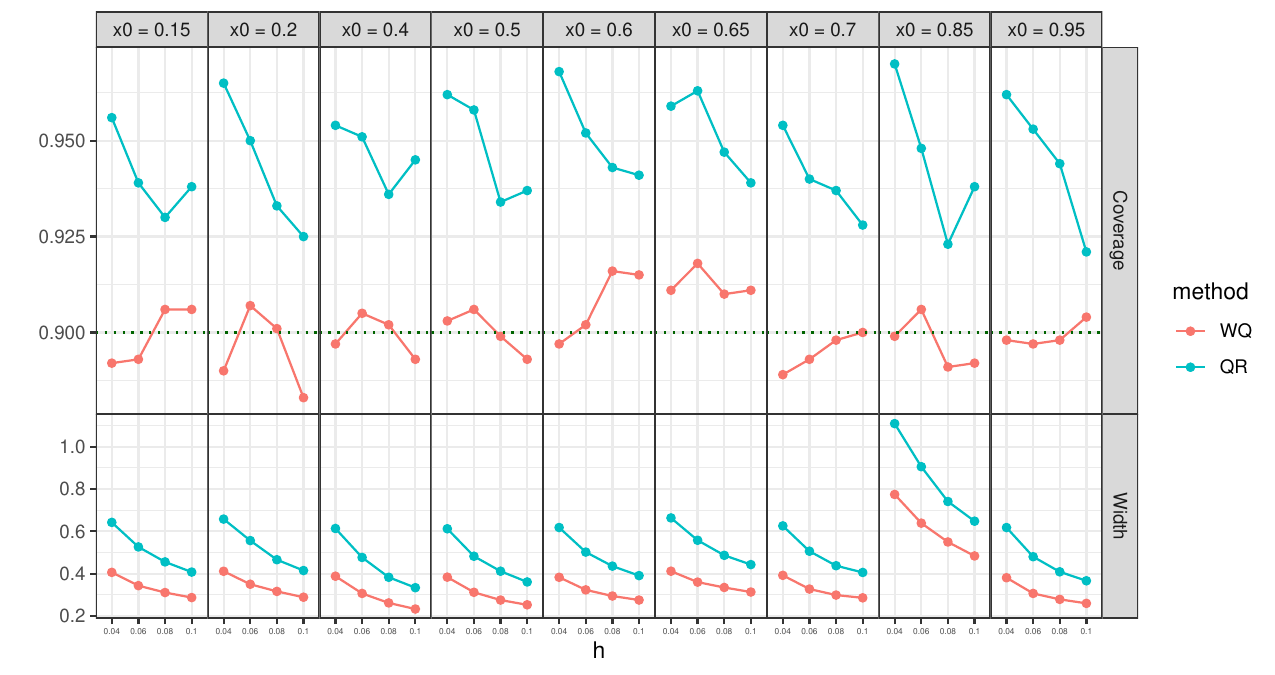}
	\caption{Coverage and width for the Angles signal, setting 1.}
	\label{fig:angles-s1-q0.5-tri}
\end{figure}

\subsection{Different quantiles}
\label{subsec:diff quantiles}
In this section, we show the empirical width and coverage using the proposed method for $\theta_p$ when $p = 0.2$ and $p = 0.7$ with the triangular kernel.
\subsubsection{$p = 0.2$}
\begin{figure}[H]
    \centering
    \includegraphics[width=\textwidth]{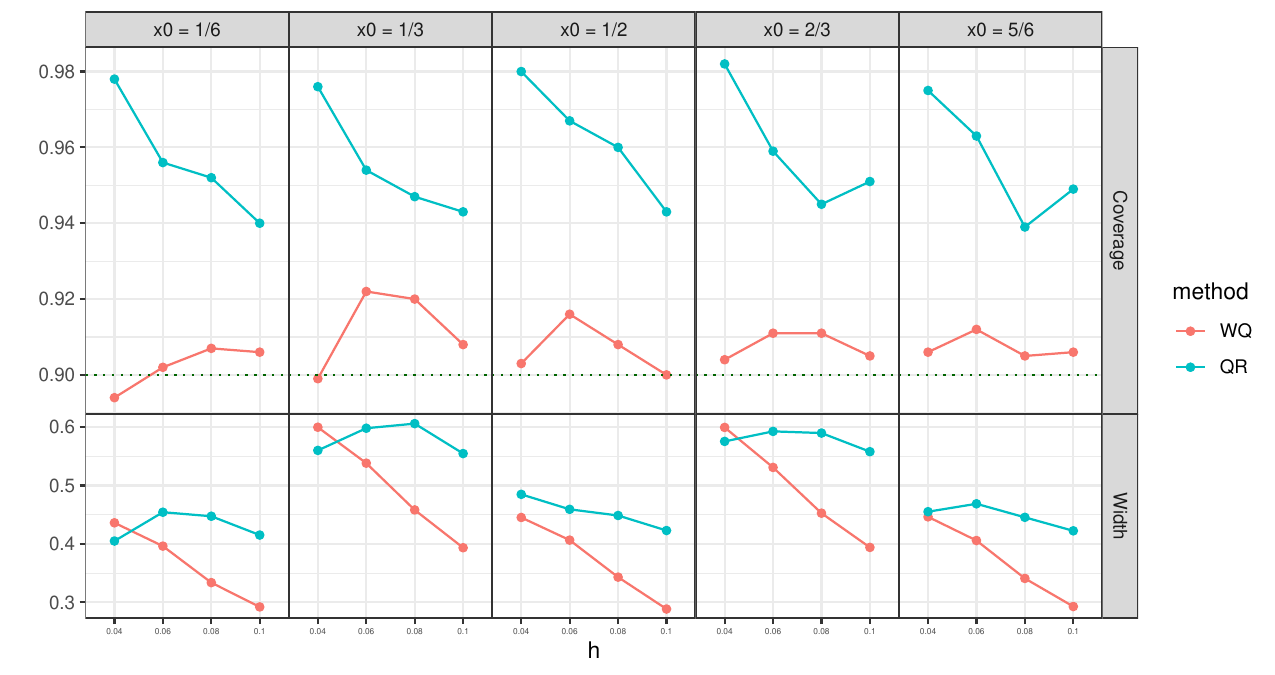}
    \caption{Coverage and width of $\theta_{0.2}$ for the Step signal, setting 1.}
    \label{fig:step-s1-q0.2-tri}
\end{figure}

\begin{figure}[H]
    \centering
    \includegraphics[width=\textwidth]{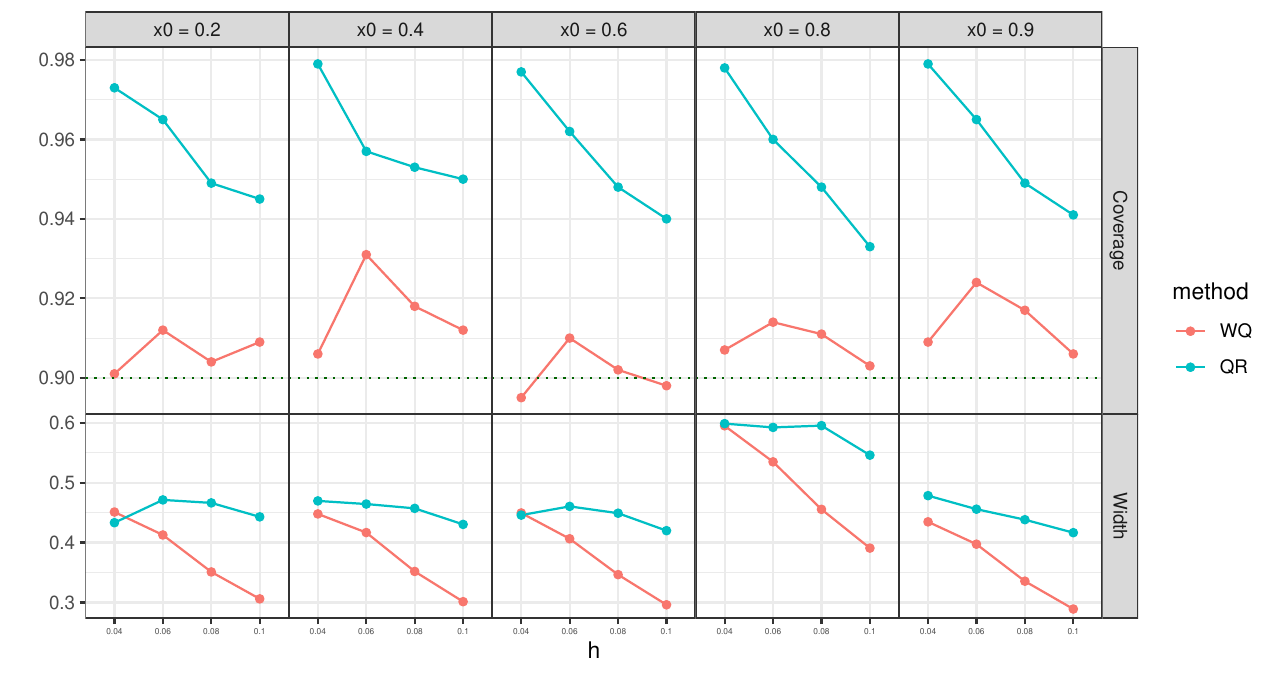}
    \caption{Coverage and width of $\theta_{0.2}$ for the Blip signal, setting 1.}
    \label{fig:blip-s1-q0.2-tri}
\end{figure}

\begin{figure}[H]
    \centering
    \includegraphics[width=\textwidth]{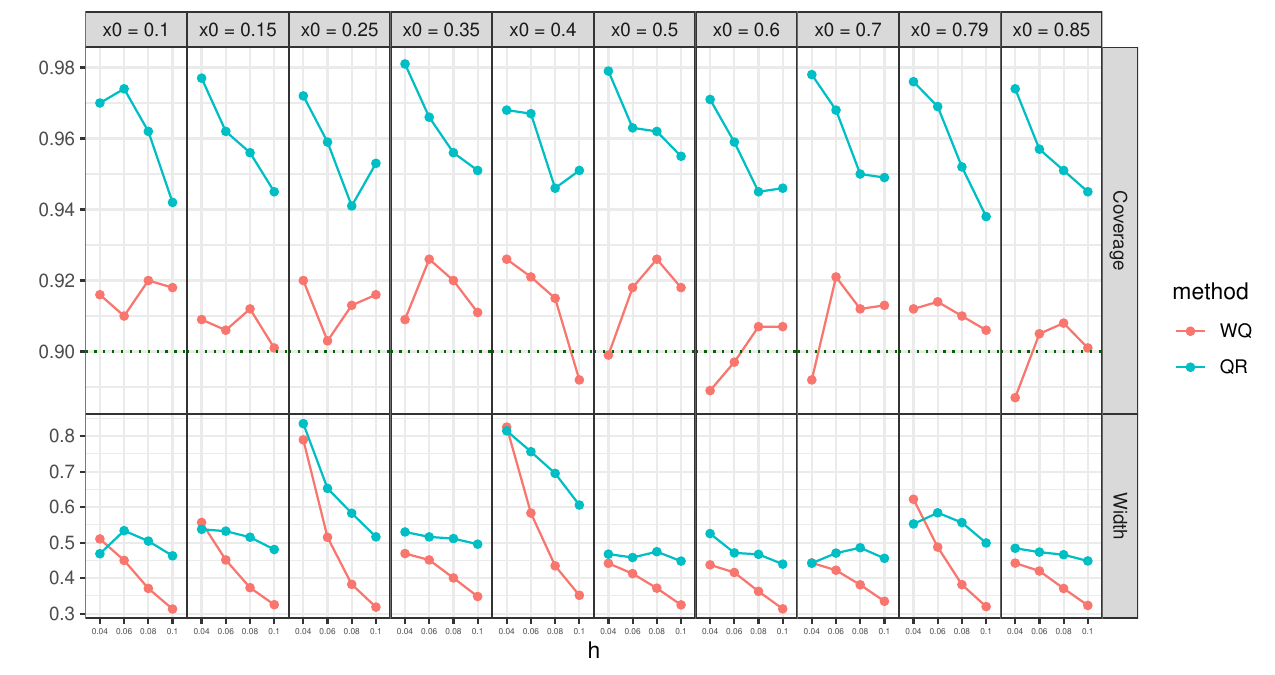}
    \caption{Coverage and width of $\theta_{0.2}$ for the Bump signal, setting 1.}
    \label{fig:bump-s1-q0.2-tri}
\end{figure}

\begin{figure}[H]
    \centering
    \includegraphics[width=\textwidth]{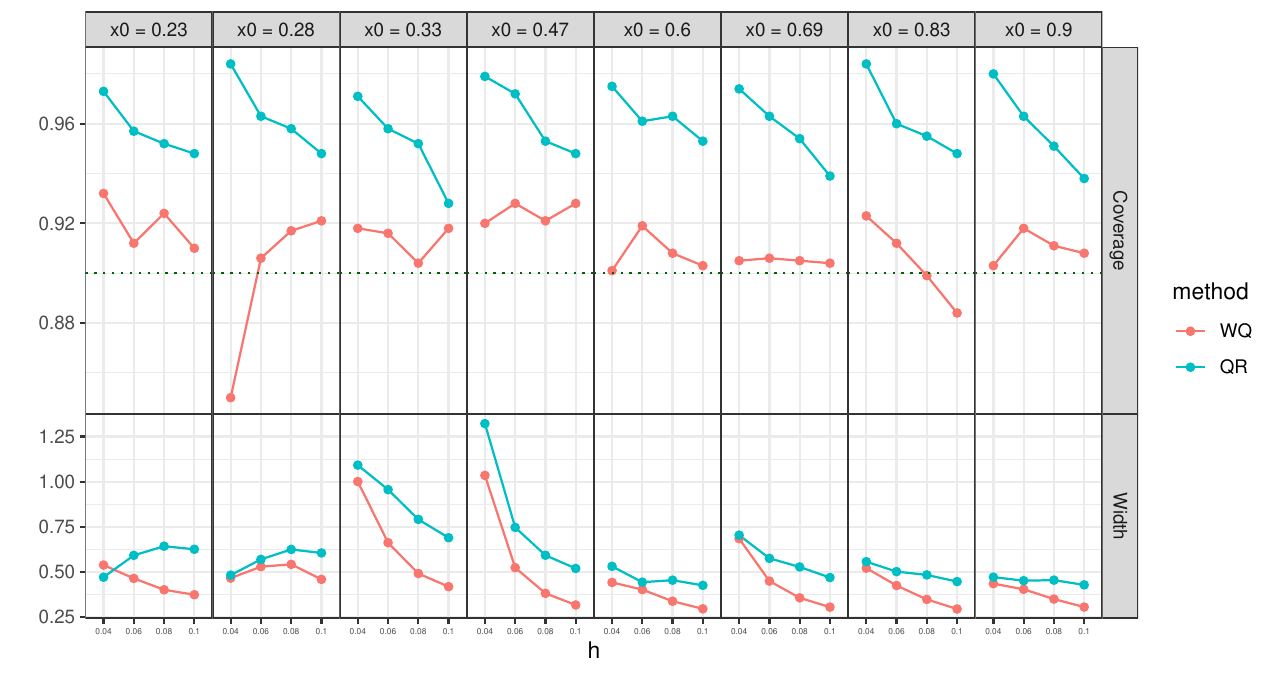}
    \caption{Coverage and width of $\theta_{0.2}$ for the Spikes signal, setting 1.}
    \label{fig:spikes-s1-q0.2-tri}
\end{figure}

\begin{figure}[H]
    \centering
    \includegraphics[width=\textwidth]{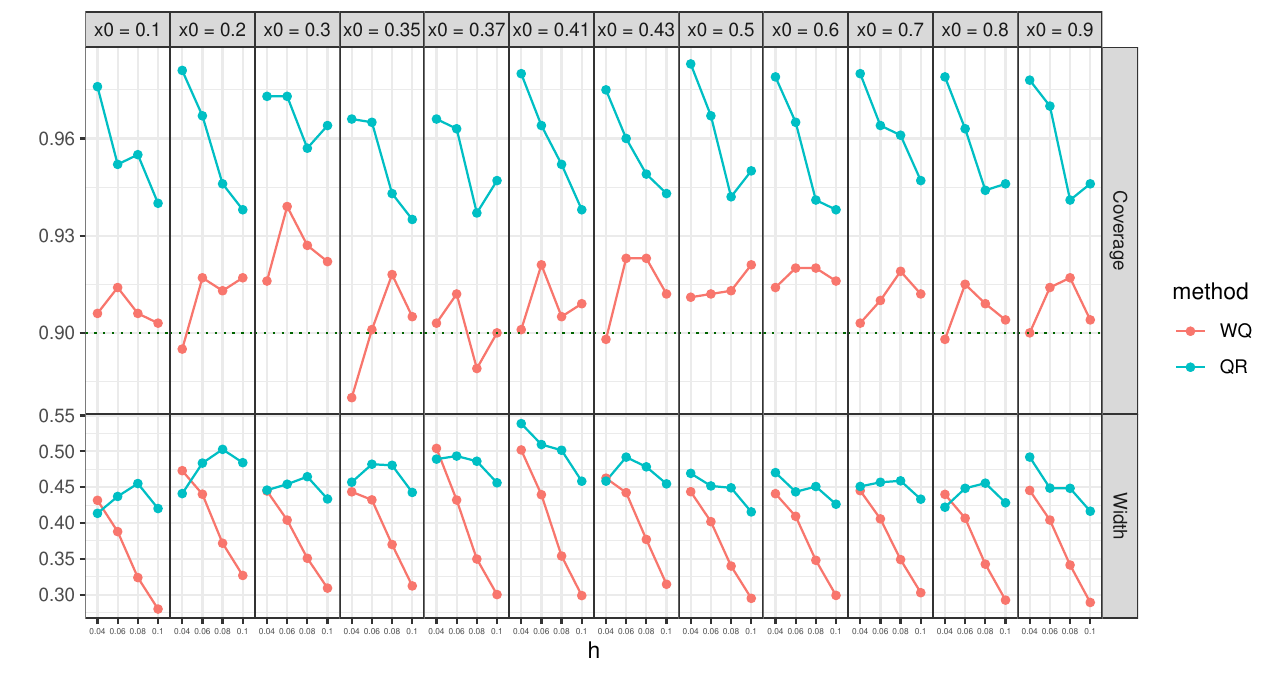}
    \caption{Coverage and width of $\theta_{0.2}$ for the Parabolas signal, setting 1.}
    \label{fig:parabolas-s1-q0.2-tri}
\end{figure}

\begin{figure}[H]
    \centering
    \includegraphics[width=\textwidth]{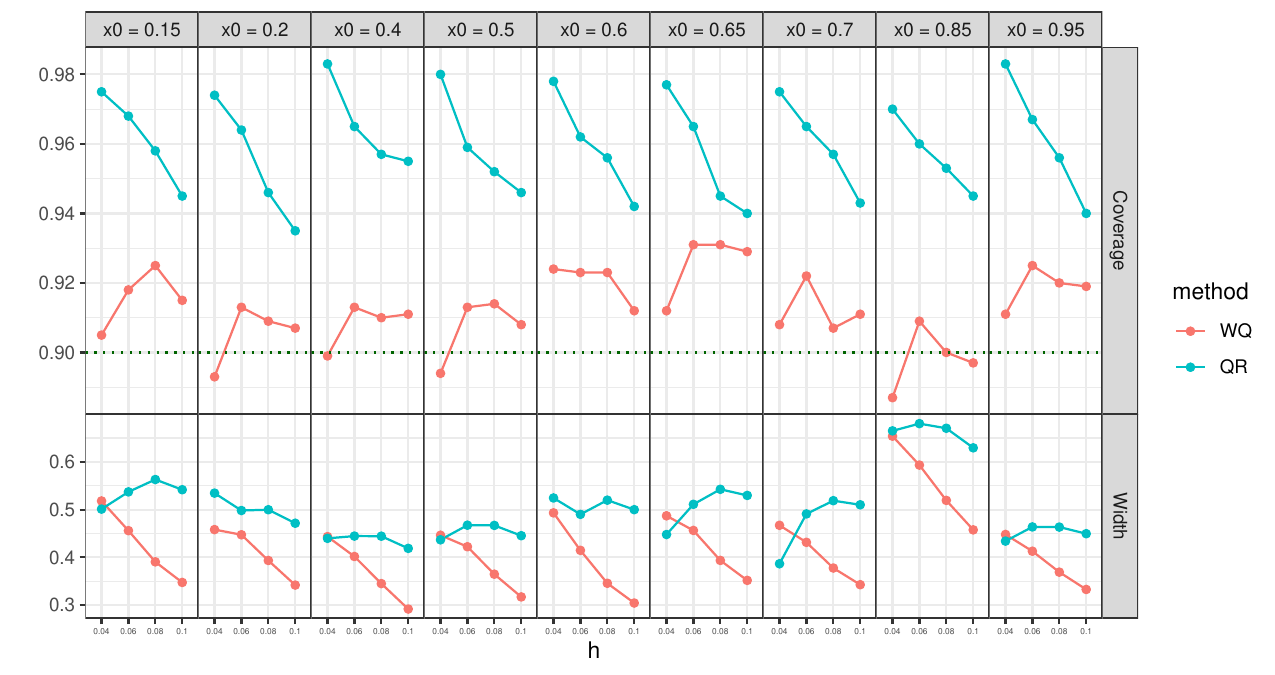}
    \caption{Coverage and width of $\theta_{0.2}$ for the Angles signal, setting 1.}
    \label{fig:angles-s1-q0.2-tri}
\end{figure}

\subsubsection{$p = 0.7$}

\begin{figure}[H]
    \centering
    \includegraphics[width=\textwidth]{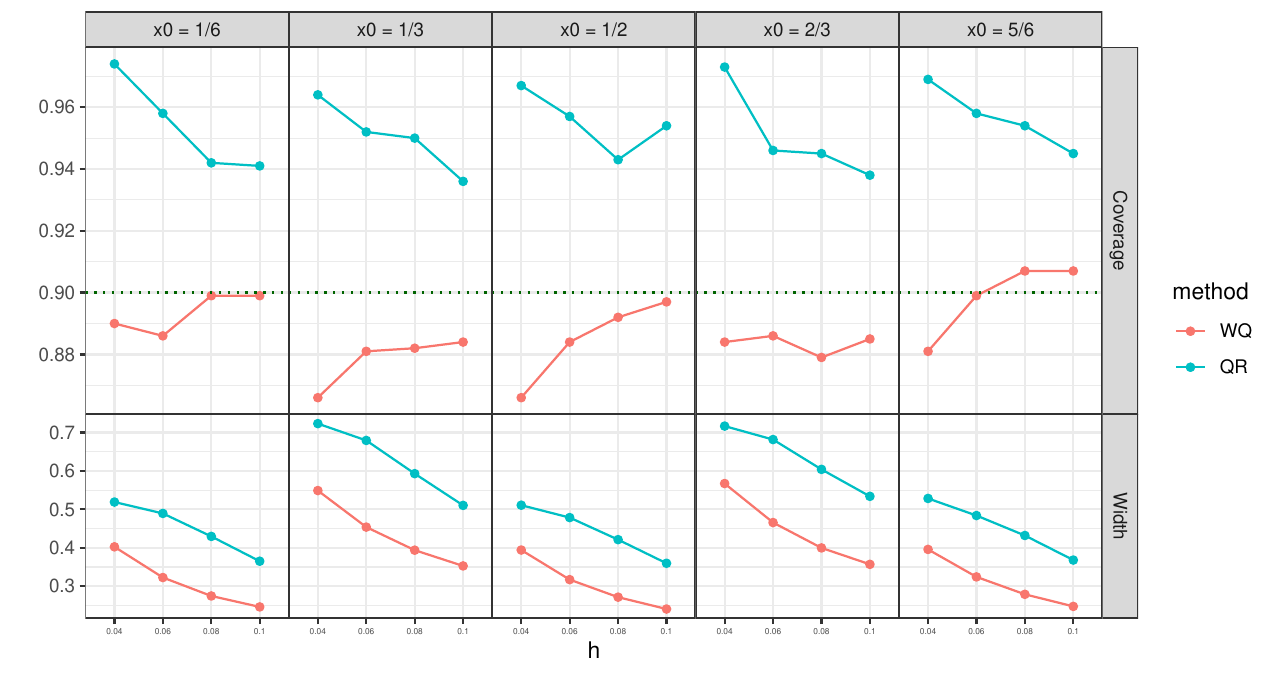}
    \caption{Coverage and width of $\theta_{0.7}$ for the Step signal, setting 1.}
    \label{fig:step-s1-q0.7-tri}
\end{figure}

\begin{figure}[H]
    \centering
    \includegraphics[width=\textwidth]{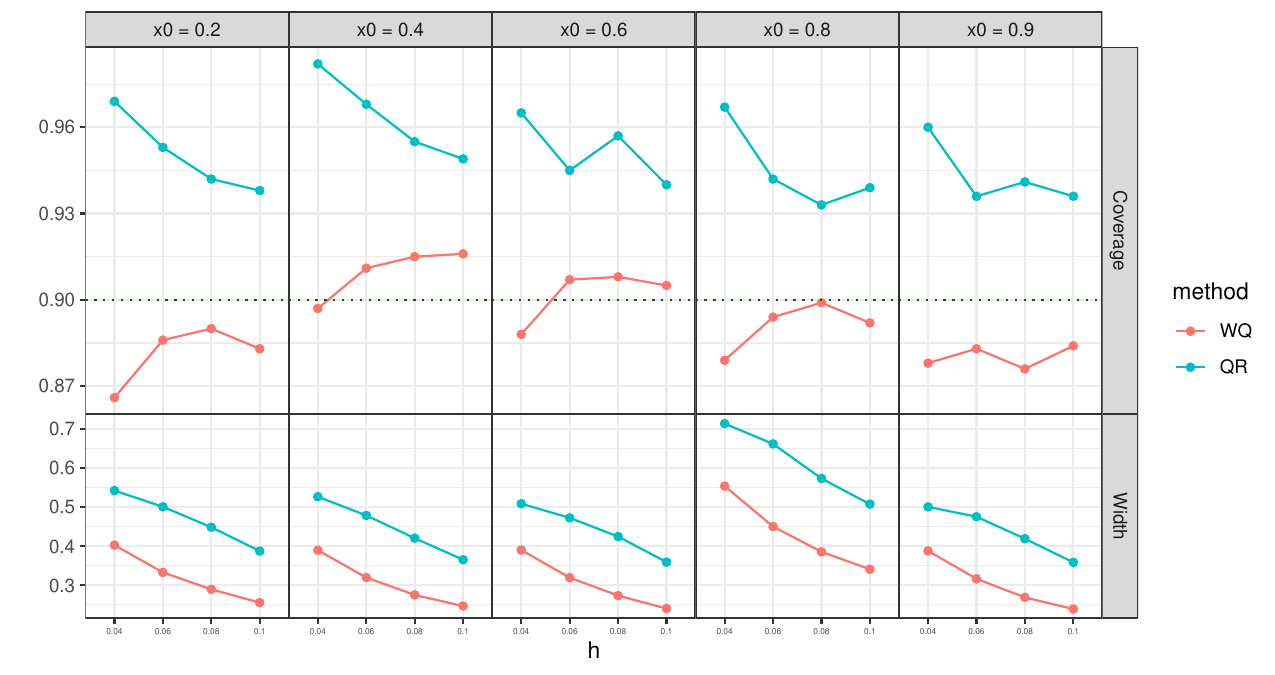}
    \caption{Coverage and width of $\theta_{0.7}$ for the Blip signal, setting 1.}
    \label{fig:blip-s1-q0.7-tri}
\end{figure}

\begin{figure}[H]
    \centering
    \includegraphics[width=\textwidth]{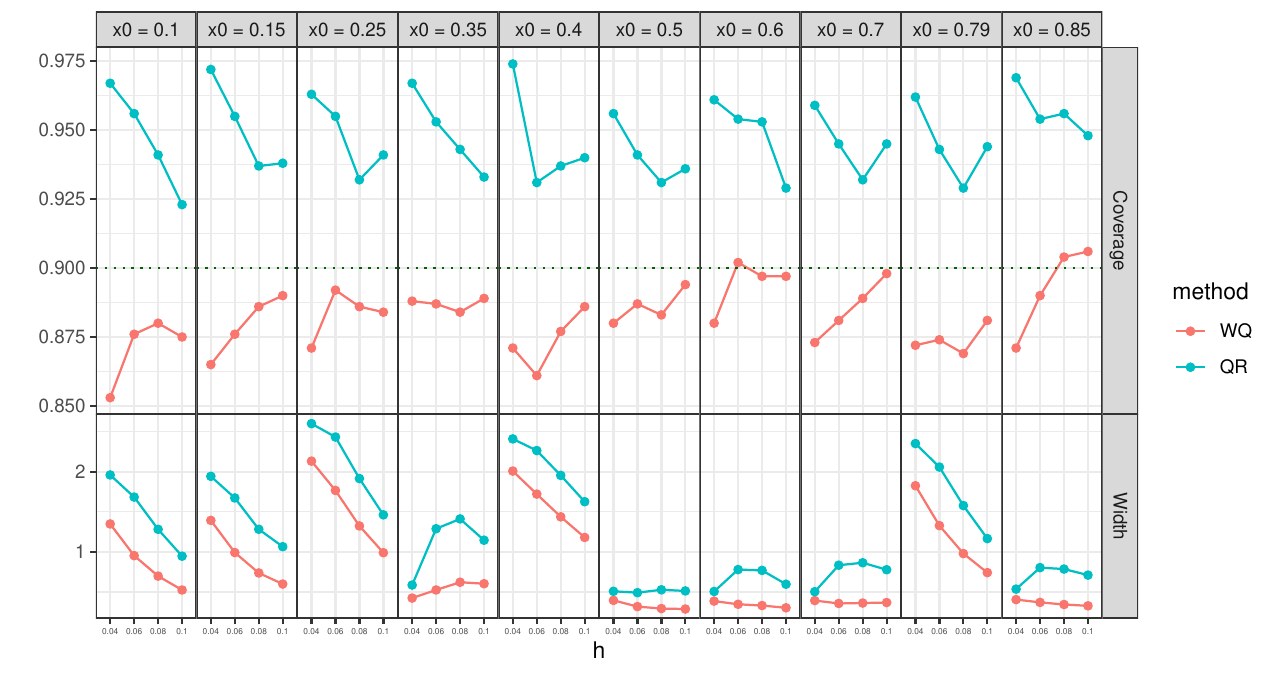}
    \caption{Coverage and width of $\theta_{0.7}$ for the Bump signal, setting 1.}
    \label{fig:bump-s1-q0.7-tri}
\end{figure}

\begin{figure}[H]
    \centering
    \includegraphics[width=\textwidth]{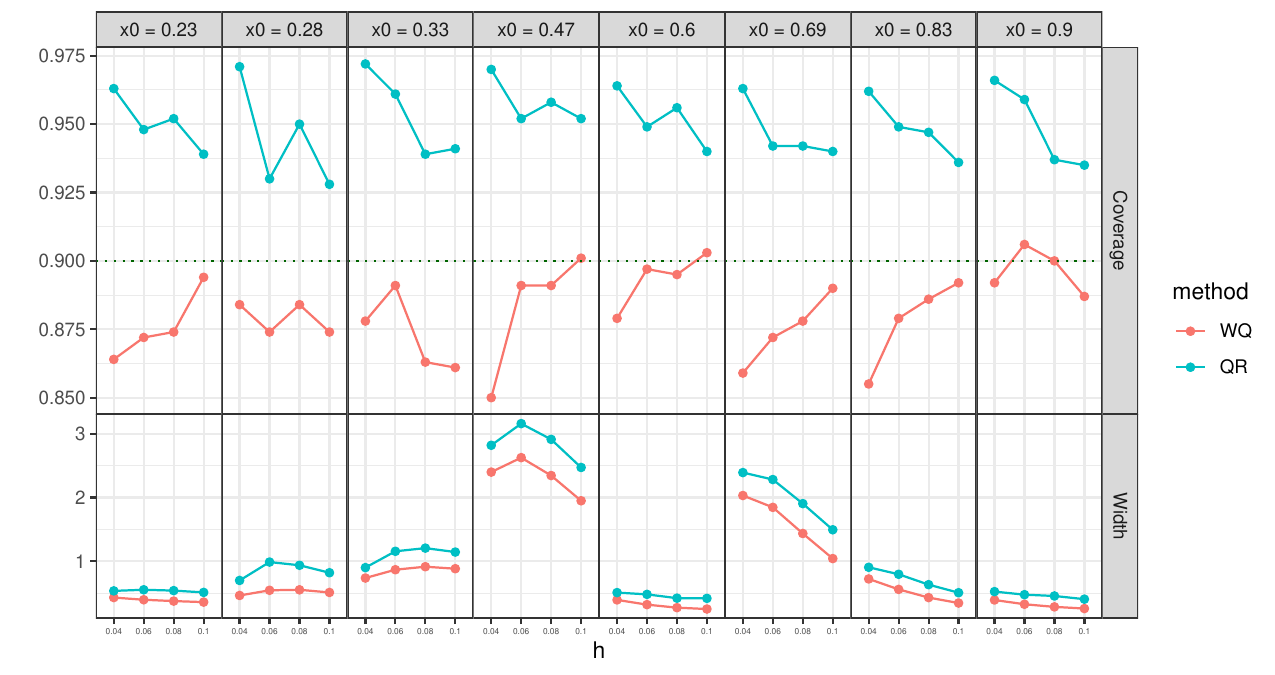}
    \caption{Coverage and width of $\theta_{0.7}$ for the Spikes signal, setting 1.}
    \label{fig:spikes-s1-q0.7-tri}
\end{figure}

\begin{figure}[H]
    \centering
    \includegraphics[width=\textwidth]{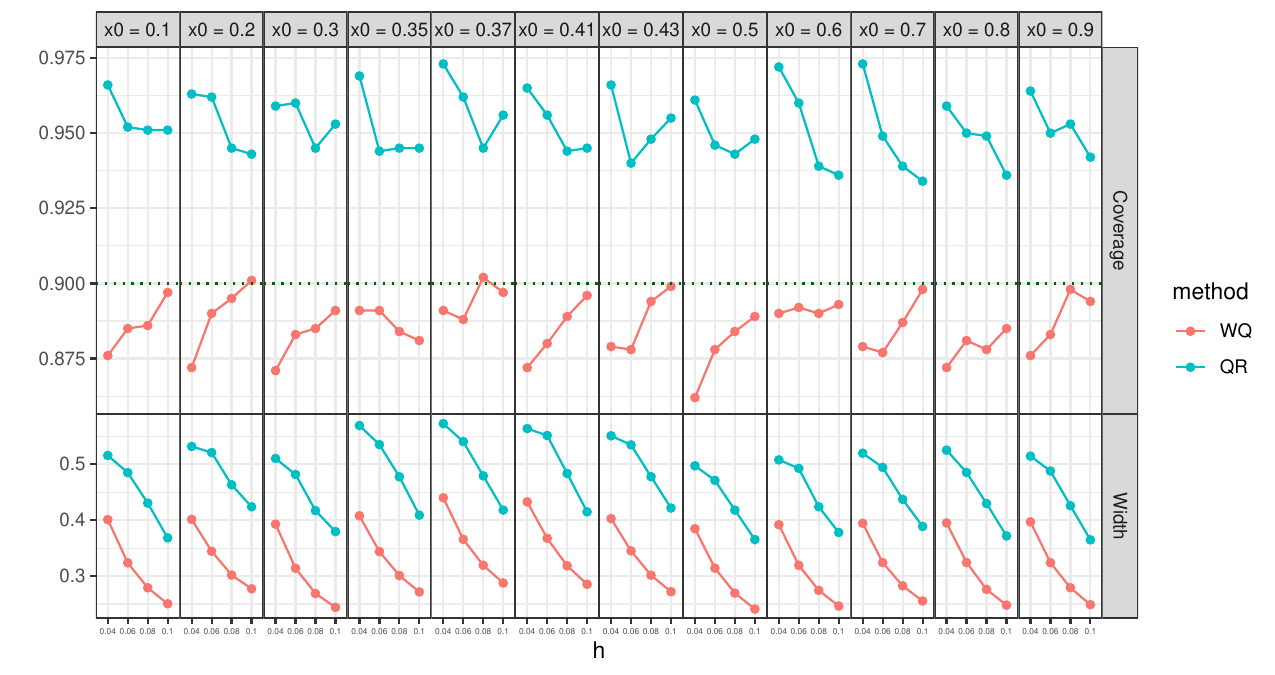}
    \caption{Coverage and width of $\theta_{0.7}$ for the Parabolas signal, setting 1.}
    \label{fig:parabolas-s1-q0.7-tri}
\end{figure}

\begin{figure}[H]
    \centering
    \includegraphics[width=\textwidth]{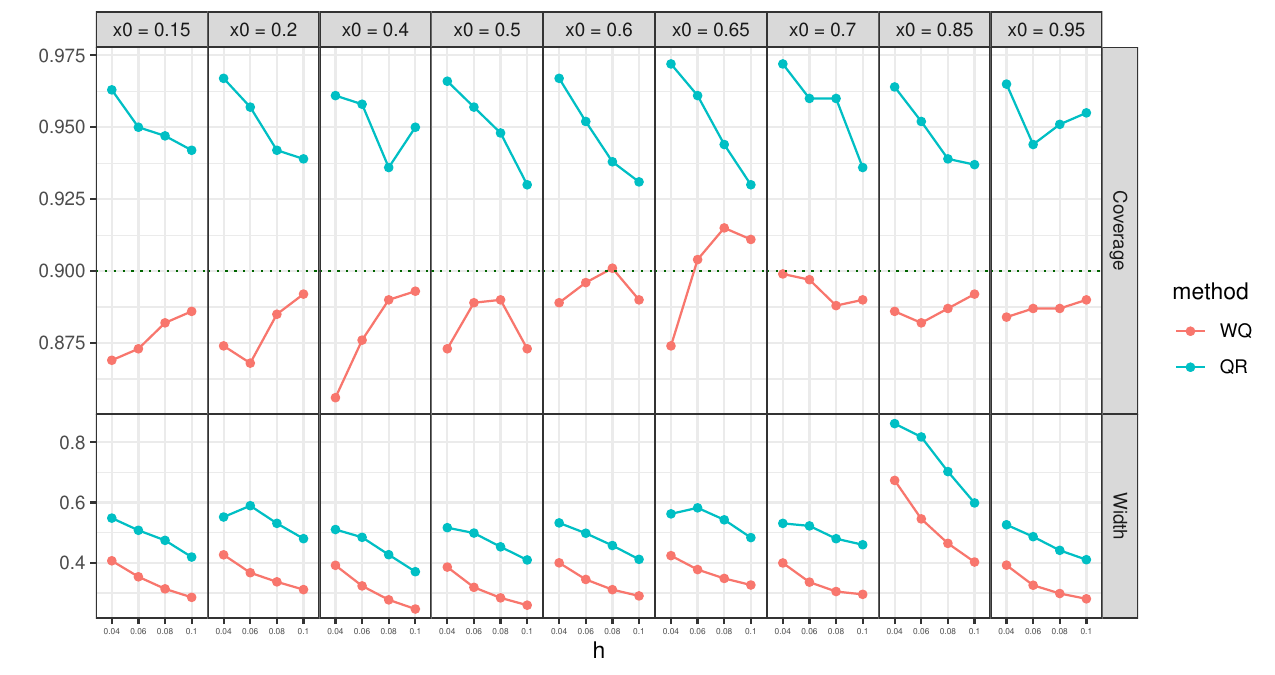}
    \caption{Coverage and width of $\theta_{0.7}$ for the Angles signal, setting 1.}
    \label{fig:angles-s1-q0.7-tri}
\end{figure}

\subsection{Heteroscedastic case}
\label{subsec:hetero}
In this section, we show the empirical width and coverage using the proposed method for $\theta_{0.5}$ when the noise distribution is heteroscedastic with the triangular kernel.
\subsubsection{Setting 2}
We consider the setting where the noise distribution is $\epsilon \mid X \sim  \mathcal{N}(0, (0.3(X^2+1))^2)$.
\begin{figure}[H]
    \centering
    \includegraphics[width=\textwidth]{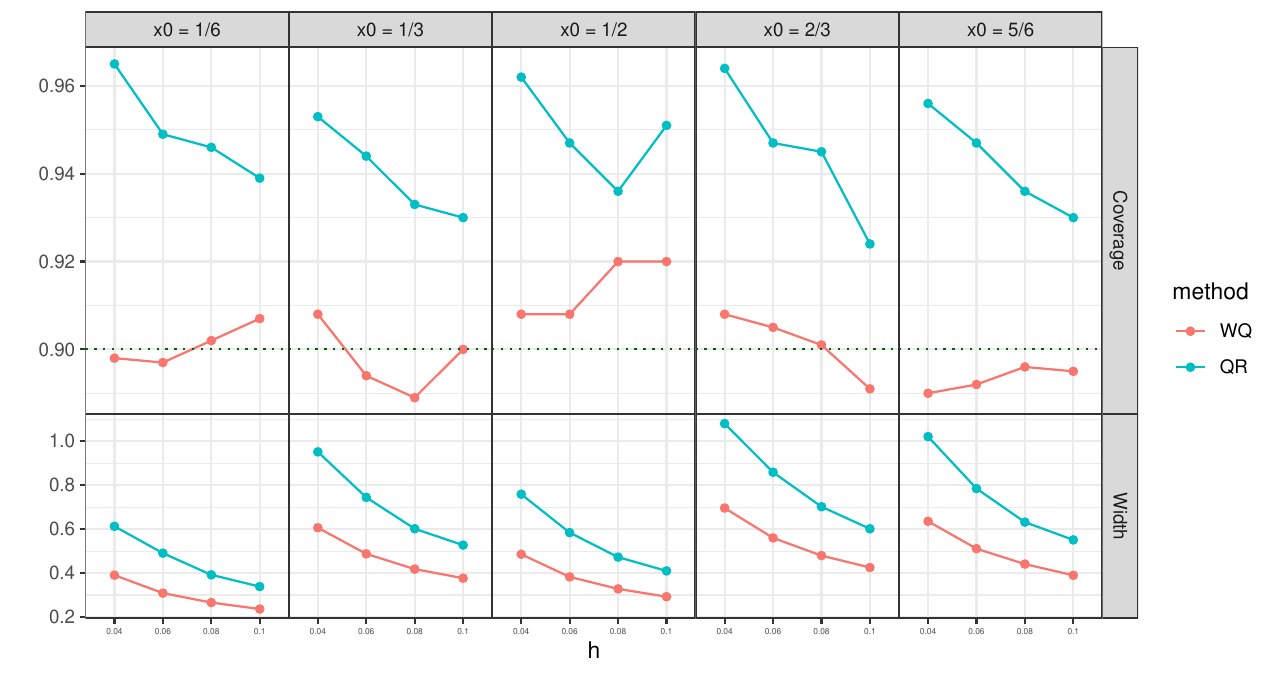}
    \caption{Coverage and width for the Step signal, setting 2.}
    \label{fig:step-s2-q0.5-tri}
\end{figure}

\begin{figure}[H]
    \centering
    \includegraphics[width=\textwidth]{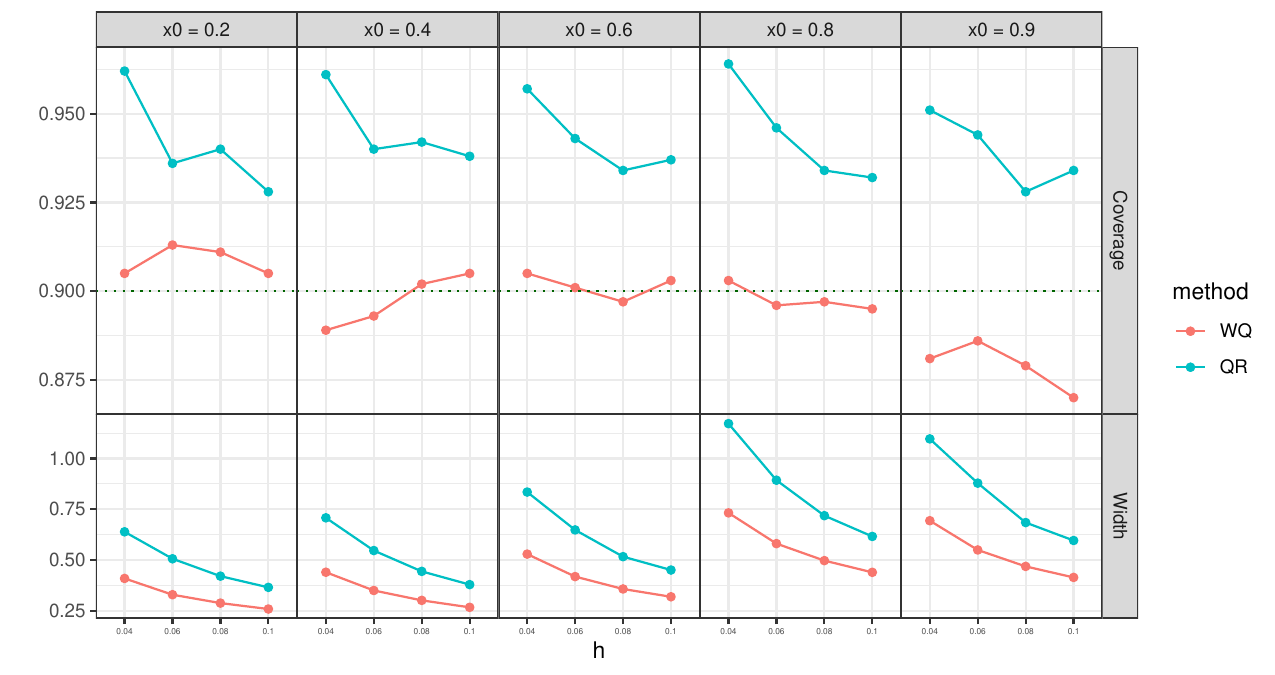}
    \caption{Coverage and width for the Blip signal, setting 2.}
    \label{fig:blip-s2-q0.5-tri}
\end{figure}

\begin{figure}[H]
    \centering
    \includegraphics[width=\textwidth]{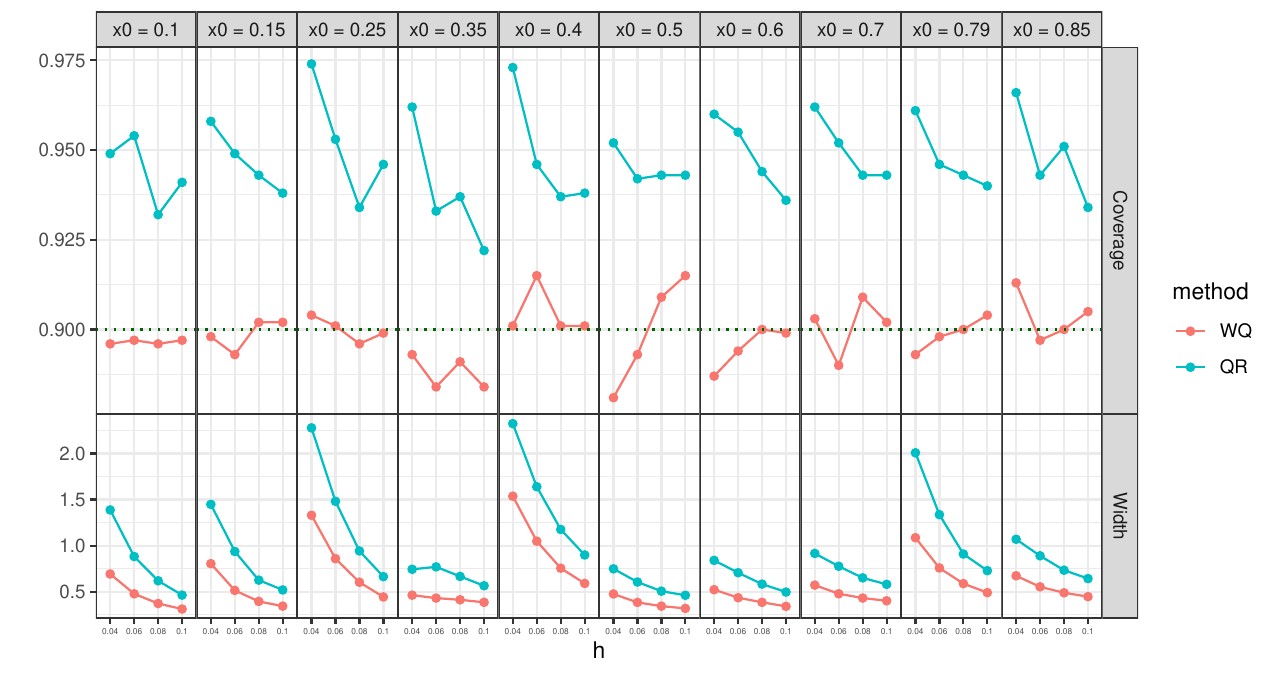}
    \caption{Coverage and width for the Bump signal, setting 2.}
    \label{fig:bump-s2-q0.5-tri}
\end{figure}

\begin{figure}[H]
    \centering
    \includegraphics[width=\textwidth]{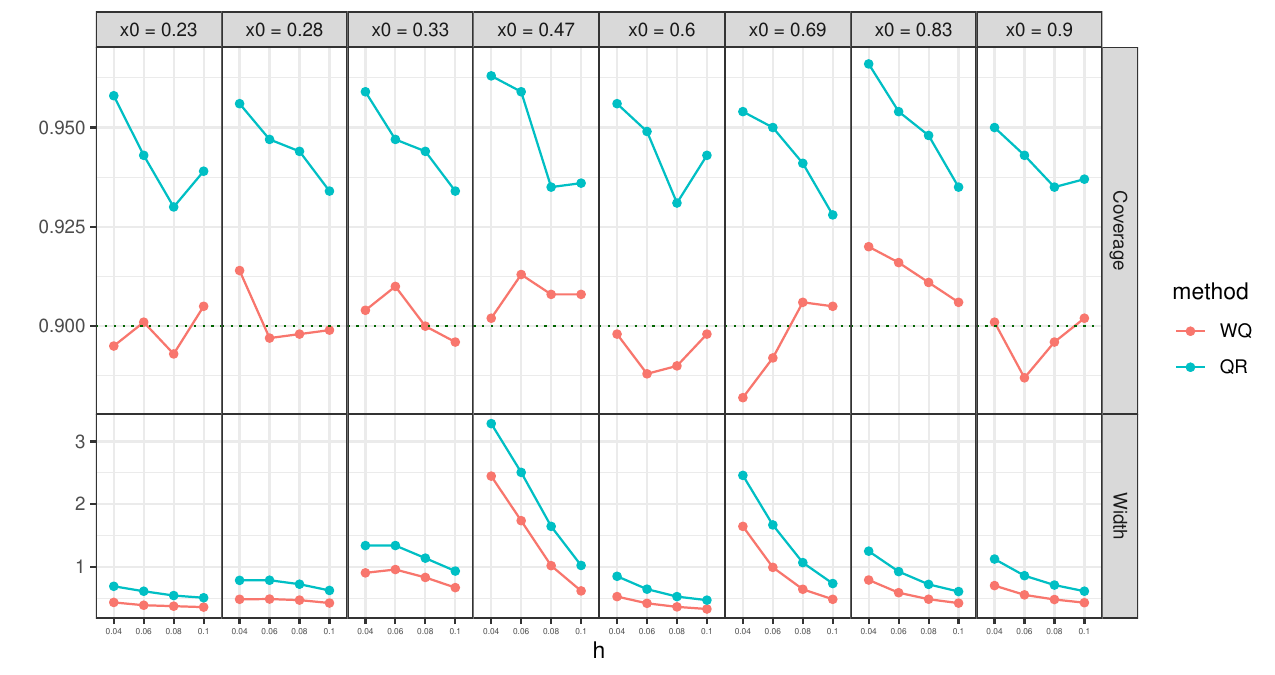}
    \caption{Coverage and width for the Spikes signal, setting 2.}
    \label{fig:spikes-s2-q0.5-tri}
\end{figure}

\begin{figure}[H]
    \centering
    \includegraphics[width=\textwidth]{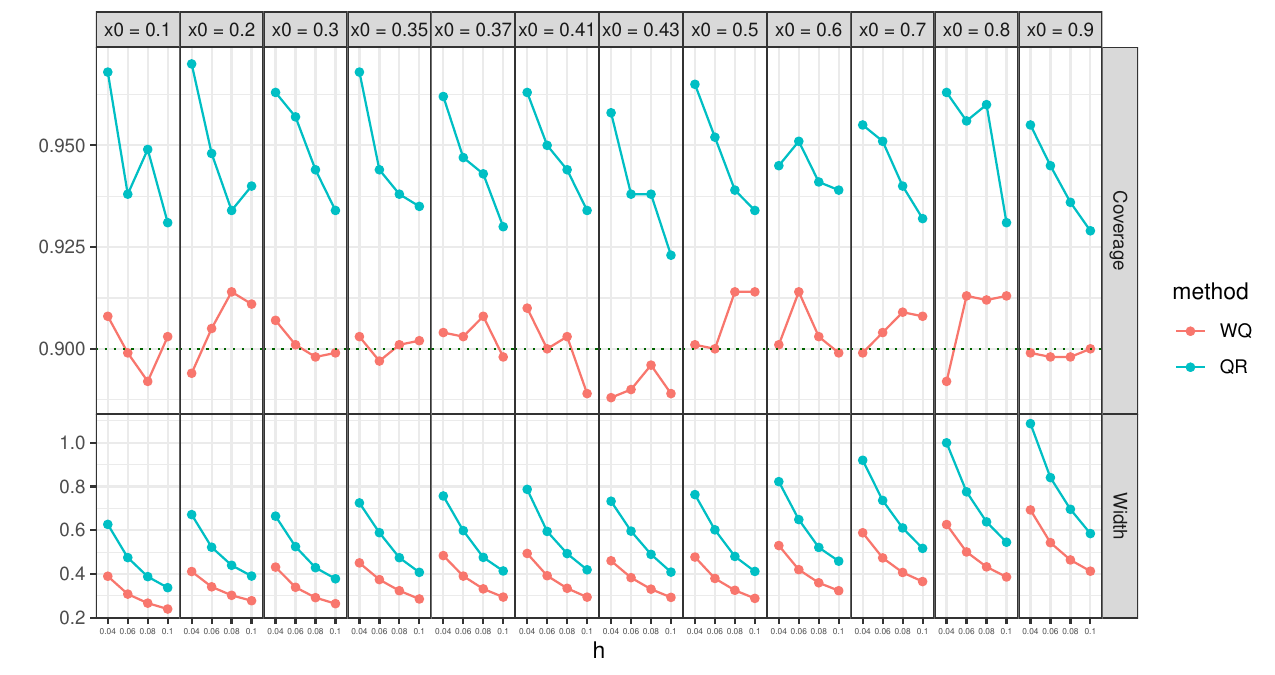}
    \caption{Coverage and width for the Parabolas signal, setting 2.}
    \label{fig:parabolas-s2-q0.5-tri}
\end{figure}

\begin{figure}[H]
    \centering
    \includegraphics[width=\textwidth]{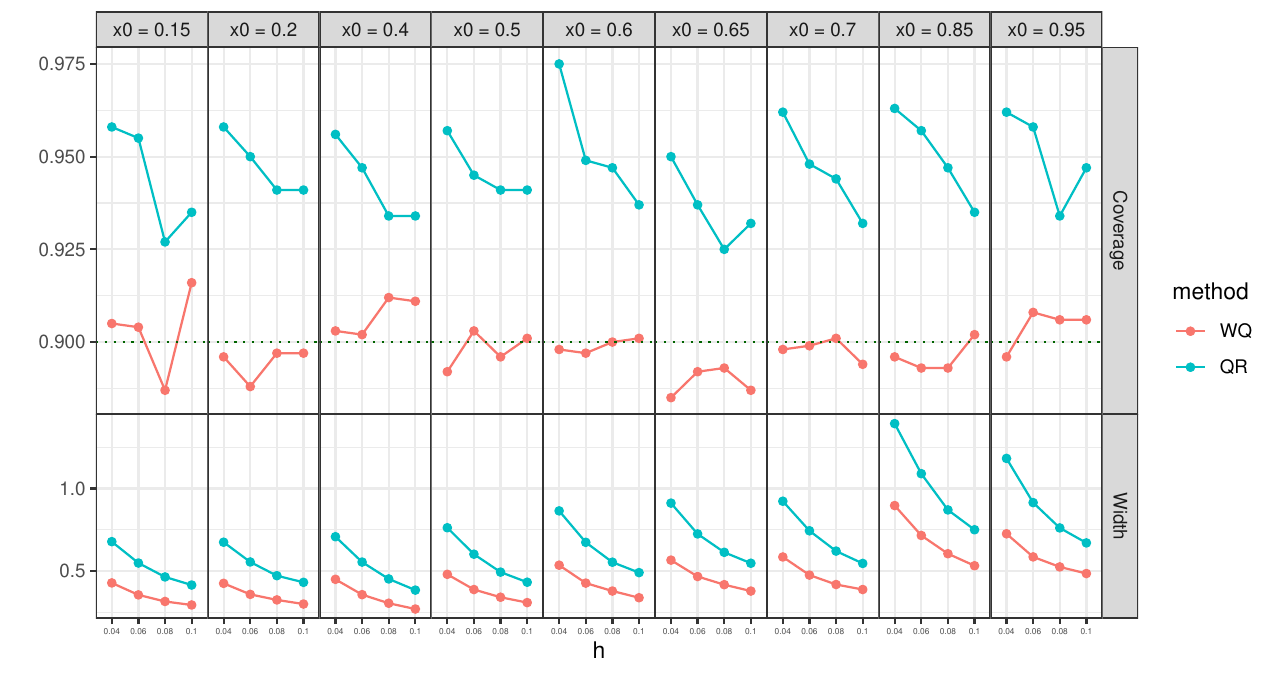}
    \caption{Coverage and width for the Angles signal, setting 2.}
    \label{fig:angles-s2-q0.5-tri}
\end{figure}

\subsubsection{Setting 3}
We consider the setting where the noise distribution is $\epsilon \mid X \sim  \mathcal{N}(0, (0.3(X^2-X+5/4))^2)$.
\begin{figure}[H]
    \centering
    \includegraphics[width=\textwidth]{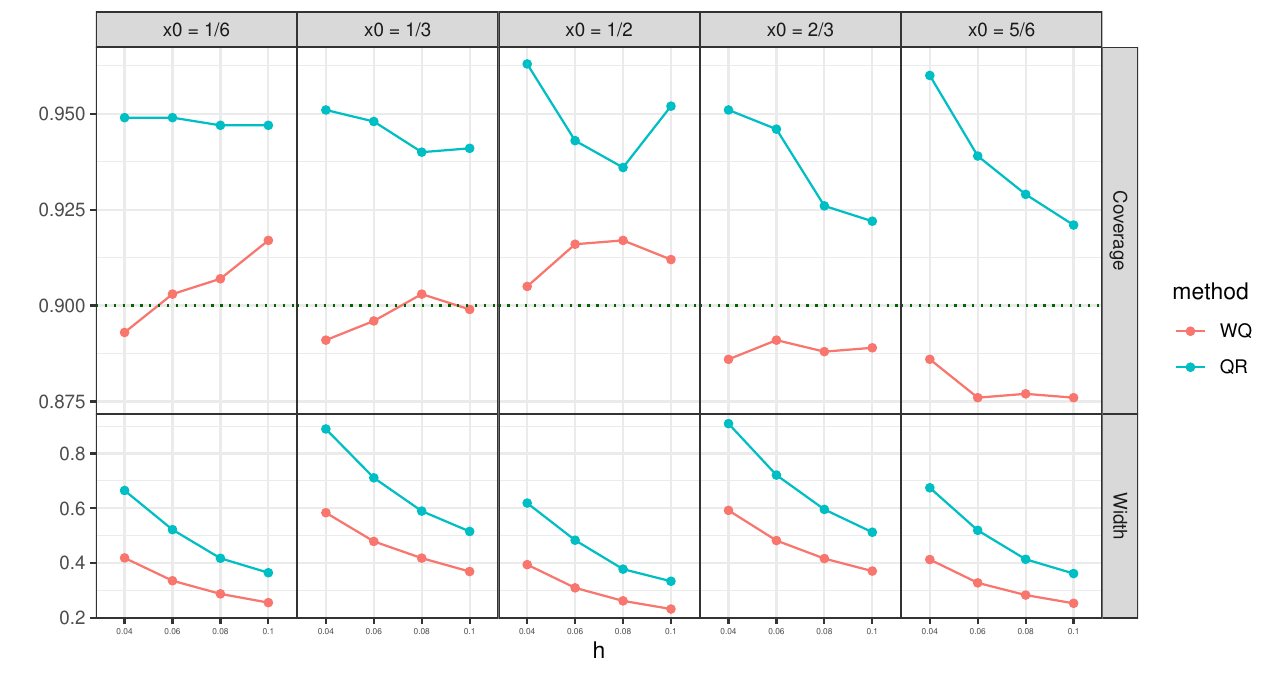}
    \caption{Coverage and width for the Step signal, setting 3.}
    \label{fig:step-s3-q0.5-tri}
\end{figure}

\begin{figure}[H]
    \centering
    \includegraphics[width=\textwidth]{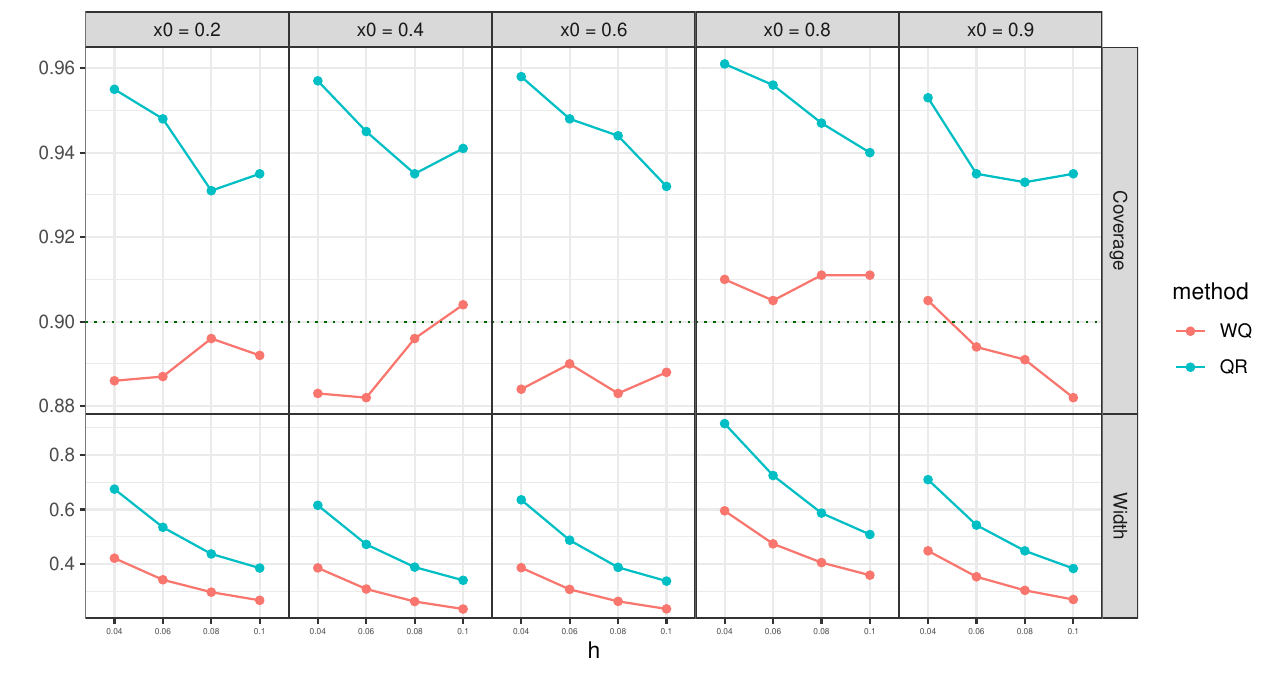}
    \caption{Coverage and width for the Blip signal, setting 3.}
    \label{fig:blip-s3-q0.5-tri}
\end{figure}

\begin{figure}[H]
    \centering
    \includegraphics[width=\textwidth]{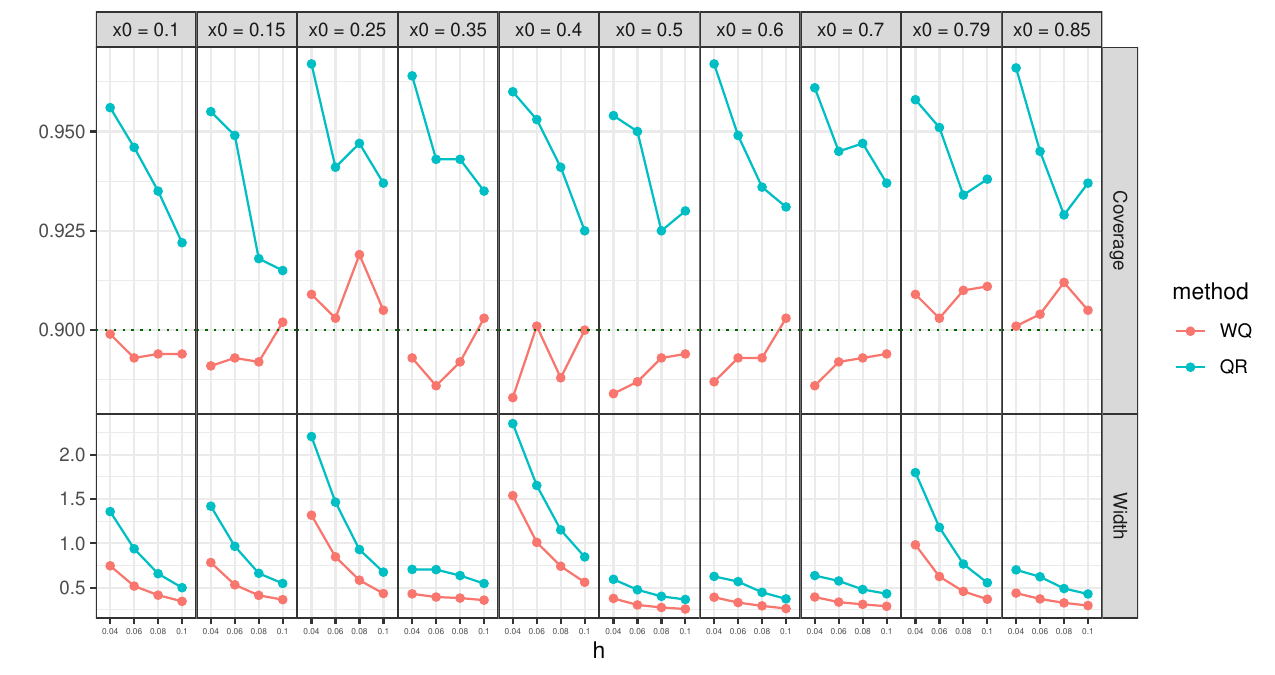}
    \caption{Coverage and width for the Bump signal, setting 3.}
    \label{fig:bump-s3-q0.5-tri}
\end{figure}

\begin{figure}[H]
    \centering
    \includegraphics[width=\textwidth]{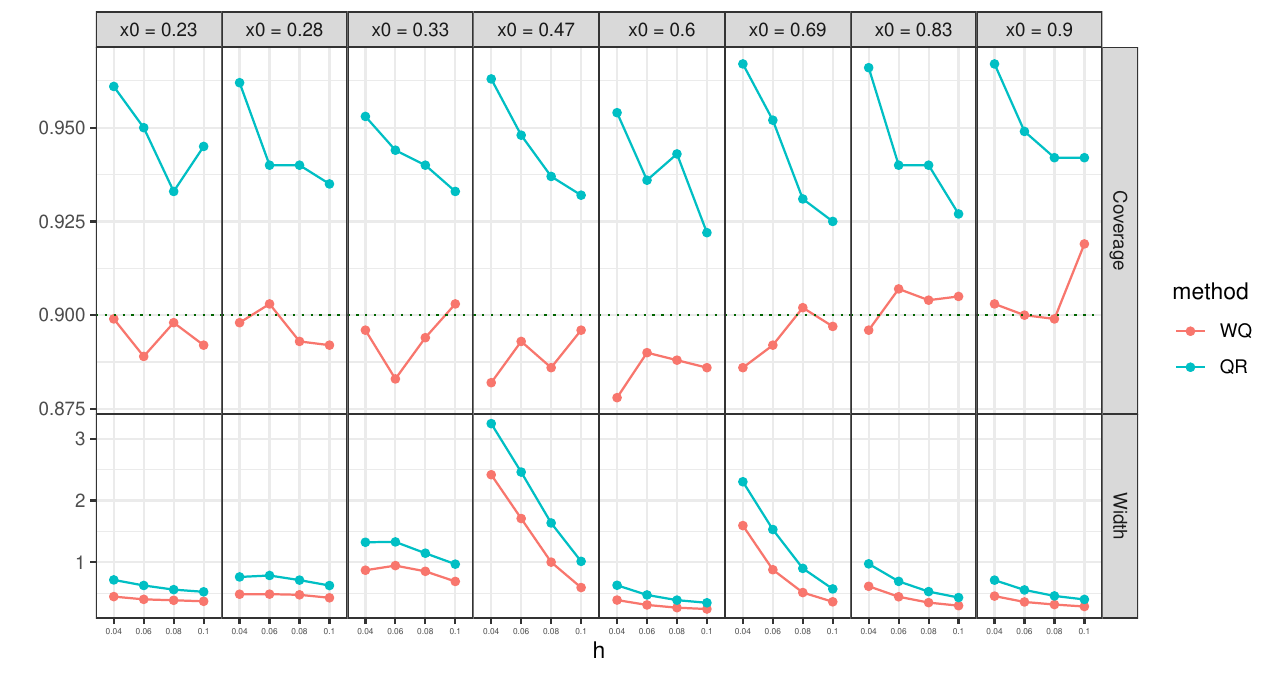}
    \caption{Coverage and width for the Spikes signal, setting 3.}
    \label{fig:spikes-s3-q0.5-tri}
\end{figure}

\begin{figure}[H]
    \centering
    \includegraphics[width=\textwidth]{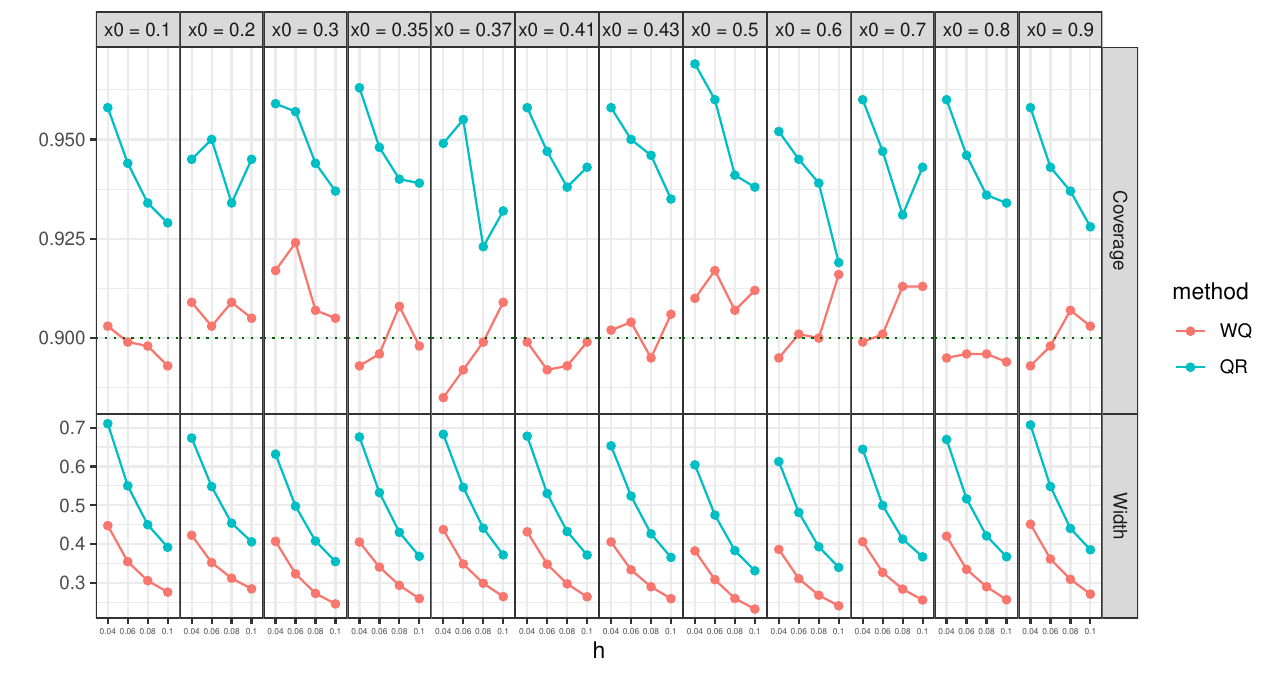}
    \caption{Coverage and width for the Parabolas signal, setting 3.}
    \label{fig:parabolas-s3-q0.5-tri}
\end{figure}

\begin{figure}[H]
    \centering
    \includegraphics[width=\textwidth]{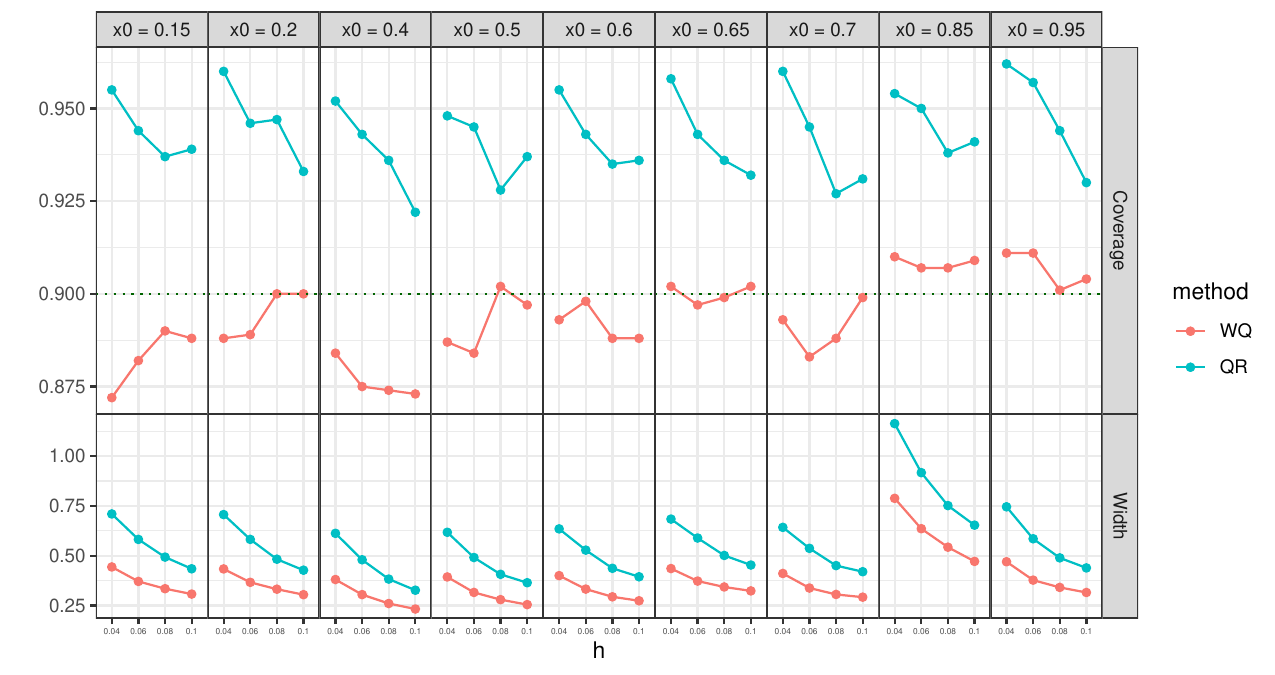}
    \caption{Coverage and width for the Angles signal, setting 3.}
    \label{fig:angles-s3-q0.5-tri}
\end{figure}

\subsection{Biweight kernel}
\label{subsec:biweight kernel}
In this section, we show the empirical width and coverage using the proposed method for $\theta_{0.5}$ with the biweight kernel. 

\begin{figure}[H]
    \centering
    \includegraphics[width=\textwidth]{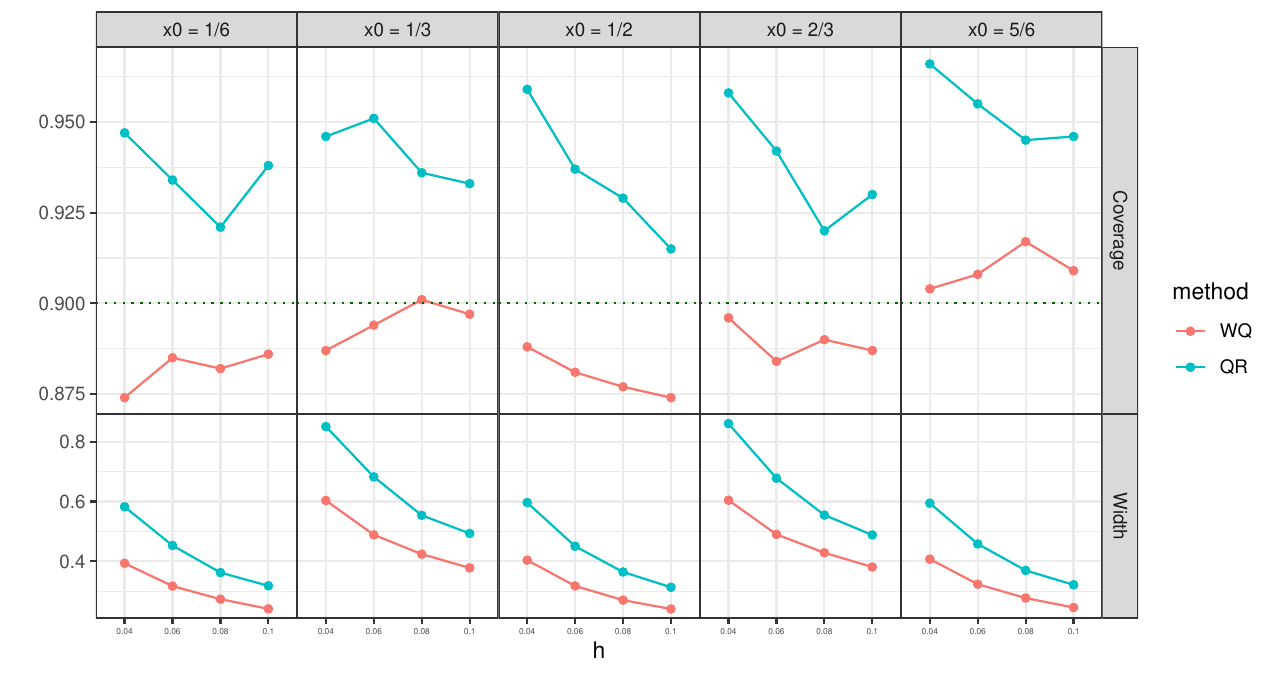}
    \caption{Coverage and width for the Step signal, setting 1.}
    \label{fig:step-s1-q0.5-biweight}
\end{figure}

\begin{figure}[H]
    \centering
    \includegraphics[width=\textwidth]{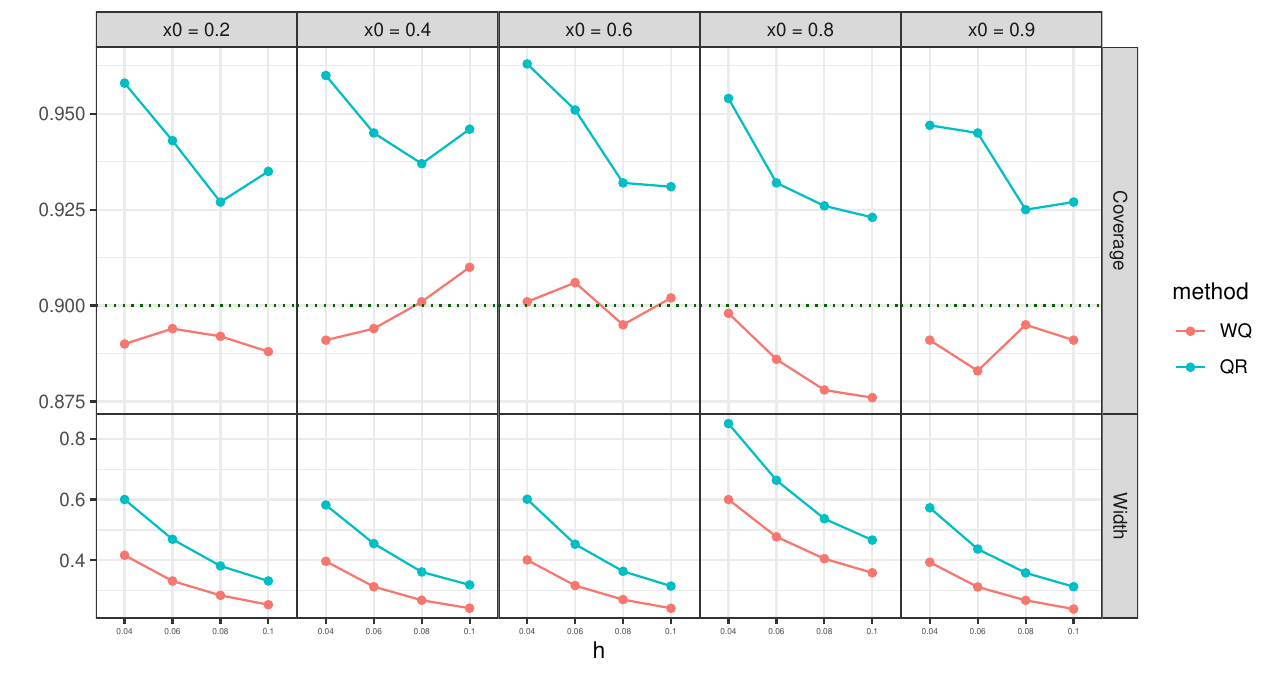}
    \caption{Coverage and width for the Blip signal, setting 1.}
    \label{fig:blip-s1-q0.5-biweight}
\end{figure}

\begin{figure}[H]
    \centering
    \includegraphics[width=\textwidth]{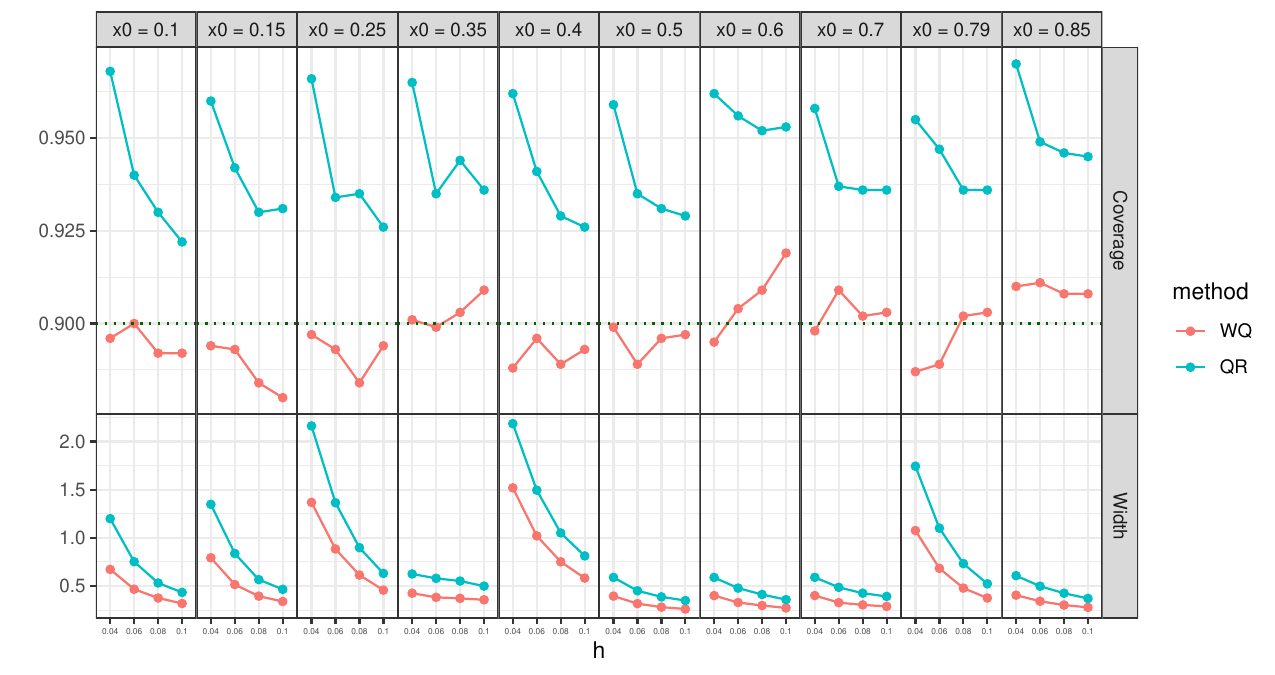}
    \caption{Coverage and width for the Bump signal, setting 1.}
    \label{fig:bump-s1-q0.5-biweight}
\end{figure}

\begin{figure}[H]
    \centering
    \includegraphics[width=\textwidth]{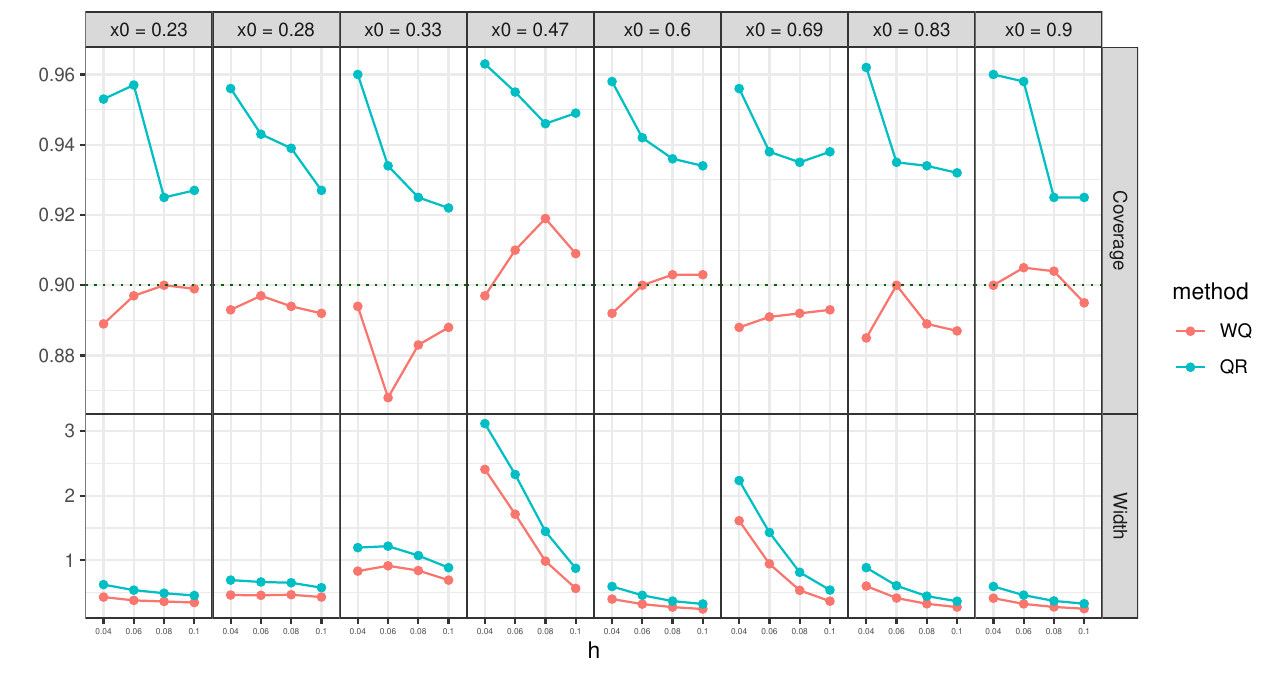}
    \caption{Coverage and width for the Spikes signal, setting 1.}
    \label{fig:spikes-s1-q0.5-biweight}
\end{figure}

\begin{figure}[H]
    \centering
    \includegraphics[width=\textwidth]{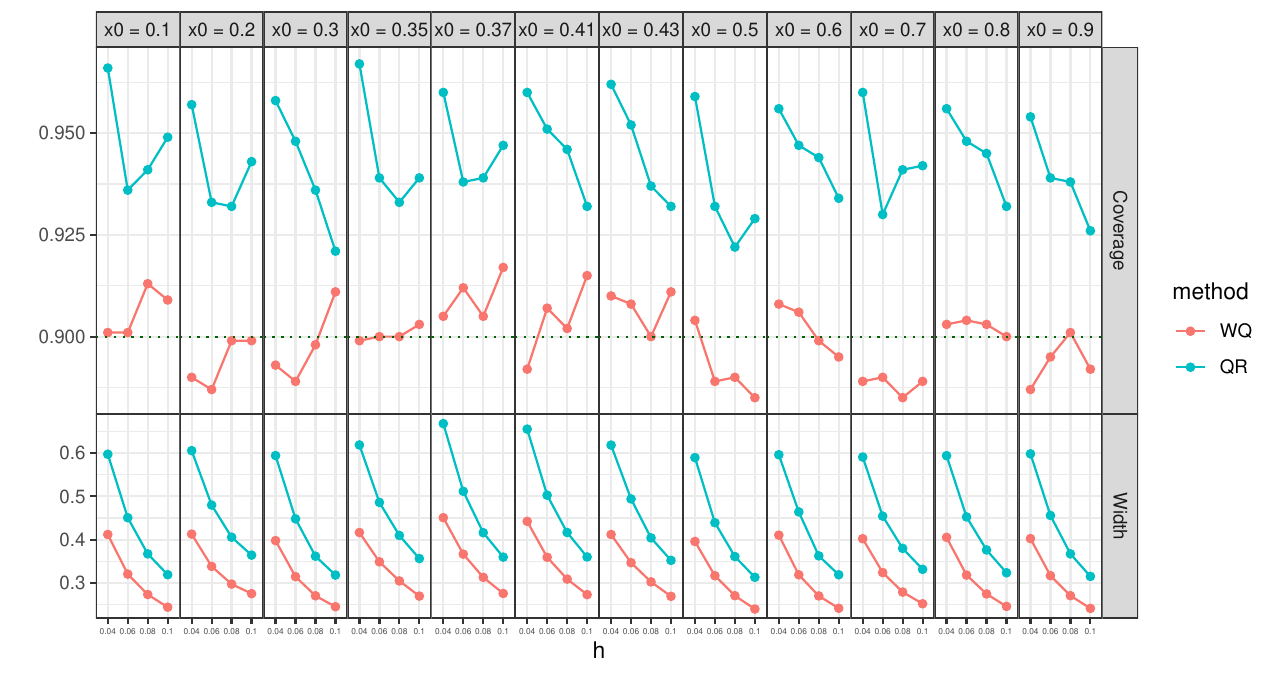}
    \caption{Coverage and width for the Parabolas signal, setting 1.}
    \label{fig:parabolas-s1-q0.5-biweight}
\end{figure}

\begin{figure}[H]
    \centering
    \includegraphics[width=\textwidth]{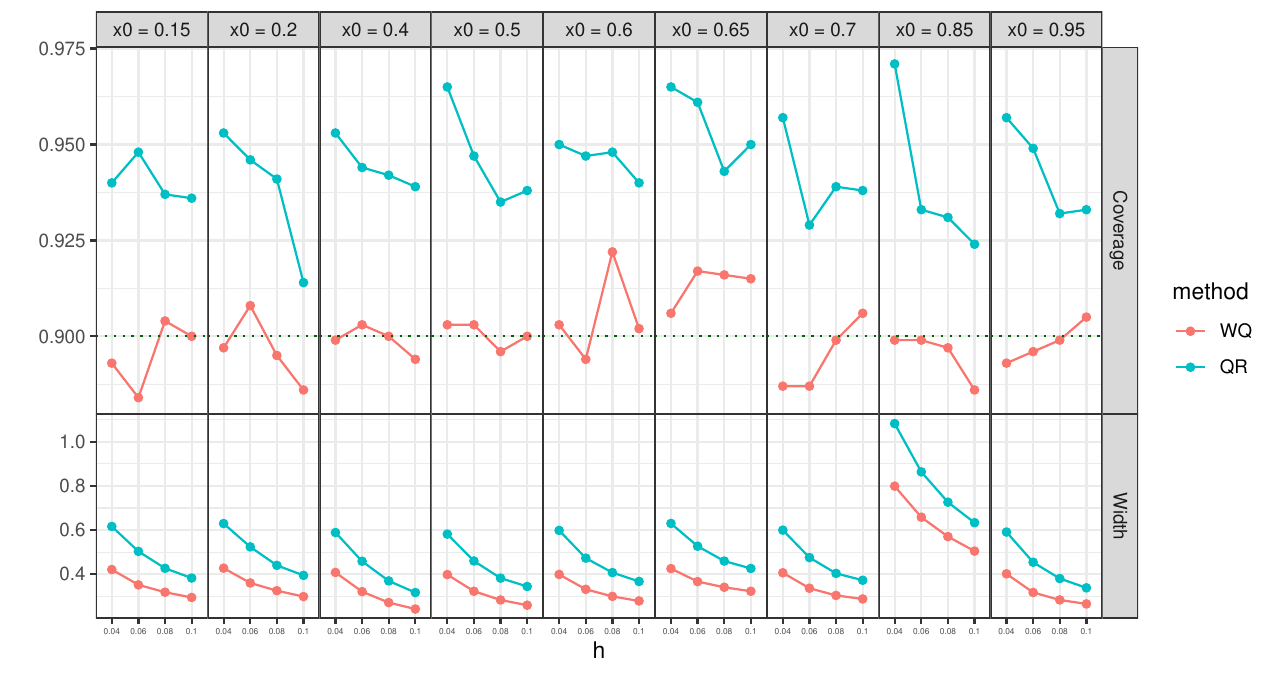}
    \caption{Coverage and width for the Angles signal, setting 1.}
    \label{fig:angles-s1-q0.5-biweight}
\end{figure}

	\bibliographystyle{apalike}
	\bibliography{conformal}
\end{document}